\documentclass[11pt,a4paper]{article}

\usepackage{SupportingFiles/jheppub}
\usepackage[utf8]{inputenc}
\usepackage[british]{babel}
\usepackage{multirow}
\usepackage[dvipsnames]{xcolor}
\usepackage{arydshln}
\usepackage{overpic}
\usepackage{etoolbox}
\usepackage{enumitem}
\usepackage{caption}
\usepackage{subcaption}
\usepackage{mathtools}
\usepackage{latexsym,exscale,amssymb,amsmath, amsthm}
\usepackage{sidecap}
\usepackage{braket}
\usepackage{tikz}
\usepackage{hyperref}
\usepackage{chngcntr}

\hypersetup{
    colorlinks,
    citecolor=red,
    filecolor=black,
    linkcolor=blue,
    urlcolor=blue
}
\counterwithin*{section}{part}
\numberwithin{equation}{section}
\numberwithin{table}{section}
\numberwithin{figure}{section}
\appto\appendix{\addtocontents{toc}{\protect\setcounter{tocdepth}{1}}}

\definecolor{highlightColour}{RGB}{0,0,150}
\definecolor{lightGrey}{RGB}{225,225,225}
\newcommand{\vect}[1]{\boldsymbol{#1}}

\newcommand{\R}{\mathbb{R}}

\newcommand{\Z}{\mathbb{Z}}
\newcommand{\N}{\mathcal{N}}
\newcommand{\A}{\mathcal{A}}
\newcommand{\C}{\mathbb{C}}
\newcommand{\X}{\mathcal{X}}

\newcommand{\pl}[1]{\left<#1\right>}

\newcommand{\g}{\vect{g}}

\newcommand{\cltrad}{\,\hat{\oplus}\,}
\newcommand{\trop}{\text{Tr}}
\newcommand{\G}[1]{\operatorname{Gr}({#1})}

\newcommand{\Conf}[1]{\widetilde{\operatorname{Gr}}({#1})}
\newcommand{\PConf}[1]{\widetilde{\operatorname{Gr}}_+({#1})}

\newcommand{\TPTC}[1]{\widetilde{\operatorname{Tr}}_+({#1})}
\newcommand{\pTPTC}[1]{\widetilde{\operatorname{pTr}}_+({#1})}

\newtheorem{theorem}{Theorem}[section]

\newtheorem{lemma}[theorem]{Lemma}
\theoremstyle{definition}

\newcommand{\be}{\begin{equation}}
\newcommand{\ee}{\end{equation}}
\allowdisplaybreaks

\title{Singularities of eight- and nine-particle amplitudes from cluster algebras and tropical geometry}
\author{Niklas Henke, Georgios Papathanasiou}
\affiliation{DESY Theory Group, DESY Hamburg, Notkestraße 85, D-22607 Hamburg}
\emailAdd{niklas.henke@desy.de}
\emailAdd{georgios.papathanasiou@desy.de}
\preprint{DESY 21-018}

\abstract{
	We further exploit the relation between tropical Grassmannians and $\operatorname{Gr}(4,n)$ cluster algebras in order to make and refine predictions for the singularities of scattering amplitudes in planar $\mathcal{N}=4$ super Yang-Mills theory at higher multiplicity $n\ge 8$. As a mathematical foundation that provides access to square-root symbol letters in principle for any $n$, we analyse infinite mutation sequences in cluster algebras with general coefficients. First specialising our analysis to the eight-particle amplitude, and comparing it with a recent, closely related approach based on scattering diagrams, we find that the only additional letters the latter provides are the two square roots associated to the four-mass box. In combination with a tropical rule for selecting a finite subset of variables of the infinite $\operatorname{Gr}(4,9)$ cluster algebra, we then apply our results to obtain a collection of  $3,078$ rational and $2,349$ square-root letters expected to appear in the nine-particle amplitude. In particular these contain the alphabet found in an explicit 2-loop NMHV symbol calculation at this multiplicity.
}

\begin{document}

\maketitle

\section{Introduction}
The study of scattering amplitudes in perturbative quantum field theories has led to a plethora of new insights, not only with regard to its direct application to phenomenology, but also in pure mathematics. Much of this progress did come from the analysis of amplitudes in $\N=4$ super Yang-Mills theory and its planar limit (pSYM). Being the simplest interacting four-dimensional gauge theory, it allows to recognize some of the underlying intricate mathematical structures more easily, which in some cases were also successfully transferred to theories more closely aligned to nature, see e.g.~\cite{Duhr:2012fh,Henn:2020omi}.

Explicit results as well as general arguments~\cite{ArkaniHamed:2012nw} (see however also~\cite{Brown:2020rda} for certain subtleties) suggest that the functions which describe the (appropriately normalised~\cite{Bern:2005iz, Alday:2009dv, Yang:2010as,Dixon:2014iba}) $\N = 4$ pSYM amplitudes in the maximally helicity violating (MHV) or next-to-MHV (NMHV) configurations are restricted to the class of \emph{multiple polylogarithms} (MPLs), a class of functions well-established in the mathematics literature~\cite{Gonch3,Gonch2,Brown:2011ik}. They can be represented as iterated integrals over rational integration kernels, or conversely be defined recursively with respect to their derivatives: An MPL $F_w$ of weight $w$ obeys
\begin{equation}
	dF_w = \sum_{\phi_\alpha} F^{\phi_\alpha}_{w-1} d\ln\phi_\alpha\,,
\end{equation}
with $F^{\phi_\alpha}_{w-1}$ being an MPL of weight $w-1$ and MPLs of weight $1$ being usual logarithms. The \emph{letters} $\phi_\alpha$ encode the branch-cut and singularity structure of the function $F_w$. Based on this definition, one recursively constructs the map $\mathcal{S}$ as
\begin{equation}
	\mathcal{S}\left[F_w\right] = \sum_{\phi_\alpha} \mathcal{S}\left[F^{\phi_\alpha}_{w-1}\right] \otimes \ln\phi_\alpha\,.
\end{equation}
It maps an MPL of weight $w$ to its \emph{symbol}~\cite{Goncharov:2010jf}, an $w$-fold tensor product of logarithms of the letters $\phi_\alpha$. The union of all the letters is called the \emph{symbol alphabet} and is the starting point of the \emph{cluster bootstrap}, see~\cite{Caron-Huot:2020bkp} for a recent review.

Whereas in theory the symbol alphabet and the entire amplitude can be computed via Feynman diagrams~\cite{DelDuca:2009au,DelDuca:2010zg}, this method becomes unwieldy very quickly with increasing loop number. Instead of directly computing the amplitudes, the cluster bootstrap attempts to first obtain the alphabet of the amplitude's symbol by some alternative way. Utilizing the observation that an $L$-loop (N)MHV amplitude is given by a weight $2L$ MPL, the space of all weight $2L$ symbols is then constructed from the alphabet. After fixing the amplitude's symbol from this space using consistency and physical constraints, the symbol can be integrated to obtain the actual function.

A key insight for this boostrap program is the observation that the letters of $n$-particle scattering are cluster $\A$-variables of the \emph{cluster algebra} associated to the Grassmannian $\G{4,n}$~\cite{Golden:2013xva}, following the emergence of these structures at the level of the integrand~\cite{ArkaniHamed:2009dn,ArkaniHamed:2012nw}. Due to the dual conformal symmetry of the theory~\cite{Drummond:2007au,Drummond:2006rz,Bern:2006ew,Bern:2007ct,Alday:2007he}, a certain quotient of the Grassmannian -- the \emph{configuration space} $\Conf{4,n}$ of $n$ points in complex projective space $\mathbb{P}^3$ -- corresponds to the space of kinematics of $n$-particle scattering, which can be conveniently described in terms of momentum twistor variables~\cite{Hodges:2009hk}. 

In cluster algebras, see section~\ref{sec:Background} or~\cite{1021.16017,1054.17024,CAIII,CAIV} for more details, the $\A$-variables are organized in overlapping sets, the \emph{clusters}, which are connected by an operation called \emph{mutation}. Starting from an initial cluster, the cluster algebra and thus all of its $\A$-variables are constructed by performing all possible mutations. In this way, the cluster algebra allows to obtain the amplitude's symbol alphabet and thus ultimately its symbol. 

This boostrap program has been successfully applied to the six-particle amplitude with up to seven loops~\cite{Dixon:2011pw,Dixon:2011nj,Dixon:2013eka,Dixon:2014voa,Dixon:2014xca,Dixon:2014iba,Dixon:2015iva,Caron-Huot:2016owq,Caron-Huot:2019vjl,Caron-Huot:2019bsq} (see also~\cite{Chestnov:2020ifg} for some more recent higher-loop results in its codimension-1 double-scaling limit) and for the seven-particle amplitude with up to four loops~\cite{Drummond:2014ffa,Dixon:2016nkn,Drummond:2018caf,Dixon:2020cnr}. For many years however, two major obstructions prevented expanding the program to higher multiplicity. First of all, the relevant cluster algebras become infinite for $n\geq 8$, that is they contain infinitely many variables. Whereas it is possible that with increasing loop numbers ever more relevant discontinuities of the amplitude and thus new letters appear, it is believed that the amplitude requires only a finite number of letters, in line with finite number of its Landau singularities~\cite{Prlina:2018ukf}, as obtained by the amplituhedron~\cite{Arkani-Hamed:2013jha,Arkani-Hamed:2017vfh}. Furthermore, cluster variables are always rational functions (in momentum twistor variables), whereas also square-root letters are required to describe all amplitudes, as is for example the case for the eight-particle two-loop NMHV amplitude~\cite{Zhang:2019vnm}.

Recently, in~\cite{Drummond:2019qjk, Drummond:2019cxm, Arkani-Hamed:2019rds, Henke:2019hve} it has been proposed that both of these obstructions can be overcome by considering the \emph{tropical} version of the configuration space $\Conf{k,n}$~\cite{Speyer2003,Speyer2005}, or equivalently its dual geometric object~\cite{Arkani-Hamed:2019mrd}. The relevance of tropical Grassmannians for scattering processes was first established in the context of tree-level amplitudes of generalized biadjoint scalar theory~\cite{Cachazo2019,Cachazo:2019apa}, see also~\cite{Sepulveda:2019vrz,Borges:2019csl,Early:2019zyi,Cachazo:2019ble,Cachazo:2019xjx,Early:2020hap,Cachazo:2020wgu} and references therein for recent progress in this direction, as well as~\cite{He:2020ray,Drummond:2020kqg,Lukowski:2020dpn,Arkani-Hamed:2020cig, Arkani-Hamed:2020tuz,Sturmfels:2020mpv,Parisi:2021oql} for work on further connections between (duals of) tropical Grassmannians, cluster algebras, and scattering amplitudes. In essence, the tropical version of the configuration space is obtained by replacing addition with the minimum and multiplication by addition in the polynomials parameterising (the totally positive) $\Conf{k,n}$. The resulting structure is a \emph{fan}, a collection of cones obtained as the positive span of the rays, half lines emanating from the origin.

We can also associate such a fan with the cluster algebra, with each $\A$-variable corresponding to a ray and each cluster to a cone. It turns out that for finite cluster algebras this \emph{cluster fan} triangulates the fan of the positive tropical configuration space~\cite{Speyer2005}, that is the former splits up the cones of the latter into simplicial cones. The cluster fan does so using the rays of the tropical fan as well as \textit{redundant rays} -- additional rays that are not required to describe the tropical fan. Moving to infinite cluster algebras, it is therefore natural to expect that the nature of their infinities can be interpreted as an infinite number of redundant triangulations. With the tropical fan being inherently finite, removing the redundant rays provides a \emph{tropical selection rule} that can be used to obtain a finite subset of $\A$-variables from the infinite cluster algebra, in other words to truncate it. This selection rule is also consistent with the case of finite cluster algebras, e.g. those describing seven particle scattering and below ($k=4$, $n\leq 7$), where it selects all $\A$-variables of the cluster algebra.

The rays of the cluster algebra truncated in this way are only a subset of the tropical rays for $n\geq 8$. However, one of the main ideas behind the aforementioned works~\cite{Henke:2019hve, Drummond:2019cxm, Arkani-Hamed:2019rds}, was that one may also access additional rays, and thus also the generalisations of $\A$-variables or letters associated with them, which turn out to contain square roots, by also considering limits of infinite mutation sequences starting from within the truncated cluster algebra. In particular these ideas were applied to the then first nontrivial case at $n=8$, where sequences of a rank-two affine or $\operatorname{A}^{(1)}_1$ subalgebra prove sufficient for obtaining all limit rays. While there exists a one-to-one mapping between limit rays and square-root letters, in~\cite{Drummond:2019cxm} it was additionally noticed that by assuming a one-to-many mapping that also takes the direction of approach to the ray into account in a certain way, then one astonishingly obtains precisely the 18 square-root letters found in the explicit expression for the two-loop NMHV eight-particle amplitude~\cite{Zhang:2019vnm}.\footnote{Note that in the literature the square-root letters are sometimes referred to as non-rational or algebraic, even though the rational part of the alphabet is of course algebraic as well.}

In this article, building on our previous work, we make an important step towards generalising these exciting developments to arbitrary multiplicity $n$. First, we analyse infinite mutation sequences of rank-two affine cluster algebras with general coefficients, which allow us to trivially obtain predictions for square-root letters for any such subalgebra of $\G{4,n}$. As a cross-check, we then apply our procedure to the known eight-particle case, not only finding agreement with the earlier analysis, but also comparing it to more recent predictions based on the closely related \emph{scattering diagram} approach~\cite{Herderschee:2021dez}. As essentially all proposals for $n$-particle alphabets to date correspond to different compactifications of the region of positive kinematics of $\Conf{4,n}$, concretely this approach amounts to a refinement of the tropical compactification, which at first sight seems to predict another 34 square-root letters on top of the two-loop NMHV ones. Very interestingly, we find that almost all of these naively square-root letters can be combined to yield rational ones that are already contained in the alphabets of~\cite{Henke:2019hve, Drummond:2019cxm, Arkani-Hamed:2019rds}: The only exceptions are the two cyclically inequivalent realisations of the four-mass box square-root, $\sqrt{\Delta_{i,i+2,i+4,i+6}}$ with $i=1,2$ and $i \sim i \mod 8$, formed by eight massless legs, see e.g.~\cite{Bourjaily:2013mma}.

Armed by this almost complete overlap between the two methods, we then move on to apply our results to the nine-particle amplitude. There are several good reasons to do so: First, because the associated cluster algebra $\G{4,9}$ is significantly ``more'' infinite than $\G{4,8}$~\cite{fomin2006cluster}, so it is not a priori certain that methods initially developed on the ground of the latter will have more general applicability. Second, because for $n=8$ there exist many subtleties in properly exploring the very interesting property of cluster adjacency~\cite{Drummond:2017ssj}, dictating how different symbol letters are allowed to sit next to each other in the symbol, see e.g.~\cite{Golden:2019kks}. Last but not least, because the field of amplitudes, and its impact to phenomenology, was shaped by carrying out initially very challenging computations, with the insights gained by the explicit results allowing their subsequent trivialisation.

Combining the tropical selection rule for rational letters with the infinite mutation sequence technology, we thus find a collection of 3,078 rational and 2,349 square-root letters expected to appear in the nine-particle amplitude, associated to 3,078 and 324 tropical rays, respectively. As a nontrivial check of our proposal, we confirm that it also contains the alphabet of the 2-loop NMHV nine-particle amplitude, whose symbol was recently computed in~\cite{He:2020vob}.

As perhaps hinted by the leap in computational complexity between the $\G{4,8}$ and $\G{4,9}$ cluster algebras, by comparing with other ways for obtaining the rays (which however provide no information on the letters associated to them), we notice that there also exist 27 rays of the minimal $\Conf{4,n}$ tropicalisation respecting the symmetries of the amplitude, which are not accessible by our procedure. We nevertheless find it very intriguing that we only fall short by such a small margin. While understanding what kind of generalisations of cluster variables could be associated to these rays, and whether they are relevant for amplitude singularities,\footnote{For example, it is not clear if the missing rays are ``just'' associated to more intricate algebraic letters beyond square roots, or point towards the need for significantly more complicated, elliptic generalisations of MPLs starting to contribute at $n=9$.} are open questions we leave for future work, here we also touch on one possibility towards addressing them. In particular, we consider more general infinite mutation sequences of higher rank ($\operatorname{A}^{(1)}_m$) algebras, and present some preliminary evidence that these may not be accessed by any type of infinite mutation sequence starting from within the cluster algebra. As the scattering diagram approach also relies on these sequences, this therefore seems to suggest that a radically different idea might be needed to tackle these exciting questions, and calls for the generation of explicit new amplitude data that will reveal it to us.

The plan of the rest of this article is as follows. In section~\ref{sec:Background}, we first briefly review some basic notions of the totally positive tropical Grassmannian, as well as the closely related, partial tropicalisation of the configuration space $\Conf{4,n}$ that will be relevant for scattering amplitudes. We also review (Grassmannian) cluster algebras, focusing especially on the formalism of coefficients, which will be advantageous for our purposes. In section~\ref{sec:A11} we present the mathematical foundation of our analysis, the general solution of infinite mutation sequences in the affine rank-$2$ cluster algebra of type $\operatorname{A}^{(1)}_1$. We then apply these results to reobtain the eight-particle alphabet as a check, and also compare with the more recent scattering diagram approach. Section~\ref{sec:NineParticle} is devoted to our main application, new predictions for the letters of the nine-particle amplitude. Section~\ref{sec:Am1} discusses higher-rank generalisations of infinite mutation sequences as well as their inherent limitations, and finally section~\ref{sec:Conclusion} contains our conclusions and outlook.

Results similar to those presented in this article were independently obtained in~\cite{Ren:2021ztg}.

\section{Tropical fans and cluster algebras}
\label{sec:Background}
In this section we introduce the mathematical concepts utilized throughout the article. We begin by reviewing the space of kinematics of $\N = 4$ pSYM, the configuration space $\Conf{4,n}$ of $n$ points in complex projective space $\mathbb{P}^3$, which is constructed as the quotient of the Grassmannian $\G{4,n}$ over the complex torus. Following this, we discuss how the configuration space is tropicalised and review the associated fan structures. Furthermore, we briefly review cluster algebras with coefficients -- a framework required for the analysis of the infinite mutation sequences -- and relate the cluster fan to the tropical fans, allowing us to review the selection rule that will be used in section~\ref{sec:NineParticle} to obtain a finite alphabet from the infinite cluster algebra.

\subsection{Grassmannians and configuration spaces}
\label{sec:BackgroundGrConf}
The \emph{Grassmannian} $\G{k,n}$ can be defined as the space of $k$-dimensional planes through the origin in an $n$-dimensional vector space. Since each of these planes is spanned by $k$ $n$-vectors, $\G{k,n}$ can be realized as $k \times n$ matrices modulo the $GL(k)$ transformations corresponding to a change of basis. The minors of this matrix are the \emph{Plücker variables} $\pl{i_1\dots i_k}$ for $i_j = 1,\dots, n$ which satisfy the Plücker relations
\begin{equation}
	\label{eq:plueckerRelations}
	\pl{i_1\dots i_r [i_{r+1} \dots i_k}\pl{j_1 \dots j_{r+1}] j_{r+2} \dots j_k} = 0\,,
\end{equation}
whereas the square brackets denote antisymmetrisation over the included indices. These relations may also be used as an alternative starting point to contsruct the Grassmannian. Starting with the ring of integer coefficient polynomials in the
\begin{equation}
D=\binom{n}{k}
\end{equation}
Plücker variables, we can identify $\G{k,n}$ with the set of solutions to the Plücker relations, eq.~\eqref{eq:plueckerRelations}, quotiened by the global scaling $\pl{i_1\dots i_k} \rightarrow t\cdot\pl{i_1\dots i_k}$ for $t\in\C\setminus\{0\}$, which leaves the Plücker relations invariant.

Further to this scaling, the Plücker relations are also invariant under the local scaling $\pl{i_1\dots i_k} \rightarrow t_{i_1}\cdots t_{i_k}\cdot\pl{i_1\dots i_k}$ for $t_{i_1},\dots, t_{i_k}\in\C\setminus\{0\}$. If we also quotient by this transformation, we obtain the \emph{configuration space} $\Conf{k,n}$, which for $k=4$ is the space of kinematics of $n$-particle $\N = 4$ pSYM amplitudes considered in this article, expressed in terms of momentum twistors $Z_{i_1},\ldots, Z_{i_n}$ \cite{Hodges:2009hk}. While the Grassmannian has dimension $k(n-k)$, the configuration space has dimension 
\begin{equation}
	d = (k-1)(n-k-1)\,. 
\end{equation}

By further also restricting all ordered Plücker variables to be positive, we obtain the positive configuration space $\PConf{k,n}$. This space can be parameterised in terms of the so-called web-parameterisation, see~\cite{Speyer2005} or the appendix of the authors' previous article~\cite{Henke:2019hve}. In this parameterisation the Plücker variables are polynomials in the $d$ web-variables $x_1,\dots,x_d$, like for example for $\pl{25}$ of $\PConf{2,5}$, whose web-parameterisation is given in terms of the two independent coordinates $x_1, x_2$ by
\begin{equation}
	\label{eq:webParamEx}
	\pl{25} = 1+x_1+x_1x_2\,.
\end{equation}

\subsection{Partially tropicalised configuration space $\pTPTC{4,n}$}
\label{sec:BackgroundTropConf}
\textit{Tropical geometry} is essentially algebraic geometry over the tropical semifield -- the real numbers with taking the minimum as \emph{tropical addition} $\oplus$, and addition as \emph{tropical multiplication} $\otimes$, see e.g. the reviews~\cite{Mikhalkin2006,Maclagan2012,Brugalle2015}. In practice, at least concerning the application of tropical geometry in this article, this means that we start with a geometric object that is described in terms of polynomials and replace addition with the minimum and multiplication with addition.\footnote{In the mathematically precise formulation of tropical geometry, see e.g.~\cite{Speyer2003}, the starting point is a variety attached to a polynomial ideal, whose tropical variety is then constructed. Since we only need the tropical version of the positive configuration space, we will make use of its web-parameterisation and review only the neccessary mathematics following \cite{Speyer2005}.}

Continuing the example of $\PConf{2,5}$, we start with the web-parameterisation of the Plücker variables and tropicalise the parameterisation polynomials, e.g.
\begin{equation}
	\label{eq:tropPolEx}
	\pl{25} = 1+x_1+x_1x_2 \quad \rightarrow \quad \trop \left(\pl{25}\right) = \min\left(0, x_1, x_1 + x_2\right)\,.
\end{equation}
Note that by construction the numerical coefficients of the polynomials are mapped to zero. These so-called \textit{tropical polynomials} are piecewise linear functions whose domains of linearity are cut out by the \textit{tropical hypersurfaces}. They can be computed by setting two of the terms in the minimum equal and smaller or equal than the remaining terms, e.g. for eq.~\eqref{eq:tropPolEx} we have
\begin{equation}
	\label{eq:tropHSEx}
	0 = x_1 \leq x_1 + x_2\,,\quad	0 = x_1 + x_2 \leq x_1\,,\quad	x_1 = x_1 + x_2 \leq 0\,.
\end{equation}
The regions cut out by these tropical hypersurfaces are actually \textit{convex cones} -- subsets of $\R^d$ closed under linear combination with positive coefficients. The regions where $d-1$ linearly independent hypersurfaces intersect are called \textit{rays} -- lines emanating from the origin -- which span the convex cones. Together, the rays and cones form a \textit{fan}, the main object of interest in this paper. For the example of $\pl{25}$ of $\PConf{2,5}$, we obtain three tropical hypersurfaces in $\R^2$ from eq.~\eqref{eq:tropHSEx}, which, due to the low dimension, are identical to the rays. These three rays span three different cones. The fan is illustrated in fig.~\ref{fig:fanExample}.
\begin{figure}[ht]
	\centering
	\includegraphics[width=0.35\textwidth]{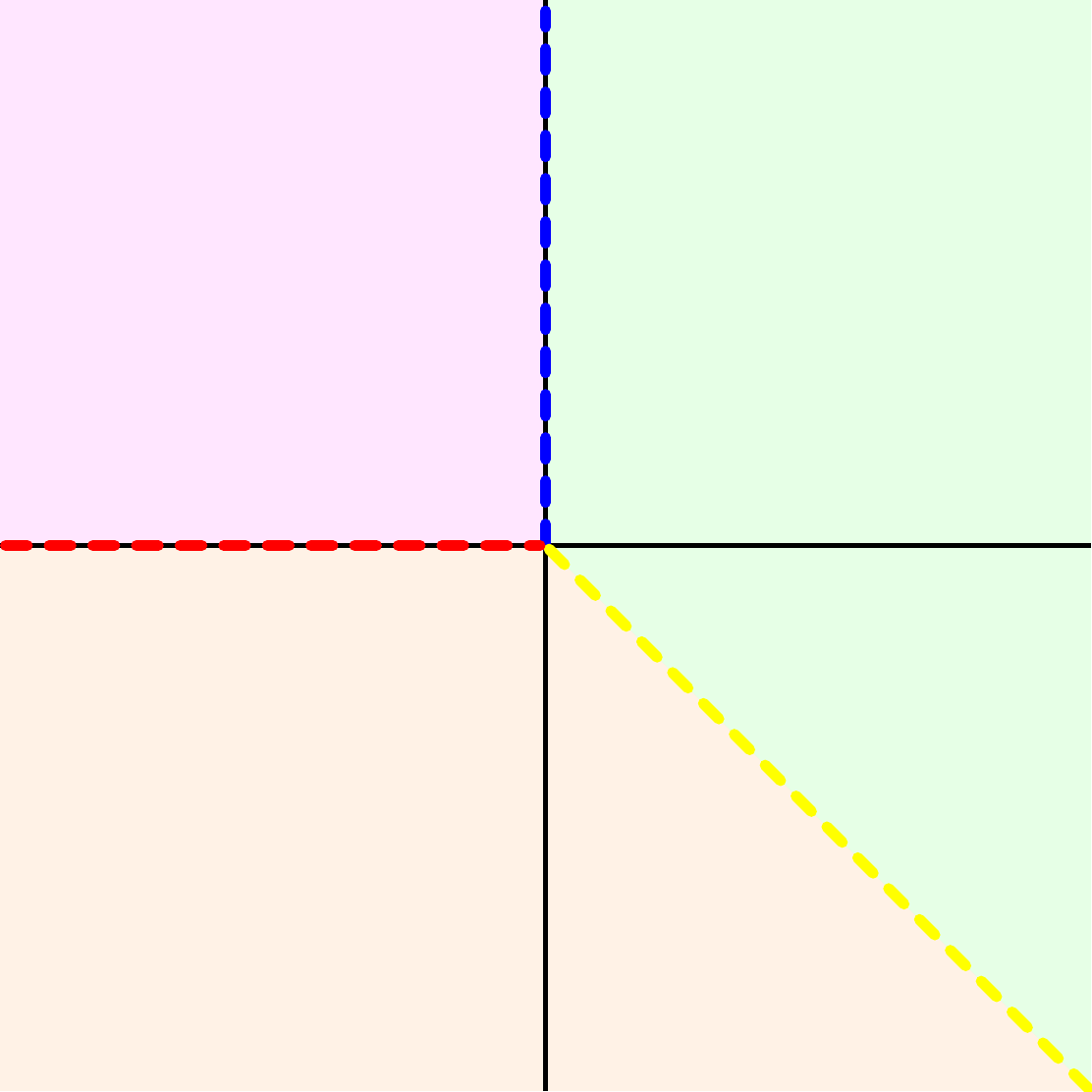}
	\caption{Fan associated to the tropical polynomial $\trop \left(\pl{25}\right) =  \min\left(0, x_1, x_1 + x_2\right)$. The rays are depicted as dashed lines in blue, red and yellow, respectively. The cones are depicted in the composite color of the two rays by which they are spanned.}
	\label{fig:fanExample}
\end{figure}

To construct the entire positive tropical configuration space $\TPTC{k,n}$, we tropicalise all paramterised Plücker variables and obtain their tropical hypersurfaces resulting in a fan in $\R^d$ for each of the Plücker variables. The fan $F_{k,n}$ of $\TPTC{k,n}$ is then given as the common refinement -- essentially the union -- of the individual fans.\footnote{Note that while the tropical version of a variety is actually obtained as the intersection of all tropical hypersurfaces, we have to take the common refinement here due to working with a parameterisation of the variety in question, see~\cite{Speyer2005}.} Note that the resulting fan is not just the union of rays and cones of the individual fans since the tropical hypersurfaces of one fan might cut a cone of another fan into several cones. For more details, see also~\cite{Henke:2019hve}. By using the tropicalised Plücker parameterisation, eq.~\eqref{eq:tropPolEx}, we can alternatively obtain an embedding of the $d$-dimensional fan in $\R^D$.

Whereas this construction is the canonical way to tropicalise the positive configuration space, we may choose to only tropicalise a subset of all Plücker variables. The resulting fan is the common refinement of the corresponding subset of fans associated to the Plücker variables and thus a coarser version of the fully tropicalised fan. This means that the partial fan consists of less rays and cones with some of the cones of the full refining those of the partial fan. In this paper we will be almost exclusively be focusing on the following \emph{partial tropicalisation} of the configuration space $\PConf{4,n}$,
\be
\pTPTC{4,n}: \begin{array}{c}
\pl{ii+1jj+1}\to \trop \left(\pl{ii+1jj+1}\right)\,,\\ \pl{ij-1jj+1}\to \trop \left(\pl{ij-1jj+1}\right)\,,
\end{array}\quad i=1,\ldots,n\,,
\ee
namely we only tropicalise the Plücker variables with indices either pairwise adjacent, or forming an adjacent triplet. The associated fan will be denoted by $pF_{4,n}$. This choice is believed to be the most relevant for $n$-particle amplitudes in $\N = 4$ pSYM, as it leads to predictions for their singularities that agree with the known $n=6,7$ cases, and more generally respects the parity symmetry of MHV amplitudes in a minimal way \cite{Drummond:2019cxm,Arkani-Hamed:2019rds,Drummond:2020kqg}. In contrast, $\TPTC{4,n}$ is not parity invariant.

\subsection{Cluster algebras with coefficients}
\label{sec:BackgroundCACoefs}
Another remarkable property of the Grassmannian $\G{k,n}$ is that its coordinate ring carries the structure of a \emph{cluster algebra}. A rank $r$ cluster algebra consists of the so-called $\A$-variables -- rational functions in $r$ arguments -- organized in overlapping sets of $r$ variables, the \emph{clusters}, which are connected to each other by a birational transformation, the \emph{mutation}. Finally, for each cluster we have the adjacency matrix $B$, a $r\times r$ antisymmetrisable matrix encoding the connection among the variables within the cluster. If the adjacency matrix is antisymmetric, we can equivalently represent the cluster by a quiver, where nodes correspond to cluster variables, and the absolute value and sign of the entries of $B$ corresponds to the number of arrows between nodes and their direction, respectively. 

If we can generate only a finite number of distinct clusters and cluster variables with mutation, the cluster algebra is called finite (and otherwise infinite). Finite cluster algebras are completely classified in terms of the Cartan-Killing classification of semisimple Lie algebras. In practice, this means that if a cluster algebra has a cluster whose quiver is equivalent to a Dynkin diagram, the cluster algebra is of the corresponding type. For more details, see e.g.~\cite{1021.16017,1054.17024,CAIII,CAIV}.

In the physics literature~\cite{Golden:2013xva}, it is common to also consider additional \emph{frozen variables}, associated to \emph{frozen nodes} in the quiver. The frozen variables behave like the $\A$-variables except that they are never mutated and that there are no arrows between them. The $\A$-variables and frozen variables together are referred to as \emph{$\A$-coordinates}. If we consider $M$ frozen nodes, the adjacency matrix is extended to a $(r+M)\times r$ matrix with the $r\times r$ part encoding the connections between the $r$ $\A$-variables being the \emph{principal part}. As an example, the quiver of the initial cluster of $\G{2,5}$ is  depicted in figure~\ref{fig:clusterAlgebraExample}. Note that the $\A$-variables are sometimes referred to as \emph{unfrozen variables}.
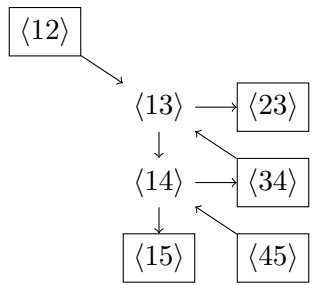
\begin{figure}[ht]
	\centering
	\begin{tikzpicture}[scale=1.0]
		\node at (0,0) (a1) {$\pl{13}$};
		\node at (1.5,0) [draw, rectangle] (a4) {$\pl{23}$};
		
		\node at (0,-1) (b1) {$\pl{14}$};
		\node at (1.5,-1) [draw, rectangle] (b4) {$\pl{34}$};
				
		\node at (0,-2) [draw, rectangle] (e1) {$\pl{15}$};
		\node at (1.5,-2) [draw, rectangle] (e4) {$\pl{45}$};
		
		\node at (-1.5,1) [draw, rectangle] (f) {$\pl{12}$};
		
		\draw[->] (f) -- (a1);
		
		\draw[->] (a1) -- (a4);
		\draw[->] (b1) -- (b4);		

		\draw[->] (a1) -- (b1);
		\draw[->] (b1) -- (e1);
		
		\draw[->] (b4) -- (a1);
		\draw[->] (e4) -- (b1);		
	\end{tikzpicture}
	\caption{Initial seed of the rank-2 cluster algebra of $\G{2,5}$. Each node corresponds to an $\A$-coordinate. The nodes surrounded by a box are associated to frozen variables, whereas the unboxed nodes are associated to $\A$-variables.}
	\label{fig:clusterAlgebraExample}
\end{figure}

In the mathematics literature, there is another equivalent description using the so-called \emph{coefficients}~\cite{CAIV}. As we will detail momentarily, in this formalism all frozen variables connected to a given $\A$-variable are grouped into a single coefficient, which also changes under mutation and is associated to the $\A$-variable. The main advantage of cluster algebras with general coefficients is that they can be constructed once, and then specialized to \emph{any} particular choice of frozen variables at the very end. This allows for a unified treatment of what would be several distinct computations in the language of frozen variables, and it will be crucial in obtaining the singularities of eight- and nine-particle amplitudes, as described in sections~\ref{sec:A11} and~\ref{sec:NineParticle}, in essentially one go.

We now consider a rank-$r$ cluster algebra with $r$ $\A$-variables $a_1,\dots,a_r$ and $M$ frozen variables $a_{r+1},\dots,a_{r+M}$ in the initial cluster. Denoting the components of the adjacency matrix of the initial cluster -- a $(r + M)\times r$ integer matrix -- by $b^0_{ij}$, we attach a \emph{coefficient} $y_i$ to each $\A$-variable $a_i$ in the initial cluster via\footnote{Since by construction there are no edges between frozen variables, the $(r+M)\times r$ \emph{extended adjacency matrix} is sufficient to describe all adjacencies in the quiver. In the framework of cluster algebras with coefficients, we instead only consider the $r\times r$ \emph{principal part} of that matrix and may use eq.~\eqref{eq:coeffs} as the definition of the extended adjacency matrix in any cluster.}
\begin{equation}
	\label{eq:coeffs}
	y_i = \prod_{j=r+1}^{r+M} a_j^{b^0_{ji}}\,.
\end{equation}
The special case that a cluster algebra of rank $r$ has $r$ frozen variables such that $y_i = a_{r+i}$ is referred to as \emph{principal coefficients}. As can already be seen from eq.~\eqref{eq:coeffs}, the coefficients are closely related to the $\X$-variables of Fock and Goncharov \cite{FG03b}, $x_i$ with $i=1,\dots,r$. The latter are defined, in any cluster of the cluster algebra, by
\begin{equation}
	\label{eq:XVars}
	x_i = \prod_{j=1}^{r} a_j^{b_{ji}} \cdot y_i\,.
\end{equation}

When mutating the cluster at one node, we obtain another cluster with mutated variables and coefficients. Hence, as is usually done in the formalism of cluster algebras with coefficients, we label clusters by an index $t$ and denote the variables of such cluster by $a_{i;t}$, whereas the index $i$ labels the position of the variable within the cluster. It is common to label the initial cluster by $t=0$ or to just drop the index, if it is clear from context. Consider now the mutation of cluster $t$ at node $j$ resulting in the cluster $t'$. The mutation rule for the $\A$-variables is given by
\begin{equation}
	\label{eq:AMutationRule}
	a_{j;t'} = \frac{y_{j;t}\prod_{i=1}^ra_{i;t}^{\left[b^t_{ij}\right]_+}+\prod_{i=1}^ra_{i;t}^{\left[-b^t_{ij}\right]_+}}{\left(1\cltrad y_{j;t}\right)a_{j;t}}\,,
\end{equation}
where $b^t_{ij}$ are the components of the adjacency matrix of cluster $t$, $\left[x\right]_+ = \max\left(0,x\right)$,  and with all other $\A$-variables remaining unchanged. In this formula, we have used the \emph{cluster-tropical addition}\footnote{Note that although closely related, cluster-tropical and tropical addition are not quite the same and hence denoted by $\cltrad$ and $\oplus$, respectively. Essentially, cluster-tropical addition on the monomials of frozen variables is given by tropical addition on the exponents of these variables.} $\cltrad$, which is defined on the frozen variables as
\begin{equation}
\label{eq:cltrad}
\prod_{i=r+1}^{r+M} a_i^{c_i} \cltrad \prod_{i=r+1}^{r+M} a_i^{d_i} = \prod_{i=r+1}^{r+M} a_i^{\min\left(c_i,d_i\right)}\,.
\end{equation}
Similar to the $\A$-variables, we also have a mutation rule for the coefficients, which is given by
\begin{equation}
	\label{eq:YMutationRule}
	y_{l;t'} = 
	\begin{cases}
		y_{j;t}^{-1}\quad&\text{if}\,\,l=j\,,\\
		y_{l;t}y_{j;t}^{\left[b^t_{jl}\right]_+}\left(1\cltrad y_{j;t}\right)^{-b^t_{jl}}\quad&\text{if}\,\,l\neq j\,.
	\end{cases}
\end{equation}
Since this mutation relation implies that the coefficients are always monomials in the frozen variables, which are the same in all clusters, cluster-tropical addition on the coefficients as given by eq.~\eqref{eq:cltrad} is well defined for all clusters. Note that the mutation rule for the $\X$-variables is the same as that of the coefficients except with normal addition instead of the cluster-tropical addition. The mutation rule for the adjacency matrix remains unchanged in comparison to the more familiar framework of frozen and unfrozen variables, and is given by
\begin{equation}
	\label{eq:BMutation}
	b^{t'}_{il} = \begin{cases}
		- b^t_{il}\quad&\text{if}\,\,i=j\,\,\text{or}\,\,l=j\,,\\
		b^t_{il}+\operatorname{sign}(b^t_{ij})[b^t_{ij}b^t_{jl}]_+\quad&\text{otherwise.}
	\end{cases}
\end{equation}
All mutation rules presented here are equivalent to those more commonly used in the physics literature, as can be easily seen by inserting the definition of the coefficients in terms of the frozen variables. 

Another advantage of considering cluster algebras with coefficients instead of frozen variables is that it makes the \emph{separation principle} manifest. Using eq.~\eqref{eq:XVars} to express $y_{j;t}$ in terms of the $\A$- and $\X$-variables of the cluster $t$, we can rewrite the mutation rule~\eqref{eq:AMutationRule} as
\begin{equation}
	\label{eq:AMutationRuleX}
	a_{j;t'} = \left(a_{j;t}\right)^{-1}\prod_{i=1}^ra_{i;t}^{\left[-b^t_{ij}\right]_+}\cdot\frac{1+x_{j;t}}{1\cltrad y_{j;t}}\,.
\end{equation}
The consequence of this factored form of the mutation relation is that, in the cases relevant to this article, any $\A$-variable can be written in such a way: a monomial in the initial $\A$-variables times some rational function in the initial $\X$-variables divided by its cluster-tropical version, that is we have
\begin{equation}
	\label{eq:AVarRayForm}
	a = \prod_{i=1}^r a_{i;0}^{g_i} \cdot \frac{F\left(x_{1;0},\dots,x_{r;0}\right)}{F_{\mathbb{T}}\left(y_{1;0},\dots,y_{r;0}\right)}\,,
\end{equation}
for some $\A$-variable $a$ and whereas $F_{\mathbb{T}}$ denotes the function obtained by replacing addition with cluster-tropical addition in the rational function $F$. In this way, we can associate a unique $\g$-vector, an integer vector in $\Z^r$ whose components are the $g_i$, to each $\A$-variable, for which we also can obtain a mutation rule from eq.~\eqref{eq:AMutationRule}, see e.g.~\cite{CAIV}.

However, for our purposes it is better to work with a modified version thereof. In order to more closely align the rays associated to the $\A$-variables to the rays of the totally positive tropical configuration space, we use a modified mutation rule to compute these \emph{cluster rays}, see also~\cite{Drummond:2019qjk,Henke:2019hve}. To construct this relation, we first attach a coefficient matrix $C$ to each cluster. In the initial cluster, it is given by $C_0 = \mathbf{1}_r$, whereas $\mathbf{1}_r$ is the $r\times r$ identity matrix. The mutation rule is given by
\begin{equation}\label{eq:CMutation}
	c^{t'}_{il} = \begin{cases}
		- c^t_{il}\quad&\text{if}\,\,i=j\,\,\text{or}\,\,l=j\,,\\
		c^t_{il} - [c^t_{ij}]_+b^t_{jl} + c^t_{ij}[b^t_{jl}]_+\quad&\text{otherwise.}
	\end{cases}
\end{equation}
Finally, we introduce the ray matrix $G$ for each cluster, whose columns are the cluster rays. By construction, in the initial cluster it is given by $G_0 = \mathbf{1}_r$. The mutation rule is given by
\begin{equation}
	\label{eq:rayMutation}
	g^{t'}_{il} = \begin{cases}
		g^t_{il}\quad&\text{if}\,\,i=j\,\,\text{or}\,\,l=j\,,\\
		-g^t_{il} + \sum_{m=1}^r\left(g^t_{im}[-b^t_{mj}]_+ + b^0_{im}[c_{mj}]_+\right)\quad&\text{otherwise.}
	\end{cases}
\end{equation}
The cluster fan, consisting of the cones spanned by the cluster rays in each cluster, is combinatorially equivalent the one obtained from the actual $\g$-vectors.

\subsection{Fans of cluster algebras}
\label{sec:BackgroundCAFan}
Besides the tropical fans reviewed above, we may also associate a fan to any cluster algebra. First, we construct the \emph{cluster polytope}, which is closely related to the \emph{exchange graph} of the algebra,\footnote{The exchange graph of the cluster algebra is the 1-skeleton of the cluster polytope.} by associating each cluster of a rank-$r$ cluster algebra to a vertex. If two clusters are related by mutating one of their variables, they are connected by a line. This $1$-dimensional face of the polytope may be alternatively described by fixing the $r-1$ $\A$-variables that are unchanged in the mutation and that are thus shared between the two clusters. Similarly, we obtain a $l$-dimensional face of the polytope by fixing $r-l$ variables that appear in a cluster together. The face is then bordered by all vertices that contain all these $r-l$ variables. 

The fan of the cluster algebra is taken to be the normal fan of this polytope, that is its rays are the inward-pointing normals to the codimension-$1$ faces of the polytope. In this way, a $l$-dimensional face of the polytope becomes a $(r-l)$ dimensional or codimension-$l$ face of the fan. For example, the $0$-dimensional vertices of the polytope correspond to the $r$-dimensional cones and the $r-1$ dimensional faces associated to each $\A$-variable correspond to the rays of the fan. We thus have a one-to-one association of $\A$-variables and rays of the cluster fan, see also table~\ref{tab:clusterFanRelations}. These rays are closely related to the $\g$-vectors associated to the variable and are obtained by a mutation rule, as sketched in the previous section.
\begin{table}[ht]
	\centering
	\begin{tabular}{|c|c|c|c|c|}
		\hline \hline
		\multirow{2}{*}{Algebra} & \multicolumn{2}{|c|}{Polytope} & \multicolumn{2}{|c|}{Fan} \\ \cline{2-5}
		& Dim. & Type & Dim. & Type \\ \hline \hline
		Cluster & $0$ & Vertex & $d$ & Cone  \\ \hline
		Mutation & $1$ & Line & $d-1$ & Facet \\ \hline
		\multicolumn{1}{;{3pt/2pt}c;{3pt/2pt}}{$\vdots$} & \multicolumn{2}{;{3pt/2pt}c;{3pt/2pt}}{$\vdots$} & \multicolumn{2}{;{3pt/2pt}c;{3pt/2pt}}{$\vdots$} \\ \hline 
		$\A$-variable & $d-1$ & Facet & $1$ & Ray \\ \hline 
		\hline
	\end{tabular}
	\caption{Comparison of the faces of a cluster algebra of rank $d$, its polytope and the cluster fan.}
	\label{tab:clusterFanRelations}
\end{table}

Remarkably, as first observed in~\cite{Speyer2005}, for finite cluster algebras of $\G{k,n}$ the cluster fan is a refinement of the fan of the totally positive tropical configuration space $\TPTC{k,n}$. That is, the cones of the cluster fan are all contained within cones of the tropical fan. Since the cones of the cluster fan are all simplicial, it triangulates $\TPTC{k,n}$. 

In this triangulation, however, the cluster algebra sometimes introduces \emph{redundant rays} -- rays of the cluster fan that are not tropical rays. Geometrically, redundant rays are not on a $1$-dimensional but a higher dimensional intersection of tropical hypersurfaces. Instead they are the positive linear combination of two tropical rays spanning some cone. This is illustrated on the left hand side of figure~\ref{fig:infiniteRedundantTriangulation}
\begin{figure}[ht]
	\centering
	\begin{subfigure}[b]{0.3\textwidth}
		\includegraphics[width=1\textwidth]{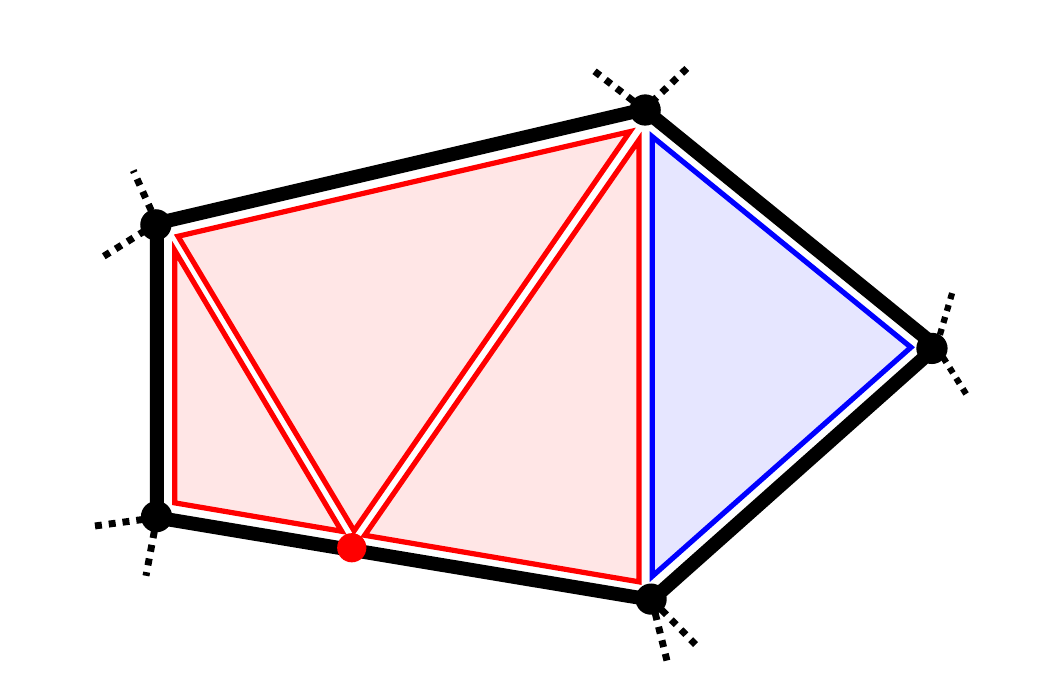}
		\caption{Finite case}
	\end{subfigure}
	\hskip 32pt
	\begin{subfigure}[b]{0.3\textwidth}
		\includegraphics[width=1\textwidth]{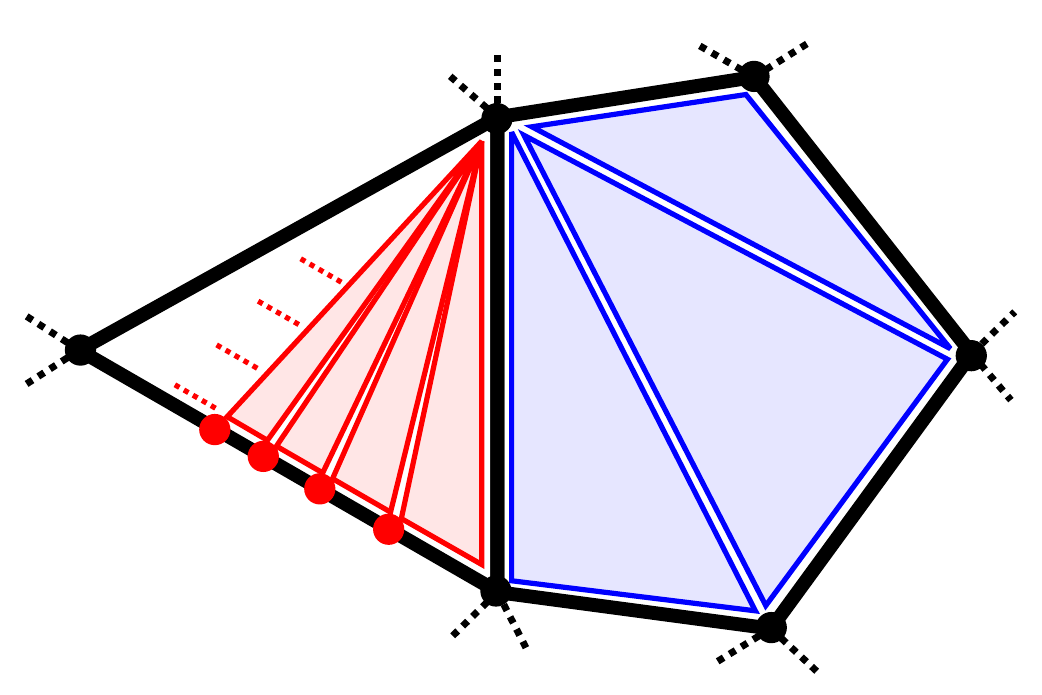}		
		\caption{Infinite case}
	\end{subfigure}
	\caption{Illustrative examples of the (redundant) triangulation of a tropical fan by a (a) finite and (b) infinite cluster algebra. Each of the figures depicts two cones of a $3$-dimensional fan intersected with the unit sphere $S^2$ in black. The cones and the redundant rays from the redundant triangulation are drawn in red, those from the non-redundant triangulation in blue.}
	\label{fig:infiniteRedundantTriangulation}
\end{figure}

If the cluster algebra is infinite, e.g. in our case of $\G{4,n}$ for $n\geq 8$, it consists of infinitely many $\A$-variables and thus also rays. In the regions of the ambient space $\R^d$ that are covered by the cluster fan, it again refines the fan $F_{k,n}$ of the tropical configuration space. Since the latter is by construction always finite, almost all of the cluster rays are redundant. It is therefore natural to expect that the nature of the infinities of the cluster algebra can be interpreted as a redundant triangulation of $F_{k,n}$ with infinitely many cones of the cluster fan containing redundant rays, as is illustrated on the right hand side of figure~\ref{fig:infiniteRedundantTriangulation}. Note that this (redundant) triangulation property also applies to the fan of any partial tropicalisation of $\PConf{k,n}$, such as $pF_{4,n}$, which is a coarser version of the fully tropicalised fan and hence also triangulated by the cluster fan.

To tame the infinity of the cluster algebra, we utilize this relation between the cluster fan and the fan $pF_{4,n}$ of the partially tropicalised positive configuration space and introduce the following selection rule: whenever we encounter a cluster containing a redundant ray we stop mutation in this direction and discard all such redundant clusters. Starting from the initial cluster, which by construction does not contain redundant rays, and mutating in this way, we obtain a finite subset of the infinite cluster algebra -- the \emph{truncated cluster algebra}. Each of the rational $\A$-variables in this subset is then by construction in one-to-one correspondence to a tropical ray of $pF_{4,n}$.

\section{Infinite mutation sequences and square-root letters}
\label{sec:A11}
In the previous section, we reviewed how the relation between (partially) tropicalised Grassmannians and $\G{4,n}$ cluster algebras always allows one to select a finite subset of $\A$-variables of the latter, even in the $n\ge 8$ case where they are infinite. This selection rule is then expected to yield the rational letters of the $n$-particle amplitude in $\N = 4$ pSYM.

In the $n\geq 8$ case, non-rational letters are also expected to appear, and a conceptual advance for obtaining them was achieved in~\cite{Henke:2019hve,Arkani-Hamed:2019rds,Drummond:2019cxm}, building on earlier mathematical developments~\cite{Canakci2018,Reading2018b}: The main idea, applied more concretely to $\G{4,8}$, was to also consider infinite mutation sequences of a rank-2 affine subalgebra, conventionally denoted as $\operatorname{A}^{(1)}_1$ in the corresponding Dynkin diagram classification, starting from the clusters singled out by the aforementioned selection rule. For certain of these mutation sequences, the limiting cluster ray does yield a ray of $\pTPTC{4,8}$ that was not previously accessible by the selection rule, as well as associated square-root letters.

When analysing these infinite mutation sequences, all but two nodes of the cluster we start from may be considered as frozen.  In order for this analysis to be able to cover $\pTPTC{4,n}$ for any $n$, where the frozen nodes will have different structure, it is therefore necessary to work out $\operatorname{A}^{(1)}_1$ sequences with general coefficients. We carry out this task, which also has its intrinsic mathematical merit, in subsection~\ref{sec:A11Theory}. 

As a check of our formalism, in subsection \ref{sec:A11Application} we reapply it to the $\pTPTC{4,8}$ case, and confirm that it provides two rays associated to 18 square-root letters, as was previously found in~\cite{Drummond:2019cxm}. Then, in subsection~\ref{sec:SDAlphabet} we compare these results with a more recent refinement of the works~\cite{Henke:2019hve,Arkani-Hamed:2019rds,Drummond:2019cxm}, based on the framework of wall-crossing and scattering diagrams~\cite{Herderschee:2021dez}. While this approach naively predicts a large number of additional non-rational letters, very interestingly we find that these are in fact only two: The inequivalent realisations of the four-mass box Gram determinant by eight cyclically ordered massless legs, $\Delta_{1,3,5,7}$ and $\Delta_{2,4,6,8}$. In the next section, we will further apply the methods of sec.~\ref{sec:A11Theory} to $\pTPTC{4,9}$, and thus obtain new predictions for the singularities of nine-particle amplitudes.

\subsection{Mutation sequence of type $\operatorname{A}^{(1)}_1$ with general coefficients}
\label{sec:A11Theory}
In this subsection, we study the infinite mutation sequence of type $\operatorname{A}^{(1)}_1$ with general coefficients, depicted in figure~\ref{fig:mutSequ}. The coefficient-free and principal coefficients case was first studied in~\cite{Canakci2018} and~\cite{Reading2018b}, and they were later rediscovered in the context of amplitude alphabets, either as above or with a special case of frozen variables in~\cite{Henke:2019hve,Arkani-Hamed:2019rds,Drummond:2019cxm}. For the convenience of the reader, here we sketch the essential steps for the solution of the corresponding recurrence relation, and defer the remaining proofs and detailed calculations to appendix~\ref{sec:Am1Proofs}.
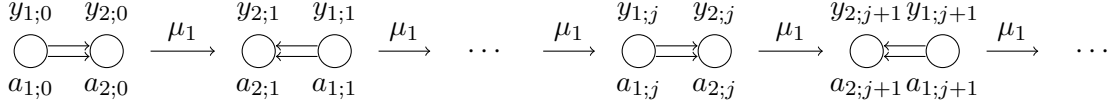
\begin{figure}[ht]
	\centering
	\begin{tikzpicture}
		\node[label={above:$y_{1;0}$},label={below:$a_{1;0}$},scale=1.2] at (0,0) [circle,draw] (a) {};
		\node[{label=above:$y_{2;0}$},label={below:$a_{2;0}$},scale=1.2] at (1,0) [circle,draw] (b) {};
		
		\node[label={above:$y_{2;1}$},label={below:$a_{2;1}$},scale=1.2] at (3,0) [circle,draw] (c) {};
		\node[label={above:$y_{1;1}$},label={below:$a_{1;1}$},scale=1.2] at (4,0) [circle,draw] (d) {};
		
		\node at (6,0) (e) {$\cdots$};
		
		\node[label={above:$y_{1;j}$},label={below:$a_{1;j}$},scale=1.2] at (8,0) [circle,draw] (f) {};
		\node[label={above:$y_{2;j}$},label={below:$a_{2;j}$},scale=1.2] at (9,0) [circle,draw] (g) {};
		
		\node[label={above:$y_{2;j+1}$},label={below:$a_{2;j+1}$},scale=1.2] at (11,0) [circle,draw] (h) {};
		\node[label={above:$y_{1;j+1}$},label={below:$a_{1;j+1}$},scale=1.2] at (12,0) [circle,draw] (i) {};
		
		\node at (14,0) (j) {$\cdots$};
		
		\draw[->] ([yshift=2pt]a.east) -- ([yshift=2pt]b.west);
		\draw[->] ([yshift=-2pt]a.east) -- ([yshift=-2pt]b.west);
		
		\draw[->,shorten >=10pt,,shorten <=10pt](b) edge node[above] {$\mu_1$} (c);
		
		\draw[->] ([yshift=2pt]d.west) -- ([yshift=2pt]c.east);
		\draw[->] ([yshift=-2pt]d.west) -- ([yshift=-2pt]c.east);
		
		\draw[->,shorten >=10pt,,shorten <=10pt](d) edge node[above] {$\mu_1$} (e);
		
		\draw[->,shorten >=10pt,,shorten <=10pt](e) edge node[above] {$\mu_1$} (f);
		
		\draw[->] ([yshift=2pt]f.east) -- ([yshift=2pt]g.west);
		\draw[->] ([yshift=-2pt]f.east) -- ([yshift=-2pt]g.west);
		
		\draw[->,shorten >=10pt,,shorten <=10pt](g) edge node[above] {$\mu_1$} (h);		
		
		\draw[->] ([yshift=2pt]i.west) -- ([yshift=2pt]h.east);
		\draw[->] ([yshift=-2pt]i.west) -- ([yshift=-2pt]h.east);
		
		\draw[->,shorten >=10pt,,shorten <=10pt](i) edge node[above] {$\mu_1$} (j);
		
	\end{tikzpicture}
	\caption{Infinite mutation sequence in the affine rank-$2$ cluster algebra of $\operatorname{A}^{(1)}_1$ Dynkin type.}
	\label{fig:mutSequ}
\end{figure}

Consider the rank-2 affine cluster algebra of $\operatorname{A}^{(1)}_1$ Dynkin type with 2 $A$-variables and $M$ frozen variables. The coefficients $y_{1;0}$ and $y_{2;0}$ are given in terms of the frozen variables $a_3,\dots,a_M$ and the adjacency matrix $b^0_{ij}$ of the initial cluster as
\begin{equation}
	\label{eq:InitialCoeffs}
	y_{j;0} = \prod_{i=3}^{M+2}a_i^{b^0_{ij}}\,,\quad j=1,2\,.
\end{equation}

As depicted in figure~\ref{fig:mutSequ}, repeated mutation at node 1 gives rise to sequences of $\A$-variables $a_{i;j}$ and coefficients $y_{i;j}$ with initial values $a_{1;0}, a_{2;0}$ and $y_{1;0}, y_{2;0}$, respectively.\footnote{Recall that according to the notation introduced in section~\ref{sec:BackgroundCACoefs}, we have found convenient to label all variables with a pair of indices indicating their position $i$ in a given cluster $j$, such that the mutation along the sequence is always on position $i=1$.} Furthermore, we may also consider the sequence of associated $\X$-variables, which are given by $x_{1;j} = a_{2;j}^{-2}y_{1;j}$ and $x_{2;j} = a_{1;j}^{2}y_{2;j}$. Using the general mutation rules, eqs.~\eqref{eq:AMutationRule} and~\eqref{eq:YMutationRule}, we obtain the following recursion relations for these sequences for $j\in \Z$
\begin{align}
	a_{2;j+1} &= \frac{a_{2;j}^2}{a_{1;j}}\frac{1+x_{1;j}}{1\cltrad y_{1;j}}\,,\quad a_{1;j+1} = a_{2;j} \,, \label{eq:ASequence} \\
	y_{1;j+1} &= \frac{y_{2;j}\,y_{1;j}^2}{\left(1\cltrad y_{1;j}\right)^2}\,,\quad y_{2;j+1} = \left(y_{1;j}\right)^{-1} \,, \label{eq:YSequence} \\
	x_{1;j+1} &= \frac{x_{2;j}\,x_{1;j}^2}{\left(1 + x_{1;j}\right)^2}\,,\quad x_{2;j+1} = \left(x_{1;j}\right)^{-1} \,. \label{eq:XSequence}
\end{align}
Note that the mutation rule for the $\A$-variables presented here is of the form of eq.~\eqref{eq:AMutationRuleX} and again the mutation rule of the $\X$-variables is the same as that of the coefficients with cluster-tropical addition replaced by normal addition. Further to these sequences, we introduce the sequence $\beta_j$ of ratios of consecutive $\A$-variables
\begin{equation}
	\beta_{j} = \frac{a_{2;j}}{a_{1;j}} \,,
\end{equation}
which, by eq.~\eqref{eq:ASequence}, can equivalently be expressed as $a_{1;j+1}/a_{1;j}$ or $a_{2;j}/a_{2;j-1}$. 

Having obtained the recursion relations for the sequence, we will now turn to its solution by linearizing eq.~\eqref{eq:ASequence}. Similar to~\cite{Arkani-Hamed:2019rds,Fordy:2009qz}, we define the two quantities $K_{1,j}$ and $K_{2,j}$ for each cluster. They are given by
\begin{align}
	K_{1,j} &= \left(\gamma_0\gamma_{j}^{-1}\beta_0^{-1}\beta_{j}\right)\left[1+x_{1;j}+x_{1;j}\left(x_{1;j-1}\right)^{-1}\right] \,,\\
	K_{2,j} &= \left(\gamma_0\gamma_{j}^{-1}\beta_0^{-1}\beta_{j}\right)^2\left[x_{1;j}\left(x_{1;j-1}\right)^{-1}\right]\,,
\end{align}
where we have included a factor $\gamma_0\beta_0^{-1}$ so as to normalize these quantities at $j=0$ for later convenience, and the auxiliary sequence $\gamma_j$ is -- up to the prefactors -- the cluster-tropical version of $K_{1,j}$, defined as
\begin{equation}
	\gamma_{j} = 1\cltrad y_{1;j}\cltrad y_{1;j}\left(y_{1;j-1}\right)^{-1}\,.
\end{equation}

Note that $K_{1,j}$, $K_{2,j}$ and $\gamma_j$ depend only on the data of the cluster $j$, as can be seen by noting that due to eqs.~\eqref{eq:YSequence}--\eqref{eq:XSequence}, $\left(x_{1;j-1}\right)^{-1} = x_{2;j}$ and $\left(y_{1;j-1}\right)^{-1} = y_{2;j}$. As follows from the mutation relations, $K_{1,j}$ and $K_{2,j}$ are actually invariant along the sequence, see appendix~\ref{sec:Am1Proofs}. From now on they will be denoted by $K_1$ and $K_2$, respectively, and in terms of the  initial $\X$-variables they are explicitly given by
\begin{equation}
	\label{eq:A11_invariants}
	K_1 \equiv K_{1,0} = 1+x_{1;0}+x_{1;0}x_{2;0}\,,\quad K_2 \equiv K_{2,0} = x_{1;0}x_{2;0}\,.
\end{equation}

Using the invariants, we can linearize the recursion relation of the $\A$-variables, eq.~\eqref{eq:ASequence}, to obtain
\begin{equation}
	\label{eq:A11LinRecAs}
	\gamma_{j}^{-1}\gamma_{j+1}^{-1} a_{1;j+2} - \gamma_0^{-1}\beta_0 K_1\cdot \gamma_{j}^{-1}  a_{1;j+1} + \gamma_0^{-2}\beta_0^2 K_2 \cdot a_{1;j} = 0\,.
\end{equation}
By using $a_{1;j+1} = a_{2;j}$, this recursion can be recast to give an equation for the variables of cluster $j+1$ in terms of variables of cluster $j$ only. While this is a linear recurrence relation, the fact that it does not have constant coefficients makes its solution more intricate. However, by considering the new sequence $\alpha_j$, defined by
\begin{equation}
	\label{eq:A11Alpha}
	\alpha_j = \gamma_0^{-1}\gamma_1^{-1}\dots \gamma_{j-1}^{-1} \cdot a_{1;j}
\end{equation}
for $j\geq 0$ we obtain another recurrence relation given by
\begin{equation}
	\label{eq:A11AlphaRecRel}
	\alpha_{j+2} - \gamma_0^{-1}\beta_0 K_1\cdot \alpha_{j+1} + \gamma_0^{-2}\beta_0^2K_2\cdot \alpha_j = 0\,.
\end{equation}
Being a linear recurrence with constant coefficients, it may now be solved by standard methods based on its characteristic polynomial, 
\begin{equation}
	P_1\left(t\right) = t^2 - \gamma_0^{-1}\beta_0 K_1 \cdot t + \gamma_0^{-2}\beta_0^2K_2\,.
\end{equation}
The latter has two roots $\beta_\pm$ that are given by
\begin{equation}
	\label{eq:beta}
	\beta_\pm = \frac{a_{2;0}}{a_{1;0}}\frac{K_1 \pm \sqrt{K_1^2 - 4 K_2}}{2\gamma_0}\,,
\end{equation}
and which correspond to the $j\to\infty$ limit of both the ratio $\alpha_{j+1}/\alpha_j$ and $\beta_j$, as can be seen by first observing that $\alpha_{j+1}/\alpha_j$ = $\gamma_j^{-1}\beta_j$. Using that $\gamma_j\to 1$ for $j\to \infty$, as proven in appendix~\ref{sec:Am1Proofs}, it follows that, assuming convergence, the ratio $\alpha_{j+1}/\alpha_j$ has the same limit as $\beta_j$. Furthermore, dividing the recurrence~\eqref{eq:A11AlphaRecRel} by $\alpha_j$ and taking the limit $j\rightarrow \infty$, we see that the limit of the $\alpha$-ratio, and thus that of $\beta_j$, is given by the roots of $P_1$.

Using the roots of the characteristic polynomial, we obtain the most general solution for the sequence $a_{1;j}$ as\footnote{Note that this is in fact the general solution of the quantity $\alpha_j$ of eq.~\eqref{eq:A11Alpha}: We have not attempted to also find the general solution of the $\gamma_j$ prefactor, since for our purposes only ratios where this prefactor cancels will be needed.}
\begin{equation}
	\label{eq:A11Sol1}
	a_{1;j} = \left(\gamma_0\gamma_1\cdots\gamma_{j-1}\right)\left[C_+\left(\beta_+\right)^j + C_-\left(\beta_-\right)^j\right]\,,
\end{equation}
whereas we have used eq.~\eqref{eq:A11Alpha} to express $a_{1;j}$ in terms of $\alpha_j$. Note that since $\gamma_j$ becomes $1$ for $j\geq J$ for some integer $J$, the prefactor becomes constant at some point. The coefficients $C_\pm$ can be fixed via the initial conditions $a_{1;0}$ and $a_{1;1} = \gamma_0^{-1}a_{2;0}$ to be
\begin{equation}
	\label{eq:A11C}
	C_\pm = a_{1;0}\frac{\pm2\mp K_1+\sqrt{K_1^2-4K_2}}{2\sqrt{K_1^2-4K_2}}\,.
\end{equation}

In addition to the infinite mutation sequence of repeatedly mutating $a_{1;j}$, we can also consider the opposite direction, which amounts to repeatedly mutating $a_{2;j}$. As is explained in more detail in appendix~\ref{sec:Am1Proofs}, the solution to this recurrence can be obtained in a similar way and is given by
\begin{equation}
	\label{eq:A11Sol2}
	a_{2;-j} = \left(\gamma_0\gamma_{-1}\cdots\gamma_{-j+1}\right)^{-1}\left[\tilde{C}_+\left(\beta_-\right)^{-j} + \tilde{C}_-\left(\beta_+\right)^{-j}\right]\,,
\end{equation}
for $j\geq 0$. As is demonstrated in appendix~\ref{sec:Am1Proofs}, $\gamma_{-j}$ also becomes $1$ for $j\geq J$ for some integer $J$. The coefficients $\tilde{C}_\pm$ can again be obtained from the initial conditions and are given by
\begin{equation}
	\label{eq:A11CT}
	\tilde{C}_\pm = a_{2;0}\frac{\pm2K_2\mp K_1+\sqrt{K_1^2-4K_2}}{2\sqrt{K_1^2-4K_2}}\,.
\end{equation}

Finally, let us briefly comment on the associated limit ray. In the $j\to\infty$ limit, the ratio of consecutive cluster variables, eq.~\eqref{eq:beta}, obeys a generalized form of the separation principle of eq.~\eqref{eq:AVarRayForm}. That is, it factorizes into a monomial in the initial $\A$-variables times a ratio of algebraic function with a cluster-tropical sum, which can be interpreted as the (generalized) cluster-tropical version of the algebraic function. In particular, as we have explained above eq.~\eqref{eq:beta}, $\gamma_0$ is the tropical version of $K_1$, and it is natural to consider it also as the  (generalized) cluster-tropical version of $\sqrt{K_1^2-4K_2}$. Analogously to the definition of $g$-vectors from the exponents of the $\A$-coordinate representation of eq.~\eqref{eq:AVarRayForm}, from eq.~\eqref{eq:beta} we may thus associate  $\g_\beta = \left(-1,1\right)$ to the limit. Indeed, considering the sequence of $\g$-vectors associated to the sequence of $\A$-variables, we find that it does converge to $\g_\beta$. Note, however, that this is the limit ray with reference to the $\operatorname{A}^{(1)}_1$ cluster algebra only. In practice, to obtain the limit ray of some embedding of such a cluster algebra, we use the mutation relation for the cluster rays given by eq.~\eqref{eq:rayMutation}.

\subsection{Application: $\pTPTC{4,8}$ and the eight-particle alphabet}
\label{sec:A11Application}
The remarkable feature of the infinite mutation sequences considered in the previous subsection is that they yield quantities containing square roots, see eqs.~\eqref{eq:A11Sol1},~\eqref{eq:A11C}, and~\eqref{eq:A11Sol2},~\eqref{eq:A11CT}. The main idea of the works~\cite{Henke:2019hve, Drummond:2019cxm, Arkani-Hamed:2019rds} was that these quantities thus provide natural candidates for the non-rational letters of amplitudes, focusing on the then-unknown frontier of multiplicity eight, related to the $\G{4,8}$ cluster algebra. Note that while the values of the above-mentioned quantities differ depending on the choice of coefficients, square roots are always present.

In more detail, all three aforementioned papers identified $\operatorname{A}^{(1)}_1$ as an affine rank-2 subalgebra of the cluster algebra in question, and~\cite{Henke:2019hve} analysed this subalgebra with principal coefficients as a proof of concept. References~\cite{Drummond:2019cxm} and \cite{Arkani-Hamed:2019rds} additionally found the generating functional of the mutation sequences, roughly equivalent to the general solutions~\eqref{eq:A11Sol1} and~\eqref{eq:A11Sol2}, for the particular case of frozen variables required to analyse the $\G{4,8}$ cluster algebra, from the tropical geometry and stringy canonical form approach, respectively.

If one assumes a one-to-one correspondence between tropical rays and letters of the symbol alphabet, then the natural choice also respecting the natural symmetry of the latter is the ratio $\beta_+/\beta_-$. This possibility cannot be currently excluded if one restricts to MHV amplitudes, which are technically speaking the only ones having $\Conf{4,n}$ as their space of kinematics (beyond MHV, the analysis of the kinematic space is complicated by the existence of rational functions on top of the transcendental ones studied here). Indeed, through the currently known loop order $L=2$, the eight-particle MHV amplitude only contains rational letters~\cite{CaronHuot:2011ky,Golden:2021ggj}, so it could be that the only additional square-root letters starting to appear at $L\ge 3$ are those uniquely associated to tropical rays.

Nevertheless, in~\cite{Drummond:2019cxm} it was further noticed that if the direction of approach to the limit ray is also taken into account, such that many square-root letters are associated to each limit ray in a particular fashion, then one in fact obtains the complete alphabet of the 2-loop NMHV amplitude~\cite{Zhang:2019vnm}, which the unique association of letters to rays cannot account for. Along with a complementary analysis based on plabic graphs~\cite{Mago:2020kmp,He:2020uhb,Mago:2020nuv}, see also~\cite{Mago:2021luw} appearing simultaneously with this paper, this seems to suggest that despite the apparent complications mentioned above for non-MHV amplitudes in $\N=4$ pSYM, the symbol alphabet may be independent of the helicity configuration, at least at multiplicity $n=8$.

In this paper, we will adopt the prescription of~\cite{Drummond:2019cxm} for associating many square-root letters to a given limiting ray associated to a given $\operatorname{A}^{(1)}_1$ subalgebra of  the truncated $\G{4,n}$ cluster algebra, which in the conventions of the previous subsection is given by\footnote{More precisely, when specialized to $n=8$ the formulas below reduce to the negative inverse of those provided in the latter reference, with this difference being immaterial at the level of symbol letters.}
\begin{equation}
	\label{eq:algLet}
	\phi_0\equiv \frac{{C}_+}{{C}_-} = \frac{2 - K_1 + \sqrt{K_1^2 - 4K_2}}{-2 + K_1 + \sqrt{K_1^2 - 4K_2}}\,,\quad \tilde{\phi}_0\equiv \frac{\tilde{C}_+}{\tilde{C}_-}=  \frac{2K_2 - K_1 + \sqrt{K_1^2 - 4K_2}}{-2K_2 + K_1 + \sqrt{K_1^2 - 4K_2}}\,,
\end{equation}
where the quantities appearing here have been defined in eqs.~\eqref{eq:A11_invariants},~\eqref{eq:A11C} and~\eqref{eq:A11CT}.

The great merit of the analysis we carried out in the previous section, and of the above formula, is that it can be directly applied to  \emph{any} such subalgebra for \emph{any} $n$, not necessarily equal to eight. All we need as input is the data of a given \emph{origin cluster}, namely the cluster containing a $\operatorname{A}^{(1)}_1$ cluster subalgebra from which the infinite mutation sequence starts, such as for example the one depicted in figure \ref{fig:OctagonOriginCluster} for $n=8$.\footnote{More concretely, to apply the above formulas in this example, we only need to evaluate eq.~\eqref{eq:A11_invariants} with $x_{1;0} \to x_{1}$ and $x_{2;0} \to x_{9}$.}

So our results can in principle be specialized to yield predictions for the symbol alphabet of scattering amplitudes at any multiplicity $n$, and in the next section we will indeed apply them to the $n=9$ case. As a first cross check however, in the remainder of this subsection we will use our method to confirm the results reported in~\cite{Drummond:2019cxm} for the eight-particle alphabet. Let us also comment that while the prescription~\eqref{eq:algLet} may currently seem ad-hoc and only justified by the agreement of its symbol alphabet predictions with explicit computations, in the next subsection we will provide further evidence about its correctness by comparing it with a more recent approach based on scattering diagrams and wall-crossing~\cite{Herderschee:2021dez}.

Focusing now on the case of eight-particle scattering, we find a total of 3,600 origin clusters with a $\operatorname{A}^{(1)}_1$ cluster subalgebra in the 121,460 clusters of the cluster algebra of $\G{4,8}$ truncated by $\pTPTC{4,8}$. The rays of the variables mutated in the infinite mutation sequences starting at these origin clusters converge to four different limit rays. Interestingly, only two of these limit rays are contained in $\pTPTC{4,8}$, whereas the other two are contained only in $\TPTC{4,8}$ and not its partially tropicalised version. Similar to the truncation rule used to obtain a finite subset of rational cluster variables, we discard the limits of the origin quivers whose limit rays are not contained in $\pTPTC{4,8}$. This leaves us with a truncated set of 2,800 origin clusters of which 56 for each of the two limit rays are unique,\footnote{The 800 origin clusters whose limit rays are only contained in $\TPTC{4,8}$ reduce to 32 different origin clusters for each of the rays.} since many of these clusters only differ in parts of the quiver that do not affect the $\operatorname{A}^{(1)}_1$ cluster subalgebra and thus also not the limit of its infinite mutation sequence.

In this way, we find a total of 112 different origin clusters giving rise to 224 square-root letters via eqs.~\eqref{eq:algLet}. However, these 224 letters are not multiplicatively independent. From the perspective of the alphabet -- essentially the set of logarithms of these letters -- this means that not all of the letters are linearly independent and hence are redundant. While the presence of square roots complicates the elimination of these redundancies, this can be done using an approach similar to that of~\cite{Mitev:2018kie}, which we will describe in more detail in section~\ref{sec:NineParticleAlgebraic}, where it is used again. For the case at hand, we find that the aforementioned 224 letters reduce to 18 multiplicatively independent letters, which are equivalent to the square-root letters reported in~\cite{Drummond:2019cxm} and previously known to appear in the two-loop NMHV eight-particle amplitude, as computed in~\cite{Zhang:2019vnm}.\footnote{Note that the $18$ square-root letters can alternatively be obtained by solving polynomial equations associated to certain plabic graphs~\cite{Mago:2020kmp,He:2020uhb}. As soon as one attempts to also incorporate rational letters in this approach, however, non-plabic graphs are required as well~\cite{Mago:2020nuv}. In this case the solution space includes all cluster variables of $\G{4,8}$, that is the alphabet becomes infinite again.}

To summarise, in total we obtain all $274$ tropical rays of $\pTPTC{4,8}$ -- $272$ rays associated to one rational letter of the truncated cluster algebra each and $2$ rays associated to $9$ multiplicatively independent square-root letters obtained as the limits of infinite mutation sequences each. This $290$-letter alphabet contains all letters previously known to appear in the eight-particle MHV and NMHV amplitudes.

\subsection{Comparison with the scattering diagram approach}
\label{sec:SDAlphabet}
A common element of the majority of different efforts to predict the symbol alphabet of $n$-particle amplitudes in $\N=4$ pSYM, initially based on Grassmannian cluster algebras~\cite{Golden:2013xva}, and more recently on tropical Grassmannians~\cite{Henke:2019hve,Drummond:2019cxm} or stringy canonical forms~\cite{Arkani-Hamed:2019rds}, is that they correspond to different compactifications of the positive part of the space of kinematics $\Conf{4,n}$. More recently, another such compactification refining the aforementioned works, and relying on the concepts of wall-crossing and scattering diagrams, has been proposed~\cite{Herderschee:2021dez}.

Having discussed the predictions of the tropical geometry approach for the eight-particle alphabet in the previous subsection, here we will compare them with those of the scattering diagram approach. In a nutshell, while the latter yields 72 letters (36 per limit ray of $\pTPTC{4,8}$), out of which 56 are naively non-rational, we will show that in fact \emph{all} of these letters are contained in the 290-letter octagon alphabet of the previous subsection, \emph{except} for the square roots of the two Gram determinants associated to the four-mass box,
\begin{align}
	\Delta_{1,3,5,7} &= \left(1 - \frac{\pl{1234}\pl{5678}}{\pl{1256}\pl{3478}} - \frac{\pl{1278}\pl{3456}}{\pl{1256}\pl{3478}}\right)^2 - 4 \frac{\pl{1278}\pl{1234}\pl{3456}\pl{5678}}{\left(\pl{1256}\pl{3478}\right)^2}\,,\label{eq:Delta8}\\
	\Delta_{2,4,6,8} &= \left(1 - \frac{\pl{2345}\pl{1678}}{\pl{2367}\pl{1458}} - \frac{\pl{1238}\pl{4567}}{\pl{2367}\pl{1458}}\right)^2 - 4 \frac{\pl{1238}\pl{2345}\pl{4567}\pl{1678}}{\left(\pl{2367}\pl{1458}\right)^2}\,,\label{eq:Delta8_2}
\end{align}
whereas $\Delta_{2,4,6,8}$ is related to $\Delta_{1,3,5,7}$ by the cyclic shift $\pl{ijkl}\to\pl{i+1\,j+1\,k+1\,l+1}$.

We view the almost complete overlap of the two approaches at multiplicity $n=8$ as a strong indication of their correctness, and will further comment on the presence or absence of the extra letters, eqs.~\eqref{eq:Delta8} and~\eqref{eq:Delta8_2}, from amplitudes and Feynman integrals. To provide more general backing to this conclusion, later in this section we will also show that the square-root letters obtained from the tropical geometry approach are always contained in those of the scattering diagram approach for any $n$. But before this, let us briefly provide some background information on scattering diagrams.

\noindent {\bf Basics of scattering diagrams}. A scattering diagram can be thought of as a generalisation of the $\g$-vector fan of the cluster algebra, as defined by eq.~\eqref{eq:AVarRayForm}, that also contains the limits of infinite mutation sequences. In the $\g$-vector fan of a cluster algebra, each $\A$-variable is associated to one of the rays in the fan. The rays of all variables in a cluster then form a cone, which intersects other cones that share some of the variables. A codimension-1 intersection, or \emph{wall}, between cones that share all but one variable corresponds to the mutation of the variable that is not shared.

In the scattering diagram, we associate a variable $x_{\gamma_i}$ to each of the rays in a cone. These \emph{cone variables} are related to the $\X$-variables of the cone in the following way. Consider a wall of the cone and the $\X$-variable $x_j$ that is mutated when passing through the wall. Denote the (appropriately normalized) vector perpendicular to the wall and pointing into the cone by $\gamma^\perp_j = c^j_i \gamma_i$, whereas the $\gamma_i$ denote the canonical basis vectors. We then have the relation
\begin{equation}
	\label{eq:XandSDvars}
	x_j = \prod_i \left(x_{\gamma_i}\right)^{c^j_i}\,.
\end{equation}
Alternatively, we can use the inverse of eq.~\eqref{eq:XandSDvars} to express the variables $x_{\gamma_i}$ in terms of the $\X$-variables $x_i$ of the cone. Note that the labelling is such that the wall denoted by $j$ is that which is spanned by the rays associated to all $\A$-variables $a_i$ except that of $a_j$. As we will see shortly, the advantage of using the cone variables is that they remain finite and have a well-defined limit in the relevant infinite mutation sequences.

The mutation, or \emph{wall crossing}, of the cone variables is implemented by multiplying them with powers of a function $f(x_{\gamma_j^\perp})$ which is attached to each of the walls, whereas the argument $x_{\gamma_j^\perp}$ is equal to $x_j$ or its inverse, depending on the side from which the wall is approached. For walls that are part of the $\g$-vector fan, this function is given by $f(x_{\gamma_j^\perp}) = 1 + x_{\gamma_j^\perp}$. Together with eq.~\eqref{eq:XandSDvars}, this reproduces the mutation rule for the $x_i$. Extending the cluster algebra framework, the wall crossing function for walls that are not part of the cluster algebra can be obtained by self-consistency conditions.

\noindent {\bf Eight-particle alphabet predictions and comparison with tropical geometry.} Let us now review and further analyse the predictions of the scattering diagrams framework for the alphabet of the eight-particle amplitude~\cite{Herderschee:2021dez}, as well as compare them to the 290-letter $\pTPTC{4,8}$ alphabet discussed in the previous section. We will only discuss the boundary structure around one of the two limit rays of $\G{4,8}$, since the letters associated to the other can be obtained by the cyclic shift $\pl{ijkl}\to\pl{i+1\,j+1\,k+1\,l+1}$.

In a first step, one mutates from the initial cluster to an origin cluster containing a $\operatorname{A}^{(1)}_1$ subalgebra. Concretely, performing the mutations $\{1,2,4,1,6,8\}$ leads to the cluster depicted in figure~\ref{fig:OctagonOriginCluster}. The parameterisation of the $\X$-variables $x_i$ in this cluster in terms of Plücker variables as well as all other data required to reconstruct the non-rational alphabet can be found in appendix~\ref{sec:EightParticleAlgebraicAlphabet}.
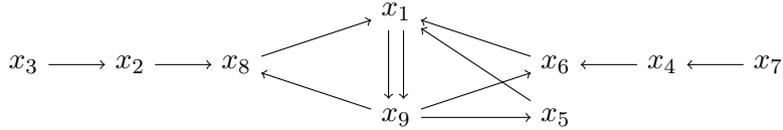
\begin{figure}[ht]
	\centering
	\begin{tikzpicture}[scale=1.4]
		\node at (0,0.5) (6) {$x_3$};
		\node at (1,0.5) (9) {$x_2$};
		\node at (2,0.5) (2) {$x_8$};
		\node at (3.5,1) (5) {$x_1$};
		\node at (5,0.5) (1) {$x_6$};
		\node at (6,0.5) (7) {$x_4$};
		\node at (7,0.5) (8) {$x_7$};
		
		\node at (3.5,0) (3) {$x_9$};
		\node at (5,0) (4) {$x_5$};
		
		\draw[->] (6) -- (9);
		\draw[->] (9) -- (2);
		\draw[->] (2) -- (5);
		\draw[->] (1) -- (5);
		\draw[->] (7) -- (1);
		\draw[->] (8) -- (7);
		
		\draw[->] (3) -- (2);
		\draw[->] (3) -- (1);
		\draw[->] (4) -- (5);
		\draw[->] (3) -- (4);
		
		\draw[->] ([xshift=2pt]5.south) -- ([xshift=2pt]3.north);
		\draw[->] ([xshift=-2pt]5.south) -- ([xshift=-2pt]3.north);
	\end{tikzpicture}
	\caption{Principal part of the origin cluster in $\G{4,8}$ utilized to find the square-root letters. From this quiver, it is also evident that $\G{4,8}\simeq E_7^{(1,1)}$ in the extended affine Dynkin diagram classification~\cite{fomin2006cluster}. }
	\label{fig:OctagonOriginCluster}
\end{figure}

Next, one expresses the cone variables along the $\operatorname{A}^{(1)}_1$ sequence originating from this cluster in terms of its $\X$-variables $x_i$ by using the inverse of eq.~\eqref{eq:XandSDvars}. Compared to the $x_i$, the cone variables do converge to a finite function when taking the limit of the infinite sequence. These limits correspond to the cone variables $x^0_{\gamma_i}$ of a cone asymptotically close to the limit ray, also known as an \emph{asymptotic chamber}, and are given by 
\begin{align}
	x_{\gamma_i}^0 &= x_i\,\quad\text{for}\quad i \in \{2,3,4,7\}\,, \nonumber\\
	x_{\gamma_i}^0 &= \frac{x_i}{2}\left(1+x_1\left(1+x_9\right) + \sqrt{\Delta'}\right)\,\quad\text{for}\quad i \in \{5,6,8\}\,, \label{eq:asymWallVars} \\
	x_{\gamma_1}^0 &= \frac{4x_1\Delta'}{\left(1+x_1-x_1x_9+\sqrt{\Delta'}\right)^2}\,,\quad 	x_{\gamma_9}^0 = \frac{x_9}{4}\left(1+\frac{1-x_1(1+x_9)}{\sqrt{\Delta'}}\right)^2\,,\nonumber
\end{align}
where $\Delta' = \left(1+x_1(1+x_9)\right)^2 - 4x_1x_9\,$. In contrast to just considering the cluster algebra itself, one can now utilise wall crossing to find other asymptotic chambers and their variables. This was carried out in~\cite{Herderschee:2021dez} by means of an extensive computer search, yielding a basis of $36$ multiplicatively independent polynomials of the $x_{\gamma_i}^0$, proposed to contain all non-rational letters (in the original $\mathcal{X}$- or Pl\"ucker variables) of the eight-particle amplitude. It was also noticed that $10$ of these polynomials depend only on the rational $\mathcal{X}$-variables $x_i\,$ for $i \in \{2,3,4,7\}$, such that the set of non-rational letters is immediately reduced to (a maximum of) $26$ letters.

Very interestingly, we notice that for another $6$ of these letters the square roots contained in them cancel out, such that they are also secretly rational. What is more, there are $10$ multiplicative combinations of the remaining $20$ letters that turn out to be rational as well,\footnote{We thank Dima Chicherin for pointing out the existence of these additional relations to us.} see appendix~\ref{sec:EightParticleAlgebraicAlphabet} for details. All in all, this implies that the scattering diagram approach in fact predicts 26 rational and 10 square-root letters for one of the two $\G{4,8}$ limit rays, and more concretely the latter ones may be chosen to be
\begin{align}
	f_1 &= \left(x_{\gamma_1}^0\right)^{-1}\left(1-x_{\gamma_1}^0x_{\gamma_9}^0\right)^2\,, &  f_2 &=x_{\gamma_9}^0\left(1-x_{\gamma_1}^0x_{\gamma_9}^0\right)^2\,, & \nonumber \\
	f_3 &= \frac{1+x_{\gamma_5}^0x_{\gamma_1}^0x_{\gamma_9}^0}{1+x_{\gamma_5}^0}\,, & f_4 &= \frac{1+x_{\gamma_8}^0x_{\gamma_1}^0x_{\gamma_9}^0}{1+x_{\gamma_8}^0}\,, & f_5 &= \frac{1+x_{\gamma_2}^0\left(1+x_{\gamma_8}^0x_{\gamma_1}^0x_{\gamma_9}^0\right)}{1+x_{\gamma_2}^0\left(1+x_{\gamma_8}^0\right)}\,, 	\nonumber \\
	f_6 &= \mathrlap{\frac{1+x_{\gamma_3}\left(1+x_{\gamma_2}^0\left(1+x_{\gamma_8}^0x_{\gamma_1}^0x_{\gamma_9}^0\right)\right)}{1+x_{\gamma_3}\left(1+x_{\gamma_2}^0\left(1+x_{\gamma_8}^0\right)\right)}\,,} \nonumber \\
	f_{10} &= x_{\gamma_5}^0\left(1-x_{\gamma_1}^0x_{\gamma_9}^0\right)\,,\label{eq:basisS}
\end{align}
together with $f_7,f_8,f_9$ obtained from replacing $x_{\gamma_3}^0 \to x_{\gamma_7}^0\,$, $x_{\gamma_2}^0 \to x_{\gamma_4}^0\,$, and $x_{\gamma_8}^0 \to x_{\gamma_6}^0\,$. As already mentioned, another 10 letters associated to the other limit ray may be obtained by a cyclic shift of the momentum twistors.

How about the relation of the scattering diagram letters to the tropical 290-letter eight-particle  alphabet, discussed in the previous section? Starting with the 26 rational scattering diagram letters, we find that they are all contained in the $\pTPTC{4,8}$ alphabet. As far as the square-root letters are concerned, as already pointed out in~\cite{Herderschee:2021dez}, 9 of them (plus cyclic) are also contained in the $\pTPTC{4,8}$ alphabet, and we also confirm this to be the case. In the square-root letter basis~\eqref{eq:basisS}, these in particular correspond to $f_1,\ldots,f_9$. So the final conclusion is that the only scattering diagram letter not contained in the $\pTPTC{4,8}$ alphabet is $f_{10}$, which remarkably can be written as (see again appendix \ref{sec:EightParticleAlgebraicAlphabet})
\begin{equation}
	f_{10} = \frac{\pl{1256}\pl{3478}}{\pl{1278}\pl{3456}}\sqrt{\Delta_{1,3,5,7}}\,,
\end{equation}
together with its cyclic image, where the square-roots associated to the four-mass box have been defined in eqs.~\eqref{eq:Delta8} and~\eqref{eq:Delta8_2}. Given that the factor in front of the square root is a monomial in the rational letters, one could equally well redefine the letter so as to remove it. For the interested reader, we provide the complete candidate eight-particle alphabet consisting of the 292 letters coming from the union of the tropical geometry and scattering diagram approaches in the ancillary file~\texttt{Gr48Alphabet.m}, which is attached to the \texttt{arXiv} submission of this article.

The fact that the two approaches overlap almost completely greatly reinforces the expectation that all singularities of eight-particle amplitudes are contained in the aforementioned candidate alphabet. In a sense, scattering diagrams provide a more systematic framework for taking infinite mutation sequences into account, and especially for taking the direction of approach to a given limit ray into account, thus justifying the particular choice of eqs.~\eqref{eq:algLet} for assigning many symbol letters (or equivalently generalisations of rational cluster variables) to it. On the other hand, while degenerate scattering diagrams have been proposed as an analog of our method for selecting a finite subset of cluster variables with the help of tropical Grassmannians, a stumbling block is currently the significant ambiguity in their construction.\footnote{Note that the set of 26+26 rational letters that come as a byproduct of the scattering diagram analysis is too small to contain the 2-loop (N)MHV eight-particle amplitude.} It would be very interesting to further clarify the relation between the two approaches. While our discussion so far has been restricted to the $n=8$ case, we will shortly show that their similarity extends to any $n$: In particular, that the tropical square-root letters are always a subset of the scattering diagram square-root letters.

Let us also comment on the plausibility of the additional scattering diagram letters, the square roots of the Gram determinants, eqs.~\eqref{eq:Delta8} and~\eqref{eq:Delta8_2}, appearing as letters of the eight-particle amplitude. On the one hand, $\Delta_{1,3,5,7}$ and $\Delta_{2,4,6,8}$ are always positive inside the positive region~\cite{Arkani-Hamed:2019rds}, and so any arguments based on the expectation that amplitudes never have singularities in this region cannot exclude it. On the other hand, we observe that these letters are not present in explicit two-loop results for the (appropriately normalised) eight-particle amplitude in $\N=4$ pSYM. As an additional source of information on this question, one could also consider the relation between the alphabet of the latter, and that of five-particle amplitudes in Lorentz-invariant theories, recently established in~\cite{Chicherin:2020umh}. There, it was pointed out that while analogous square-root letters appear in individual integrals contributing to the two-loop five-point amplitudes, these cancel out in appropriately defined finite remainders, see also~\cite{Badger:2021nhg}. This analogy seems to suggest that at a minimum, $\Delta_{1,3,5,7}$ and $\Delta_{2,4,6,8}$ may contribute to eight-point integrals contributing to the $\N=4$ pSYM amplitude. Settling whether they survive in the final expression for the latter calls for explicit higher-loop computations, however already this discussion points to scattering diagrams as an attractive tool for studying singularities of Feynman integrals. Their potential in this respect will be studied elsewhere~\cite{CHHPZ}.

Finally, it is interesting to note that the discrete symmetry of the eight-particle alphabet respects some of the structure of the infinite cluster algebra of $\G{4,8}$. In particular, the group of automorphisms of the origin quiver, which can be traced back to the group of automorphisms of the initial quiver, is given by the transformation
\be
x_2\leftrightarrow x_4\,,\quad x_3\leftrightarrow x_7\,,\quad x_6\leftrightarrow x_8\,.
\ee
By the general theory~\cite{Assem:2010}, this quiver automorphism extends to an automorphism of the entire (infinite) cluster algebra. When replacing
\be
x_i\to x^0_{\gamma_i}\,,
\ee
in the above equation, this is also a symmetry of the square-root letters. Furthermore, it can be easily verified that the rational part of the alphabet is also symmetric under the same transformation, implying that the truncation procedure as well as the procedure by which we obtained the square-root letters from the scattering diagram are compatible with this symmetry of the infinite cluster algebra. Note that this symmetry is specific to the eight-particle alphabet since this automorphism only exists for $\G{4,8}$.

\noindent {\bf Comparison of algebraic letters at any multiplicity.} We now proceed to show that the tropical square-root letters of eq.~\eqref{eq:algLet} are contained in the alphabet obtained from the scattering diagram approach at any multiplicity $n$. In particular this adds further support to our analysis of the $n=9$ case in the next section, which has been carried out relying on the aforementioned equation.

For simplicity, let us start by considering an $\operatorname{A}^{(1)}_1$ cluster algebra with principal coefficients. The cone variables along the infinite mutation sequence are given by
\begin{equation}
	x_{\gamma_1;j} = \left(x_{1;j}\right)^{1-j}\left(x_{2;j}\right)^{-j}\,, \quad x_{\gamma_2;j} = \left(x_{1;j}\right)^{j}\left(x_{2;j}\right)^{1+j}\,.
\end{equation}
We can now use that $x_{1;j} = a_{2;j}^{-2}y_{1;j}$ and $x_{2;j} = a_{1;j}^{2}y_{2;j}$ to express the cone variables in terms of the $\A$-variables and coefficients along the sequence. Due to working with principial coefficients, it can be shown that $(y_{1;j})^{1-j}(y_{2;j})^{-j} = y_{1;0}$ and $(y_{1;j})^{j}(y_{2;j})^{1+j} = y_{2;0}$ such that we can use eq.~\eqref{eq:A11Sol1} to perform the limit $j\to\infty$, which is given by
\begin{equation}
	\label{eq:asymCVar}
	x_{\gamma_1}^+ \equiv x_{\gamma_1;\infty} = y_{1;0}\left(\tilde{C}_-\right)^{-2}\,,\quad x_{\gamma_2}^+ \equiv x_{\gamma_2;\infty} = y_{2;0}\left(C_+\right)^2 \,,
\end{equation}
where $C_\pm$ and $\tilde{C}_\pm$ have been defined in eqs.~\eqref{eq:A11C} and~\eqref{eq:A11CT}, respectively. These are the variables attached to the asymptotic chamber, which is the cone asymptotically close to the limit ray that we land in when following the infinite mutation sequence in this direction. Note that from the aforementioned equations it follows that these variables are actually algebraic functions in the $\X$-variables $x_{1;0},x_{2;0}$ of the initial cluster only, since $x_{1;0} = a_{2;0}^{-2}\cdot y_{1;0}$ and $x_{2;0} = a_{1;0}^2\cdot y_{2;0}$. 

We can now use wall-crossing to obtain the variables of the asymptotic chamber accessed by following the other direction of the mutation sequence, that is by repeatedly mutating $a_{2;j}$. The function associated to the \emph{limiting wall} of the scattering diagram that separates the two asymptotic chambers accessed by following the two directions of the mutation sequence is given by
\begin{equation}
	\label{eq:wallCrossing}
	f(x_{\gamma^\perp}) = y_{2;0}\frac{\left(C_-\right)^2\tilde{C}_-}{\tilde{C}_+} \equiv \frac{1}{y_{1;0}}\frac{\left(\tilde{C}_-\right)^2C_-}{C_+}\,,
\end{equation}
such that the variables of the other asymptotic chamber can be obtained from
\begin{equation}
	\label{eq:asymCVar2}
	x_{\gamma_1}^+ \, \longrightarrow \, x_{\gamma_1}^- = x_{\gamma_1}^+\cdot f(x_{\gamma^\perp})^2\,, \qquad 	x_{\gamma_2}^+ \, \longrightarrow \, x_{\gamma_2}^- = x_{\gamma_2}^+\cdot f(x_{\gamma^\perp})^{-2}\,.
\end{equation}

From eqs.~\eqref{eq:asymCVar}--\eqref{eq:asymCVar2}, we see that the four variables $x_{\gamma_i}^\pm$ associated to the two asymptotic chambers are made up of the four coefficients $C_\pm,\tilde{C}_\pm$ as well as the initial coefficients $y_{1;0},y_{2;0}$, which are monomials in the rational cluster variables. It then follows immediately that the algebraic letters of eqs.~\eqref{eq:algLet} are given as monomials in the multiplicative basis formed by these four variables. To be precise, we have $(\phi_0)^2 = (x_{\gamma_1}^+x_{\gamma_1}^-)^{-1}$ and $(\tilde{\phi}_0)^2 = x_{\gamma_2}^+x_{\gamma_2}^-(x_{\gamma_1}^+x_{\gamma_1}^-)^{2}$. An equivalent statement also holds true when considering general coefficients, since the generalized versions of eqs.~\eqref{eq:asymCVar}--\eqref{eq:asymCVar2} only differ by a further monomial in the rational variables.

\section{$\pTPTC{4,9}$ and the nine-particle alphabet}
\label{sec:NineParticle}
In this section, we apply the techniques first introduced in~\cite{Henke:2019hve, Drummond:2019cxm, Drummond:2019qjk}, and further developed in the previous sections, in order to obtain predictions for the symbol alphabet of the nine-particle amplitude in $\N=4$ pSYM. In subsection~\ref{sec:NineParticleRational}, we first truncate the infinite $\G{4,9}$ cluster algebra with the help of the inherently finite partially tropicalised positive configuration space $\pTPTC{4,9}$ as reviewed in  section~\ref{sec:BackgroundGrConf}, in order to obtain the rational part of the alphabet, which we find consists of 3,078 letters in one-to-one correspondence to tropical rays of $\pTPTC{4,9}$. Then, in subsection~\ref{sec:NineParticleAlgebraic} we study infinite mutation sequences of $\operatorname{A}^{(1)}_1$ subalgebras of the truncated cluster algebra, and in this fashion determine an additional 324 limit rays of $\pTPTC{4,9}$. It is especially here, that our new results for such subalgebras with general coefficients, presented in subsections~\ref{sec:A11Theory} and~\ref{sec:A11Application}, allow us to associate to these rays square-root letters expected to appear in the amplitude, in particular a total of 2,349 multiplicatively independent such letters. A new feature of the nine-particle case is that the procedure we have described falls short of yielding 27 rays of $\pTPTC{4,9}$. The discussion of alternative ways for accessing these rays, and of their possible significance for amplitudes, are presented in the next section.

Before moving on, let us briefly recall the discrete symmetries of $\N=4$ pSYM amplitudes, which will be useful in what follows.  Using the supersymmetry of the theory and combining the amplitudes with different external states into a superamplitude, the latter can be shown to be invariant under the transformations of the \emph{dihedral group}~\cite{Elvang:2009wd}. This symmetry group consists of the $n$ \emph{cyclic permutations} $i \to i + 1$ of the integer indices of the Plücker variables, which is equivalent to the cyclic permutation of the columns of the $4\times n$ matrix describing $\G{4,n}$, as well as the \emph{dihedral flip} $i \to n + 1 - i$. For MHV amplitudes, the aforementioned symmetries straightforwardly carry over to the transcendental functions appearing in them, and hence also to their alphabet. Note that in the case discussed here $n=9$ and throughout the text the identification $i + n \sim i$ for the indices of the Plücker variables is implied.

\subsection{Rational letters from the truncated cluster algebra}
\label{sec:NineParticleRational}
Similarly to the eight-particle case, we start mutating from the initial cluster of $\G{4,9}$, which is depicted in figure~\ref{fig:gr49Seed}. 
\begin{figure}[ht]
	\centering
	\begin{tikzpicture}[scale=1.3]
	\node at (0,0) (a1) {$\pl{1235}$};
	\node at (1.5,0) (a2) {$\pl{1245}$};
	\node at (3,0) (a3) {$\pl{1345}$};
	\node at (4.5,0) [draw, rectangle] (a4) {$\pl{2345}$};
	
	\node at (0,-1) (b1) {$\pl{1236}$};
	\node at (1.5,-1) (b2) {$\pl{1256}$};
	\node at (3,-1) (b3) {$\pl{1456}$};
	\node at (4.5,-1) [draw, rectangle] (b4) {$\pl{3456}$};
	
	\node at (0,-2) (c1) {$\pl{1237}$};
	\node at (1.5,-2) (c2) {$\pl{1267}$};
	\node at (3,-2) (c3) {$\pl{1567}$};
	\node at (4.5,-2) [draw, rectangle] (c4) {$\pl{4567}$};
	
	\node at (0,-3) (d1) {$\pl{1238}$};
	\node at (1.5,-3) (d2) {$\pl{1278}$};
	\node at (3,-3) (d3) {$\pl{1678}$};
	\node at (4.5,-3) [draw, rectangle] (d4) {$\pl{5678}$};
	
	\node at (0,-4) [draw, rectangle] (e1) {$\pl{1239}$};
	\node at (1.5,-4) [draw, rectangle] (e2) {$\pl{1289}$};
	\node at (3,-4) [draw, rectangle] (e3) {$\pl{1789}$};
	\node at (4.5,-4) [draw, rectangle] (e4) {$\pl{6789}$};
	
	\node at (-1.5,1) [draw, rectangle] (f) {$\pl{1234}$};
	
	\draw[->] (f) -- (a1);
	\draw[->] (a1) -- (b1);
	\draw[->] (b1) -- (c1);
	\draw[->] (c1) -- (d1);
	\draw[->] (d1) -- (e1);
	
	\draw[->] (a2) -- (b2);
	\draw[->] (b2) -- (c2);
	\draw[->] (c2) -- (d2);
	\draw[->] (d2) -- (e2);
	
	\draw[->] (a3) -- (b3);
	\draw[->] (b3) -- (c3);
	\draw[->] (c3) -- (d3);
	\draw[->] (d3) -- (e3);
	
	\draw[->] (a1) -- (a2);
	\draw[->] (a2) -- (a3);
	\draw[->] (a3) -- (a4);
	
	\draw[->] (b1) -- (b2);
	\draw[->] (b2) -- (b3);
	\draw[->] (b3) -- (b4);
	
	\draw[->] (c1) -- (c2);
	\draw[->] (c2) -- (c3);
	\draw[->] (c3) -- (c4);
	
	\draw[->] (d1) -- (d2);
	\draw[->] (d2) -- (d3);
	\draw[->] (d3) -- (d4);
	
	\draw[->] (b2) -- (a1);
	\draw[->] (c3) -- (b2);
	\draw[->] (d4) -- (c3);
	\draw[->] (b3) -- (a2);
	\draw[->] (b4) -- (a3);
	\draw[->] (c2) -- (b1);
	\draw[->] (c4) -- (b3);
	\draw[->] (d2) -- (c1);
	\draw[->] (d3) -- (c2);
	\draw[->] (e2) -- (d1);
	\draw[->] (e3) -- (d2);
	\draw[->] (e4) -- (d3);
	
	\end{tikzpicture}
	\caption{Initial seed of the cluster algebra of $\G{4,9}$. The boxed variables are frozen and hence not mutated.}
	\label{fig:gr49Seed}
\end{figure}
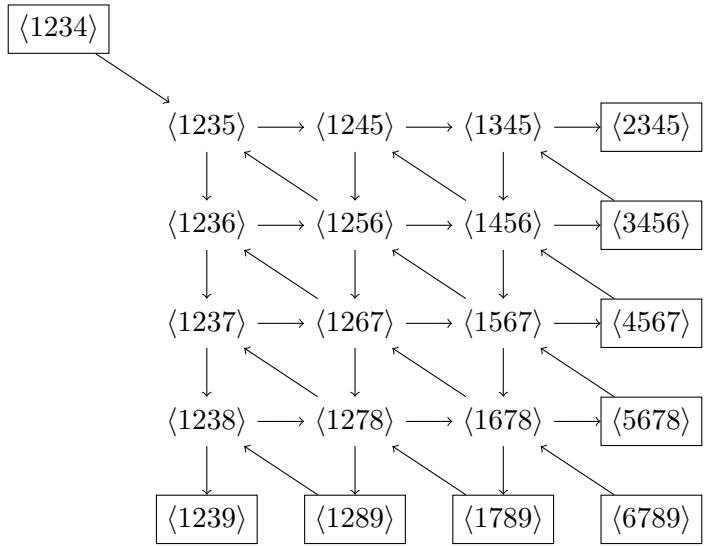
To tame the infinity of the cluster algebra we stop mutating in a given direction whenever the result of this mutation is a cluster variable whose associated ray is redundant with respect to the totally positive (partially) tropicalised configuration space -- that is, the ray does not lie on a maximal intersection of tropical hypersurfaces. By this truncation procedure, we obtain a finite subset of the infinite cluster algebra and thus a finite collection of $\A$-variables.

For computational purposes, we computed the truncated cluster algebra in two steps. First, we performed the aforementioned finite number of mutations only on the adjacency matrix, eq.~\eqref{eq:BMutation}, and cluster rays, eqs.~\eqref{eq:CMutation}--\eqref{eq:rayMutation}. Being only matrix operations, this can be done much more efficiently than factoring the rational expressions of the variables. Having computed the truncated cluster fan, we scanned it for paths of mutations connecting the initial cluster seed with a cluster containing a given ray. Due to the one-to-one correspondence of $\A$-variables and rays, we finally mutated the variables along these paths to obtain all cluster variables.

We find that the truncated cluster algebra obtained from the partially tropicalised totally positive configuration space $\pTPTC{4,9}$ contains 3,078 rational $\A$-variables in 24,102,954 clusters. These variables are all homogeneous polynomials in the Plücker variables of degree up to 6, see table \ref{tab:nineParticleLetterDegrees}. For comparison, we have also carried out the same truncation procedure for the full positive tropical configuration space $\TPTC{4,9}$, this time finding 12,645 $\A$-variables distributed in 55,363,988 clusters, whose multiplicity per degree are also listed in the same table.
\begin{table}[h]
	\centering
	\begin{tabular}{|l|c|c|c|c|c|c|c|c|c|c|c|}
		\hline \hline
		Degree & $1$ & $2$ & $3$ & $4$ & $5$ & $6$ & $7$ & $8$ & $9$ & $10$ & Total \\ \hline \hline
		$\pTPTC{4,9}$ & $117$ & $576$ & $1287$ & $963$ & $126$ & $9$ & - & - & - & - & $3078$ \\ \hline
		$\TPTC{4,9}$ & $117$ & $576$ & $1854$ & $3159$ & $2943$ & $1926$ & $1296$ & $531$ & $180$ & $63$ & $12645$ \\ \hline
		\hline
	\end{tabular}
	\caption{Number of $\A$-variables of the truncated cluster algebra of $\G{4,9}$ grouped by their (homogeneous) polynomial degree in the Plücker variables. At degree 1, all Plücker variables appear when also including the frozen  $\langle i\, i+1\, i+2\, i+3\rangle$ variables.}
	\label{tab:nineParticleLetterDegrees}
\end{table}

As already mentioned in section~\ref{sec:BackgroundTropConf}, existing data and symmetry reasons point to $\pTPTC{4,9}$ as the minimal choice relevant for scattering amplitudes, so unless otherwise stated we will be focusing on the latter. Explicitly, the $\A$-coordinates of degree up to three that make up the rational part of our candidate nine-particle alphabet are schematically given by\footnote{In this list, all types of letters are to be read as disjoint sets, such that e.g. the 8 classes of type $\pl{(\bar{i})\cap(jkl)\cap(mno)\cap(pqr)}$ are meant to not include those of type $\pl{(\bar{i})\cap(jkl)\cap(\bar{m})\cap(nop)}$, etc. We use notations where $(\bar{a})$ corresponds to the plane $(a-1\,a\,a+1)$ in momentum twistor space, the intersection of a line and a plane is given by $\pl{I(ab)\cap(cde)J} = \pl{IaJ}\pl{bcde} + \pl{IbJ}\pl{cdea}\,$, and the intersection of two planes by $\pl{I(abc)\cap(def)J} = \pl{IabJ}\pl{cdef} + \pl{IbcJ}\pl{adef} + \pl{IcaJ}\pl{bdef}\,$ for appropriate index sets $I, J \subset \{1,\dots,9\}$. See e.g.~\cite{Arkani-Hamed:2010pyv} for more details.}
\begin{itemize}
	\item all single Plücker variables $\pl{ijkl}$,
	\item 64 cyclic (37 dihedral) classes of degree two consisting of $\pl{1i(jkl)\cap(mno)}$ with $i\in\left\{1,2,3,5,6,7\right\}$,
	\item 143 cyclic (74 dihedral) classes of degree three consisting of
	\begin{itemize}[label=$\circ$]
		\setlength\itemsep{0.5em}
		\item 8 classes of type $\pl{(\bar{i})\cap(jkl)\cap(mno)\cap(pqr)}$ with $i\in\left\{2,3,4,5,7\right\}$, 
		
		35 of type $\pl{(\bar{i})\cap(jkl)\cap(\bar{m})\cap(nop)}$ with $2\leq i \leq 7$ and $4\leq m \leq 9$,
		
		20 of type $\pl{(\bar{i})\cap(jkl)\cap(\bar{m})\cap(\bar{n})}$ with $i\in\left\{1,2,4,5\right\}$, $4\leq m \leq 7$, $6\leq n \leq 9$,
		
		1 of type $\pl{(\bar{2})\cap(\bar{7})\cap(\bar{5})\cap(\bar{9})}$,

		\item 40 classes of type $\pl{i(12)\cap(jkl)(mno)\cap(pqr)}$ with $3 \leq i \leq 9$,
		
		2 of $\pl{4(13)\cap(896)(895)\cap(\bar{6})}$ and $\pl{8(13)\cap(\bar{5})(549)\cap(679)}$,
		
		2 of $\pl{3(18)\cap(\bar{8})(245)\cap(267)}$ and $\pl{7(18)\cap(235)(236)\cap(\bar{5})}$,

		\item 30 classes of type $\pl{i(12)\cap(klm)(no)\cap(pqr)s}$ with $o=n+1$,
		
		2 of $\pl{4(12)\cap(\bar{8})(68)\cap(\bar{2})5}$ and $\pl{5(79)\cap(\bar{2})(12)\cap(\bar{4})6}$,
		
		2 of $\pl{4(13)\cap(\bar{8})(67)\cap(\bar{2})5}$ and $\pl{6(89)\cap(\bar{2})(13)\cap(\bar{5})7}$,	
		
		1 of $\pl{8(14)\cap(\bar{6})(56)\cap(\bar{3})9}$.
	\end{itemize}
\end{itemize}
We have included these, as well as the remaining higher-degree letters, in the ancillary file~\texttt{Gr49RationalAlphabet.m} attached to the \texttt{arXiv} submission of this article, where the precise ranges of the indices not stated in the text may be found as well. Note that the representation is not unique due to the many ways to equivalently express these polynomials via the Plücker identities.

In total, the rational part of our candidate nine-particle alphabet consists of 3,087 $\A$-coordinates forming 3,078 dual conformally invariant letters, arranged in 342 cyclic classes which always have multiplicity 9. If one also considers the dihedral flip transformation, the alphabet consists of 9, 37, 74, 57, 7, and 1 dihedral classes of rational letters of degree 1 to 6, respectively (the multiplicity of dihedral classes may be either 9 or 18, depending on whether a flip relates two cyclic classes or maps one back to itself). Consequently, the proposed rational alphabet is dihedrally invariant. 

Let us conclude this section with some further comparisons and remarks on the structure of the rational part of our candidate alphabet. First of all, we can compare it with the explicit results for the symbol of the two-loop NMHV nine-particle amplitude, computed recently in~\cite{He:2020vob}. The latter contains 99 square-root letters, which we will discuss in the next section, as well as  522 dual conformally invariant rational letters that are polynomials in the Plücker variables of degree up to three. We find that all of these rational letters are indeed contained in our alphabet, which serves as a first consistency check. Note that the 216 nine-particle MHV letters~\cite{CaronHuot:2011ky} are all rational and contained in the NMHV ones, and hence our proposal trivially covers this helicity configuration as well.

To make the comparison more precise,  the alphabet of~\cite{He:2020vob} contains all degree one letters of the alphabet proposed here except $\pl{1357}$ plus cyclic permutations. Furthermore, several cyclic classes of higher degree letters are missing compared to our alphabet. Among those are 27 cyclic classes of degree two letters, 27 cyclic classes of the type $\pl{(\bar{i})\cap(jkl)\cap(\bar{m})\cap(nop)}$ as well as all but one of the 30 cyclic classes of the type $\pl{i(12) \cap (klm)(no) \cap (pqr)s}$ our proposal includes. Letters of the type $\langle i(jk)\cap(lmn)(opq) \cap(rst)\rangle$ are not at all contained in the alphabet of the two-loop NMHV nine-particle amplitude. Additional evidence in support of our proposal for the rational part of the alphabet is that it agrees with the one obtained by other means in~\cite{Ren:2021ztg}, appearing simultaneously with this article.

Next, in search of interesting patterns, we may look at what part of the infinite $\G{4,9}$ cluster algebra is chosen by our tropical selection rule according to the degree of the $\A$-variable with respect to the Pl\"ucker variables. As shown in~\cite{Chang:2019}, $\G{4,9}$ contains 576 cluster variables of degree two, 2,421 of degree three, 8,622 of degree four, and 27,054 variables of degree five. A comparison with table~\ref{tab:nineParticleLetterDegrees} demonstrates that our alphabet contains all possible quadratic cluster $\A$-variables but only a subset of those of degree three or higher.\footnote{The same statement for degrees two and three in fact holds also for the eight-particle rational alphabet we proposed in~\cite{Henke:2019hve}, which contains all 120 quadratic but not all 174 cubic cluster variables of $\G{4,8}$~\cite{Chang:2019,Mago:2020nuv}} For example, the polynomials
\begin{equation}
\pl{4(56)\cap(\bar{8})(78)\cap(\bar{2})1}\,,\quad\pl{(127)\cap(\bar{5})\cap(13i)\cap(\bar{8})}\,,
\end{equation}
with $i=5$ or $i=6$, are cluster $\A$-variables of degree three but are not associated to tropical rays of $\pTPTC{4,9}$, and hence are not selected by our procedure. 

Also, the new data point we have achieved affords us the possibility to also study the maximal Pl\"ucker degree of the letters as a function of the multiplicity $n$: Considering the cluster algebra truncated by $\pTPTC{4,n}$, we observe that the maximum degree of letters for $n=6,7,8$, and $9$ is given by $1,2,3$, and $6$ for the same values of $n$, which interestingly matches the first few values of the sequence
\begin{equation}
d_{\max}\left(n\right) =
\begin{pmatrix}
n-5\\
\lfloor\frac{n-5}{2}\rfloor
\end{pmatrix}\,,
\end{equation}
where the brackets denote the binomial coefficient and $\lfloor x \rfloor$ is the floor function. 
For $\TPTC{4,n}$, the same count is $1,2,5,$ and $10$, which agrees with $(n-6)^2+1$. Given that the subset of $\A$-variables of fixed degree in the infinite cluster algebras of $\G{4,n}$ with $n\geq 8$ can be computed by other means~\cite{Chang:2019}, knowing the maximal degree of those that are chosen by our tropical selection rule could thus provide a more direct means for their determination also at higher $n$.

As the approach we develop in this paper in principle applies to any $n$, let us conclude this section with some further general predictions. More precisely, the truncation procedure for selecting a finite subset of $\G{4,n}$ cluster variables, as predictions for the rational part of the symbol alphabet, is algorithmic (and similarly for our infinite mutation sequence analysis yielding predictions for the square-root letters). While the constructive determination of this finite subset has to be done separately for each value of $n$, and its computational complexity increases with $n$, the initial cluster of $\G{4,n}$ will always be selected. Hence the cluster variables it contains, namely the Plücker variables
\begin{equation}
	\pl{1234},\pl{123i},\pl{12i-1\,i},\pl{1i-2\,i-1\,i},\pl{i-3\,i-2\,i-1\,i}\quad\text{with}\quad5\leq i\leq n\,,
\end{equation}
as well as their dihedral images (since our truncation procedure respects dihedral symmetry) will always be included in our prediction for the rational part of the symbol of the $n$-particle amplitude.

\subsection{Square-root letters from infinite mutation sequences}
\label{sec:NineParticleAlgebraic}
Having obtained a candidate for the rational part of the alphabet of nine-particle amplitudes in the previous subsection, here we will enlarge it so as to also include square-root letters. The general procedure for doing so has been presented in section~\ref{sec:A11}, and relies on considering infinite mutation sequences starting from any origin cluster of the truncated cluster algebra, that contains an $\operatorname{A}^{(1)}_1$ subalgebra. As with the eight-particle case, discussed in section~\ref{sec:A11Application}, in the first instance we examine the limit of cluster rays along the sequence, and only select it, along with its associated square-root letters (or generalised cluster variables) if this limit coincides with a tropical ray of $\pTPTC{4,9}$. 

This step therefore requires knowledge of all tropical rays. While these may be obtained with the help of dedicated software such as \texttt{polymake}~\cite{polymake:2000}, fortunately most of the work has already been done in~\cite{He:2020ray}. There, the dual fan of the tropical configuration space $\TPTC{4,9}$~\cite{Speyer2005,WilliamsAmplitudes2020}, namely the Minkowski sum of the Newton polytopes $\mathcal{P}\left(4,9\right)$ obtained from the web-parameterisation of the Plücker variables, has been computed. Since all $\pTPTC{4,9}$ rays are contained in $\TPTC{4,9}$, we may thus extract them from the provided $\mathcal{P}\left(4,9\right)$ data. As a technical remark, in this data the rays (or facets, in the language of the dual polytope) are provided in the space of $D=126$  Plücker coordinates $\pl{ijkl}$, however it is easy to obtain its rays in the space of the $d=12$ variables of the (tropicalised) web-parameterisation we review in appendix~\ref{sec:WebParam}: One needs to simply express the Plücker coordinates in terms of the $d$ web-parameters, and solve for the latter after equating the former to the value of the $D$-dimensional rays.

In this fashion, we find that out of the 19,395 $\TPTC{4,9}$ rays, a total of 3,429 is contained in the partially tropicalised, totally positive configuration space $\pTPTC{4,9}$. We have already seen that 3,078 of these rays are associated to $\A$-variables in the truncated cluster algebra thus corresponding to rational letters. Next, we scan the cluster algebra of $\G{4,9}$ truncated by $\pTPTC{4,9}$ and identify a total of 549,180 origin clusters, from which we obtain the limit rays by numerically evaluating eq.~\eqref{eq:rayMutation} for a sufficient number of mutations along the infinite sequence. This yields another 324 rays of $\pTPTC{4,9}$ that are not in the truncated cluster algebra.

After having obtained the limit rays, the next step is to associate two square-root letters to each origin cluster they can arise from, according to eq.~\eqref{eq:algLet}. As already pointed out, the merit of the analysis of $\operatorname{A}^{(1)}_1$ sequences with general coefficients, that we carried out in subsection~\ref{sec:A11Theory}, is that we can immediately obtain the square-root letters in question by simply plugging in the $\X$-coordinates of the $\operatorname{A}^{(1)}_1$ subalgebra of a given origin cluster in eq.~\eqref{eq:A11_invariants}. It is important to bear in mind, however, that many of these letters are identical, since mutating an origin cluster at a node not connected to the subalgebra will not change the data relevant for the limit. Furthermore, the resulting distinct square-root letters are not all multiplicatively independent, such that not all of them are required to describe the symbol of the amplitude. 

To obtain a basis of these letters, we adopt the following approach, which is similar to that of ref.~\cite{Mitev:2018kie}. Given any set of letters, we first express them in terms of 12 independent variables, for example with the help of the web-parameterisation. We then evaluate these variables at some prime values, and find multiplicative relations of numerically evaluated letters by sampling all possible such relations with a fixed total integer power.\footnote{Note that once we find a relation in this numeric evaluation, we can also verify it symbolically.} Having reduced the letters to some smaller set, say of size $m$, we can verify whether no more relations exist by evaluating the logarithm of the letters at $m$ different evaluation points. In this logarithmic form, multiplicative relations among the letters correspond to integer-coefficient linear relations. Hence, if the rank of the $m\times m$ matrix formed by the evaluations of the letters is maximal, no further relations can exist.

For the case of the square-root letters obtained from $\operatorname{A}^{(1)}_1$ sequences, we find that there is a one-to-one correspondence between the radicand of the square-root and the limit ray. This implies that there can only be multiplicative relations among letters obtained from sequences with the same limit ray. In total, we find $2,349$ multiplicatively independent square-root letters in 36 cyclic (21 dihedral) classes associated to the $324$ limit rays. Arranged according to the number of multiplicatively independent sets of letters per ray, or equivalently per radicand, they consist of
\begin{itemize}
	\item 6 cyclic (3 dihedral) classes of sets of 5 multiplicatively independent letters,
	\item 8 cyclic (4 dihedral) classes of sets of 6 independent letters,
	\item 8 cyclic (4 dihedral) classes of sets of 7 independent letters,
	\item 6 cyclic (4 dihedral) classes of sets of 8 independent letters,
	\item 2 cyclic (1 dihedral) classes of sets of 9 independent letters,
	\item 5 cyclic (4 dihedral) classes of sets of 10 independent letters, and
	\item 1 cyclic (1 dihedral) class of a set of 11 independent letters,
\end{itemize}
Since the explicit expressions for these letters are quite complicated, we will refrain from quoting them here, and instead provide them in the ancillary file~\texttt{Gr49AlgebraicAlphabet.m} attached to the \texttt{arXiv} submission of this article.  We now briefly comment on some properties of this alphabet.

First of all, comparing the square-root letters presented in this article to those of the two-loop nine-particle NMHV amplitude reported in~\cite{He:2020vob}, we find that the $9\times 11$ letters put forward in the aforementioned reference precisely correspond to the last cyclic class of $11$ multiplicatively independent letters of our proposed non-rational alphabet. Together with a similar analysis we carried out in the previous section for the rational letters, this implies that we obtain the entire two-loop nine-particle (N)MHV alphabet as part of our approach, which thus passes a quite nontrivial consistency check.

Furthermore, as can be seen from the above presentation of the square-root letters, the alphabet is invariant under the dihedral transformations. The dihedral flip transformation maps a cyclic class with a given number of multiplicatively independent letters to another class of the same size, whereas the letters obtained from $\phi_0$, eq.~\eqref{eq:algLet}, get mapped to those obtained from $\tilde{\phi}_0$ and vice versa.

Let us also comment on the structure of square roots appearing in our algebraic letters. We can rewrite $\phi_0$ and $\tilde{\phi}_0$ into the form of $(a_i\pm\sqrt{a_i^2-4b_i})/2$. As we saw in subsection~\ref{sec:A11Application}, in the case of eight-particle amplitudes the radicand $\Delta = a_i^2-4b_i \equiv K_1^2-4K_2$ is always proportional to one of the square-roots of the eight-point four-mass boxes $\Delta_{1,3,5,7}$ and $\Delta_{2,4,6,8}$, see in particular eqs.~\eqref{eq:Delta8} and~\eqref{eq:Delta8_2}, and e.g.~\cite{Bourjaily:2013mma} for more details on the four-mass boxes. However, in the non-rational alphabet for nine-particle amplitudes suggested here, we find that only the radicands of the last cyclic class of $11$ independent letters each are proportional to the square-roots of the nine-point four-mass boxes, $\Delta_{1,3,5,7}$ and its cyclic permutations (a total of nine), which are the square-root singularities obtained from the Landau analysis at two loops~\cite{Prlina:2017tvx}. For example, we also find square-root letters whose radicand is given by
\begin{align}
	\Delta &\propto A^2-4B\,,\quad\text{with}\\
	A &= 1 - \frac{\pl{6789}\pl{13\left(278\right)\cap\left(246\right)}^2}{\pl{1235}\pl{1289}\pl{3567}\pl{1679}^2} + \frac{\pl{1267}\pl{23\left(146\right)\cap\left(178\right)}\pl{46\left(278\right)\cap\left(129\right)}}{\pl{1235}\pl{1289}\pl{3567}\pl{1679}^2}\,,\\
	B &= \frac{\pl{1267}\pl{23\left(146\right)\cap\left(178\right)}\pl{46\left(278\right)\cap\left(129\right)}}{\pl{1235}\pl{1289}\pl{3567}\pl{1679}^2}\,,
\end{align}
which is not proportional to one of the four-mass boxes. Attributing the additional square roots we find to particular integrals is a very interesting question we leave for future work.

As a further check of the nine-particle singularities we have obtained, the Landau equations also predict that the branch points $b_i = 0$ correspond to the zero loci of some rational letters. And indeed, we confirm that this holds for all the square-root letters of our candidate non-rational alphabet for nine-particle scattering. In fact, this is a general property for all square-root letters obtained by the prescription of eq.~\eqref{eq:algLet}, independent of the particle number $n$, as we now show. Rewriting these in the form discussed in the previous paragraph, we obtain $b_1 = -x_{1;0}$ for $\phi_0$ and $b_2=-x_{1;0}^2x_{2;0}$ for $\tilde{\phi}_0$, whereas $x_{1;0}, x_{2;0}$ are the $\X$-variables corresponding to the $\operatorname{A}^{(1)}_1$ cluster subalgebra in the origin clusters. Since the $\X$-variables are monomials in the cluster $\A$-variables, the branch points $b_i=0$ of all square-root letters are the zero loci of some letters from the rational alphabet.

To summarise, we have obtained 3,078 rational and 2,349 square-root letters. Whereas the rational letters are in one-to-one correspondence to tropical rays of $\pTPTC{4,9}$, the square-root letters are associated to a total of $324$ tropical rays, in sets containing between $5$ and $11$ letters per ray. This is very similar to the eight-particle case, where 9 square-root letters are associated to each of the two limit rays. 

A great qualitative difference between the eight- and the nine-particle case is that we no longer obtain all tropical rays of $\pTPTC{4,9}$ from the Grassmannian cluster algebra by selection or an $\operatorname{A}^{(1)}_1$ mutation sequence: In particular we can access 3,402 out of the 3,429 such rays in this manner, so we fall short of 27 $\pTPTC{4,9}$ rays. Given the great jump in complexity  between the $\G{4,8}$ and $\G{4,9}$ cluster algebras,\footnote{In particular, while the cluster algebras of both $\G{4,8}$ and $\G{4,9}$ are infinite, the former one is of \emph{finite mutation type}, implying that it consists of only a finite number of different quivers, see e.g.~\cite{CAIV,fomin2006cluster}.} perhaps the real surprise is not that we cannot access all rays by our method, but that the number of rays we cannot access is so small.

Nevertheless, in the next section we will explore more general infinite mutation sequences of $\operatorname{A}^{(1)}_m$ Dynkin type as a possible means for obtaining the missing rays, as well as touch on their implications for amplitude singularities. Before concluding, it's also worth mentioning that in its current state of development, neither the scattering diagram approach~\cite{Herderschee:2021dez} that we discussed in detail in subsection \ref{sec:SDAlphabet} can solve the mystery of the 27 missing rays, as it too relies on infinite mutation sequences starting from within the cluster algebra.

\section{One generalisation of infinite mutation sequences}
\label{sec:Am1}
For the case of eight-particle scattering we have seen that all rays of $\pTPTC{4,8}$ can be obtained from the $\G{4,8}$ cluster algebra with the help of infinite mutation sequences of type $\operatorname{A}^{(1)}_1$, however an analogous statement is not true for nine-particle scattering. While the square-root letters associated to the accessible rays agree with explicit two-loop results for both multiplicities, it remains unclear whether the missing rays contribute additional amplitude singularities in the latter case.

For this reason, in this section we will study a generalisation of the considered mutation sequences aiming to access the missing rays. The starting point is the observation that the property allowing us to associate a sequence to a cluster algebra is the periodicity of its quivers \cite{Fordy:2009qz}. A quiver is said to be \emph{cluster-mutation periodic} of period $p$, if there is a sequence of $p$ mutations resulting in a quiver isomorphic to the initial one. Consider for example the quiver of the $\operatorname{A}^{(1)}_1$ cluster algebra, fig.~\ref{fig:mutSequ}, which has period one since mutating any of its two nodes flips the double-arrow resulting in the same quiver up to relabelling the nodes. 

If a quiver is periodic in this sense, we can repeat the same mutation infinitely many times thus giving rise to an infinite mutation sequence. Since the rational functions, eqs.~\eqref{eq:AMutationRule} and~\eqref{eq:YMutationRule}, that realise the mutation are determined by the quiver, the periodicity allows us to write down a recurrence relation for all clusters along the sequence,
\begin{equation}
	a_{i;j+1} = M_i\left(a_{1;j},\dots,a_{r;j};y_{1;j},\dots,y_{r;j}\right)\,.
\end{equation}
Unlike for generic cluster algebras, due to the periodicity of the quiver, the rational function $M_i$ does not depend on $j$ but is the same for all clusters along the sequence.

Note that in general, while we may repeat the mutation infinitely many times, the resulting sequence may be periodic, that is consisting of only a finite set of different $\A$-variables. This is the case whenever the considered cluster algebra is finite, like for example for the cluster algebra of $A_2$ Dynkin type, whose quiver is given by two nodes connected by one arrow. The corresponding cluster algebra is finite with five clusters and five variables.

Periodic clusters have been studied and classified in~\cite{Fordy:2009qz} for period one and two. Using this perspective, in this section we consider cluster algebras of $\operatorname{A}^{(1)}_m$ Dynkin type for $m\in\mathbb{N}$, which are the largest class of period one \emph{primitives}, the building blocks of all period one cluster algebras. To the best of our knowledge the analysis of their infinite mutation sequences is new, and thus may be of intrinsic mathematical interest irrespective of the question of the missing rays. Subsection~\ref{sec:Am1Theory} works them out in analogy to subsection~\ref{sec:A11Theory}, subsection~\ref{sec:Am1Application} explores the possibility of using these sequences to obtain algebraic letters beyond the $\operatorname{A}^{(1)}_1$ singularities, and subsection~\ref{sec:triangulation} discusses the inherent limitations of accessing the limit rays from the $\G{4,n}$ cluster algebra for $n\ge 9$. We again refer to appendix~\ref{sec:Am1Proofs} for details of the proofs which are omitted in the main text.

\subsection{Mutation sequences in $\operatorname{A}^{(1)}_m$ with general coefficients}
\label{sec:Am1Theory}
The cluster algebras of $\operatorname{A}^{(1)}_m$ Dynkin type are rank-$(m+1)$ cluster algebras whose eponymous quivers are depicted in figure~\ref{fig:AM1Seq}. As can be easily seen, mutating at either the source or sink (ie. the node of $a_{1;j}$ or $a_{m+1;j}$) leads to the same quiver with the labels of the nodes rotated clockwise by one position. In this section, we will discuss the repeated mutation at the source, that is we always mutate $a_{1;j}$. For the other direction, see appendix~\ref{sec:Am1Proofs}.
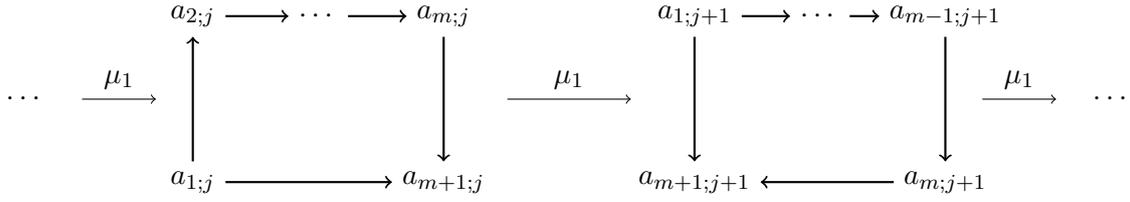
\begin{figure}[ht]
	\centering
	\begin{tikzpicture}[scale=1.1]
		\node at (0,1) (g) {$\cdots$};
		
		\node at (2,1) (h) {};
		\node at (2,0) (i) {$a_{1;j}$};
		\node at (2,2) (j) {$a_{2;j}$};
		\node at (3.5,2) (k) {$\cdots$};
		\node at (5,2) (l) {$a_{m;j}$};
		\node at (5,0) (m) {$a_{m+1;j}$};
		\node at (5,1) (n) {};
		
		\node at (8,1) (o) {};
		\node at (8,0) (p) {$a_{m+1;j+1}$};
		\node at (8,2) (q) {$a_{1;j+1}$};
		\node at (9.5,2) (r) {$\cdots$};
		\node at (11,2) (s) {$a_{m-1;j+1}$};
		\node at (11,0) (t) {$a_{m;j+1}$};
		\node at (11,1) (u) {};
		
		\node at (13,1) (v) {$\cdots$};

		\draw[->,thick] (i) -- (j);
		\draw[->,thick] (i) -- (m);
		\draw[->,thick] (j) -- (k);
		\draw[->,thick] (k) -- (l);
		\draw[->,thick] (l) -- (m);
		
		\draw[->,thick] (q) -- (p);
		\draw[->,thick] (t) -- (p);
		\draw[->,thick] (q) -- (r);
		\draw[->,thick] (r) -- (s);
		\draw[->,thick] (s) -- (t);
		
		\draw[->,shorten >=10pt,,shorten <=10pt] (g) edge node[above] {$\mu_1$} (h);
		\draw[->,shorten >=20pt,,shorten <=20pt] (n) edge node[above] {$\mu_1$} (o);
		\draw[->,shorten >=10pt,,shorten <=10pt] (u) edge node[above] {$\mu_1$} (v);
		
	\end{tikzpicture}
	\caption{Clusters $j$ and $j+1$ along the considered mutation sequence of the cluster algebra of $\operatorname{A}^{(1)}_m$ Dynkin type. The coefficients are omitted in the figure.}
	\label{fig:AM1Seq}
\end{figure}

The $\operatorname{A}^{(1)}_m$ cluster algebras can also be considered as arising from the surface $A(m,1)$, which is the annulus with $m$ marked points on the outer and one marked point on the inner boundary. In this formalism, a cluster of the cluster algebra corresponds to a triangulation of this surface, in which each arc connecting two marked points (or one with itself) corresponds to a variable of the cluster, see fig.~\ref{fig:A21Surface} for an example or~\cite{fomin2006cluster} for details on this relation.

\begin{figure}[ht]
	\centering
	\includegraphics[width=0.25\textwidth]{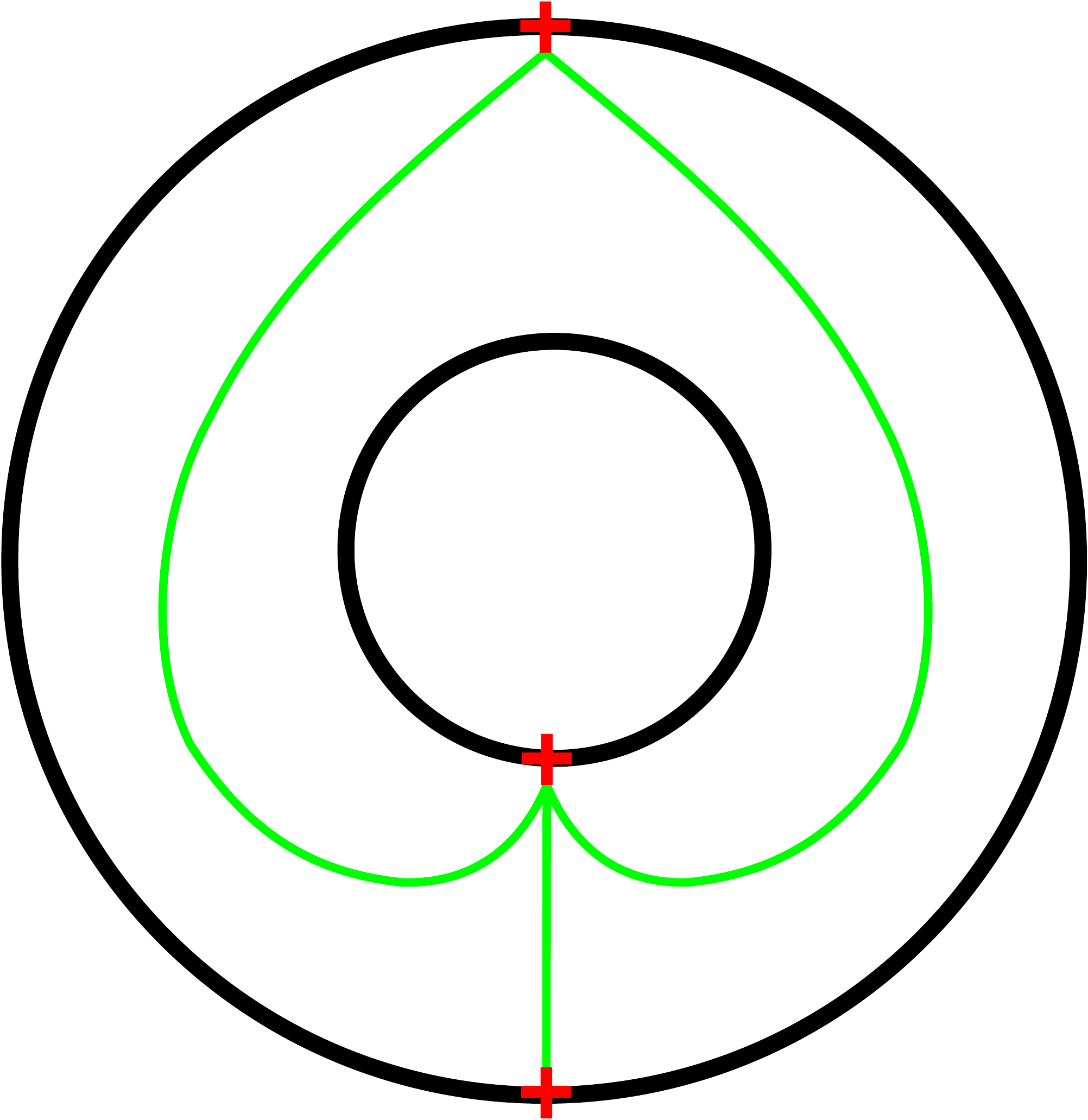}
	\caption{Annulus with two marked points on the outer and one marked point on the inner boundary. Corresponds to the initial cluster of the $\operatorname{A}^{(1)}_2$ cluster algebra.}
	\label{fig:A21Surface}
\end{figure}

The cluster algebras of $\operatorname{A}^{(1)}_m$ Dynkin type are cluster-mutation periodic with period one. From the geometric perspective of the annulus we described, this periodicity corresponds to winding the arcs around the inner boundary, which can be considered as a redundancy in the cluster algebra \cite{NimaEdinburgh2020, NimaAmplitudes2020}. We may remove this redundancy by replacing it with quantities that are invariant under mutation or, geometrically speaking, a winding by $2\pi$ along the sequence. For $m=1$ these invariants are the building blocks of the square-root letters for eight- and nine-particle scattering.

Consider now the infinite mutation sequence depicted in fig.~\ref{fig:AM1Seq}, which starts at the initial cluster $j=0$. We can immediately write down the mutation relations for the variables and coefficients along this sequence by applying eqs.~\eqref{eq:AMutationRule} and~\eqref{eq:YMutationRule} to the depicted clusters. They are given for any $j\in\Z$ by
\begin{gather}
	a_{m+1;j+1} = \frac{a_{2;j}a_{m+1;j}}{a_{1;j}}\frac{1+x_{1;j}}{1\cltrad y_{1;j}} \,, \label{eq:Am1ASequence} \\
	y_{1;j+1} = \frac{y_{2;j}\,y_{1;j}}{\left(1\cltrad y_{1;j}\right)}\,,\quad y_{m;j+1} = \frac{y_{m+1;j}\,y_{1;j}}{\left(1\cltrad y_{1;j}\right)}\,,\quad y_{m+1;j+1} = \left(y_{1;j}\right)^{-1} \,, \label{eq:Am1YSequence} \\
	x_{1;j+1} = \frac{x_{2;j}\,x_{1;j}}{\left(1 + x_{1;j}\right)}\,,\quad x_{m;j+1} = \frac{x_{m+1;j}\,x_{1;j}}{\left(1 + x_{1;j}\right)}\,,\quad x_{m+1;j+1} = \left(x_{1;j}\right)^{-1} \,. \label{eq:Am1XSequence}
\end{gather}
Note that as in the case of $\operatorname{A}^{(1)}_1$ sequences (and, in fact, for any cluster algebra) the mutation rule of the $\X$-variables is that of the coefficients with tropical addition changed to normal addition. Also similar to the case of $\operatorname{A}^{(1)}_1$, we used that the $\X$-variable associated to $a_{1;j}$ is given by $x_{1;j} = \left(a_{2;j}a_{m+1;j}\right)^{-1}y_{1;j}$ to arrive at eq.~\eqref{eq:Am1ASequence}. The other variables and coefficients are not mutated but only shifted in their first index by
\begin{align}
	a_{i;j+1} &= a_{i+1;j}\quad\text{for }i\neq m+1\,, \label{eq:Am1ARels}\\
	y_{i;j+1} &= y_{i+1;j}\quad\text{for }i\notin\left\{1,m,m+1\right\}\,, \label{eq:Am1YRels}\\
	x_{i;j+1} &= x_{i+1;j}\quad\text{for }i\notin\left\{1,m,m+1\right\}\,. \label{eq:Am1XRels}
\end{align}

In~\cite{Fordy:2009qz}, infinite mutation sequences of type $\operatorname{A}^{(1)}_m$ without coefficients\footnote{The cluster algebra without coefficients can be obtained from that with general coefficients by setting the initial coefficients to one. Due to the mutation relations, the coefficients and $(1\cltrad y_{1;j})$ will then be equal to one in every cluster and will not influence the other mutation relations.} were analysed by linearising the recursion relations. Following this approach, which we have also used in section~\ref{sec:A11}, we first introduce the sequence $\beta_j$ of ratios, which is defined as
\begin{equation}
	\beta_{j} = \frac{a_{m+1;j}}{a_{1;j}}\,.
\end{equation}
Again this is the ratio of the sink-variable over the source-variable. Furthermore, we can again express $\beta_j$ also as $a_{1;j+m}/a_{1;j}$ and other equivalent ways by using eq.~\eqref{eq:Am1ARels}. We also define the auxiliary sequence $\gamma_j$ as
\begin{equation}
	\gamma_j = 1\cltrad y_{1;j} \cltrad y_{1;j}\left(y_{1;j-m}\right)^{-1}\,.
\end{equation}
Having defined the generalisations of $\beta_j$ and $\gamma_j$ for any $m$, we continue by also defining the two quantities
\begin{align}
	K_{1,j} &= \left(\gamma_0\gamma_{j}^{-1}\beta_{0}^{-1}\beta_{j}\right)\left[1+x_{1;j} + x_{1;j}\left(x_{1;j-m}\right)^{-1}\right] \,, \label{eq:Am1K1j}\\
	K_{2,j} &= \left(\gamma_0\gamma_{j}^{-1}\beta_{0}^{-1}\beta_{j}\right)\left(\gamma_0\gamma_{j-m+1}^{-1}\beta_{0}^{-1}\beta_{j-m+1}\right)\left[x_{1;j}\left(x_{1;j-m}\right)^{-1}\right]\,, \label{eq:Am1K2j}
\end{align}
whereas we inserted the factor $\gamma_0\beta_0^{-1}$ to normalise these quantities at $j=0$ for later convenience. We also may express $K_{1,j}$, $K_{2,j}$ and $\gamma_j$ in terms of the variables of cluster $j$ only. Using the mutation rules in the opposite direction, ie. mutating at $a_{m+1;j}$, see appendix~\ref{sec:Am1Proofs}, we can express $x_{1;j-m}$ in terms of cluster $j$ as
\begin{equation}
	\label{eq:x1Rel}
	\left(x_{1;j-m}\right)^{-1} = x_{2;j}\left(1+x_{3;j}\left(1+\cdots x_{m;j}\left(1+x_{m+1;j}\right)\right)\right)\,.
\end{equation}
Equivalently, we obtain the same relation for $y_{1;j-m}$ by replacing $x_{i;j}$ by $y_{i;j}$ and addition by cluster-tropical addition. Using the original definition, eq.~\eqref{eq:Am1K1j}, and the mutation rules, eqs.~\eqref{eq:Am1ASequence}--\eqref{eq:Am1XSequence}, it can be shown that $K_{1;j}$ and $K_{2;j}$ are in fact invariant along the sequence and are thus given in terms of the initial cluster variables as
\begin{align}
	K_1 &\equiv K_{1,0} = 1 + x_{1;0}\left(1+x_{2;0}\left(1+x_{3;0}\left(1+\cdots x_{m;0}\left(1+x_{m+1;0}\right)\right)\right)\right)\,, 	\label{eq:Am1K1}\\
	K_2 &\equiv K_{2,0} = x_{1;0}x_{2;0}\left(1+x_{3;0}\left(1+\cdots x_{m;0}\left(1+x_{m+1;0}\right)\right)\right)\,. 	\label{eq:Am1K2}
\end{align}

Using these invariants we can linearise the recursion relation~\eqref{eq:Am1ASequence}. As before, the homogeneous, linear recurrence obtained when considering the sequence $a_{1;j}$, given by
\begin{equation}
	\label{eq:Am1LinearARec}
	\gamma_j^{-1}\gamma_{j+m}^{-1}a_{1;j+2m} - \gamma_0^{-1}\beta_0K_1\cdot \gamma_j^{-1} a_{1;j+m} + \gamma_0^{-2}\beta_0^2K_2\cdot a_{1;j} = 0\,,
\end{equation}
does not have constant coefficients. Hence, we define the new variable $\alpha_j$ for $j\geq 0$ by
\begin{equation}
	\alpha_j = \gamma^{-1}_{j\,\text{mod}\,m}\gamma^{-1}_{(j\,\text{mod}\,m) + m}\cdots\gamma_{j-2m}^{-1}\gamma_{j-m}^{-1}\cdot a_{1;j}\,.
\end{equation}
Note that this is to be read that for $0\leq i < m$ we have $\alpha_{m+i} = \gamma_i^{-1}\cdot a_{1;m+i}$, and $\alpha_{2m+i} =\gamma_i^{-1}\gamma_{m+i}^{-1}\cdot a_{1;2m+i}$, and so on. In terms of this sequence, the recurrence is given by
\begin{equation}
	\label{eq:Am1AlphaRec}
	\alpha_{j+2m} - \gamma_0^{-1}\beta_0K_1 \cdot \alpha_{j+m} + \gamma_0^{-2}\beta_0^2 K_2 \cdot \alpha_j = 0\,,
\end{equation}
with initial values $\alpha_0,\dots,\alpha_{2m-1}$, which in turn can be expressed in terms of the variables and coefficients of the initial cluster via the mutation relations, eqs.~\eqref{eq:Am1ASequence}--\eqref{eq:Am1XSequence}. Observing that $\alpha_{j+m}/\alpha_j=\gamma_j^{-1}\beta_j$ and that $\gamma_j\to 1$ for $j\to \infty$, as proven in appendix~\ref{sec:Am1Proofs}, we see from this recurrence that, assuming convergence, the respective limit of $\beta_j$ is obtained as the solution of the equation
\begin{equation}
	\beta^2 - \gamma_0^{-1}\beta_0K_1\cdot \beta + \gamma_0^{-2}\beta_0^2K_2 = 0\,,
\end{equation}
which has the two solutions $\beta_\pm$ given by
\begin{equation}
	\beta_\pm = \beta_0\frac{K_1 \pm \sqrt{K_1^2 - 4K_2}}{2\gamma_0}\,.
\end{equation}

Similar to before, we now turn to discussing the solution of the recurrence~\eqref{eq:Am1AlphaRec}, using standard methods based on its characteristic polynomial,
\begin{equation}
	P_m(t) = t^{2m} - \gamma_0^{-1}\beta_0K_1\cdot t^m + \gamma_0^{-2}\beta_0^2K_2\,.
\end{equation}
Its $2m$ roots are given by $\beta_\pm^{1/m}\eta_m^{i}$ for $i=0,\dots,m-1$ and whereas $\eta_m$ is the $m$-th root of unity. To see this, note that we may first solve for the roots in terms of $t^m$ resulting in $t^m=\beta_\pm$. Accordingly, the most general solution to the recurrence is given by
\begin{equation}
	\alpha_j = \left[c_0^+ + c_1^+\eta_m^j + \cdots + c_{m-1}^+\eta_m^{(m-1)j}\right]\left(\beta_+\right)^{\frac{j}{m}} + \left[c_0^- + c_1^-\eta_m^j + \cdots + c_{m-1}^-\eta_m^{(m-1)j}\right]\left(\beta_-\right)^{\frac{j}{m}}\,.
\end{equation}
The $2m$ coefficients $c_i^\pm$ for $i=0,\dots,m-1$ can be obtained from the initial values $\alpha_0,\dots,\alpha_{2m-1}$ and can thus ultimately be expressed in terms of the quantities of the initial cluster. Note that since $\eta_m$ is $m$-periodic, the overall coefficients multiplying $\left(\beta_\pm\right)^{j/m}$, denoted by $C_\pm(j)$, only depend on $(j\,\text{mod}\,m)$, implying that they assume a total of $m$ different values.

\subsection{Beyond $\operatorname{A}^{(1)}_1$ singularities?}
\label{sec:Am1Application}
Having obtained the general solution of the infinite mutation sequences of type $\operatorname{A}^{(1)}_m$, let us now discuss how we could attribute algebraic letters, or generalised cluster variables, to their rays.

Similar to the discussion of the $m=1$ case, if we were to take the direction of approach to the ray into account, then the $m\ge 1$ analog of eq.~\eqref{eq:algLet} would mean to assign $2m$ letters defined as $C_+(i)/C_-(i)$ and $\tilde{C}_+(i)/\tilde{C}_-(i)$ for $i=0,\dots,m-1$ to each ray. As is discussed in appendix~\ref{sec:Am1Proofs}, these are given by
\begin{align}
	\phi_i &\equiv \frac{C_+(i)}{C_-(i)} = \left(\frac{K_1 - \sqrt{K_1^2-4K_2}}{K_1 + \sqrt{K_1^2-4K_2}}\right)^{i/m}\frac{2F_i - K_1 + \sqrt{K_1^2-4K_2}}{-2F_i + K_1 + \sqrt{K_1^2-4K_2}}\,, \label{eq:Am1Phi}\\
	\tilde{\phi}_i &\equiv \frac{\tilde{C}_+(i)}{\tilde{C}_-(i)} = \left(\frac{K_1 - \sqrt{K_1^2-4K_2}}{K_1 + \sqrt{K_1^2-4K_2}}\right)^{i/m}\frac{2\tilde{F}_iK_2 - K_1 + \sqrt{K_1^2-4K_2}}{-2\tilde{F}_iK_2 + K_1 + \sqrt{K_1^2-4K_2}} \label{eq:Am1PhiTilde}\,,
\end{align}
for $i=0,\dots,m-1$ and whereas $F_i$ and $\tilde{F}_i$ are rational functions of the $\X$-variables of the initial cluster $j=0$, which are given by
\begin{equation}
	\tilde{F}_i = \frac{K_{m+1-i}}{K_{m+1}}\,,\qquad\qquad
	F_i = 
	\begin{cases}
	1\quad&\text{if}\,\,i=0\,,\\
	K_1 - K_{i+1}\quad&\text{otherwise.}
	\end{cases}
\end{equation}
In the definition of these rational functions, we used $K_i$, which is a generalisation of the invariants, eqs.~\eqref{eq:Am1K1} and~\eqref{eq:Am1K2}, and defined for $1 < i \leq m+1$ as
\begin{equation}
K_i = x_{1;0}\cdots x_{i;0}\left(1+x_{i+1;0}\left(1+\cdots x_{m;0}\left(1+x_{m+1;0}\right)\right)\right)\,.
\end{equation}
Note that since $F_0 = \tilde{F}_0 = 1$, the expressions~\eqref{eq:Am1Phi} and~\eqref{eq:Am1PhiTilde} simplify for $i=0$ to those of eq.~\eqref{eq:algLet}, evaluated with the generalised invariants of eqs.~\eqref{eq:Am1K1} and~\eqref{eq:Am1K2}.

However, the non-rational letters obtained from the above formulas for $m > 1$ qualitatively differ from those with $m=1$. For $m=1$ we observe a one-to-one association of the radicand, $K_1^2-4K_2$, to the limit ray of the sequence. Since such non-rational letters with the same radicand can have multiplicative relations among each other, this allows the reduction of the letters associated to any such ray to a smaller, multiplicatively independent set. For $m>1$, however, this is no longer true, as we observe that these radicands are different for every origin cluster irrespective of the limit ray, implying the multiplicative independence of all such letters. Considering for example just the letters obtained from eqs.~\eqref{eq:Am1Phi} and~\eqref{eq:Am1PhiTilde} with $m=2$ and $i=0$, one would obtain 2912 additional, multiplicatively independent square-root letters for eight-particle scattering, where all $\pTPTC{4,8}$ rays have already been determined.\footnote{We find sequences with up to $m=4$ in the cluster algebra of $\G{4,8}$ truncated by $\pTPTC{4,8}$.}

The fact this large number of additional letters is not encountered in the existing amplitude computations, seems to suggest their irrelevance. We stress again that they arise as a generalisation of the prescription of eq.~\eqref{eq:Am1Phi}, which may not be applicable to higher $m$. Nevertheless, we find it interesting that it is possible to obtain additional algebraic letters in this fashion.

\subsection{The limitations of infinite mutation sequences}
\label{sec:triangulation}

Finally, let us turn to the question of whether infinite mutation sequences more general than $\operatorname{A}^{(1)}_1$ can account for the missing $\pTPTC{4,9}$ rays. Scanning the approximately 24 million clusters of the cluster algebra of $\G{4,9}$ truncated by $\pTPTC{4,9}$, we find that they contain cluster subalgebras of type $\operatorname{A}^{(1)}_m$ with up to $m=5$. Unfortunately, however, the limit rays of these sequences are only a subset of the $324$ limit rays of $\operatorname{A}^{(1)}_1$ ones. In addition, we have checked all primitives of period one with rank up to 6\footnote{In the notations of~\cite{Fordy:2009qz}, these correspond to the quivers labelled by $P^{(2)}_i$ for $i=4,5,6$ and $P^{(3)}_6\,$.} and found that they also do not account for these missing rays. 

These results are in fact in line with an inherent limitation on accessing all tropical rays from within the cluster algebra, as we will now discuss. For infinite cluster algebras, e.g. for those of the Grassmannians with $k=4, n\geq 8$, the cluster fan is not complete (see eg.~\cite[Remark 3.2]{Reading2018a}). This means that the fan does not cover the entire ambient space $\R^d$, with $d$ the rank of the cluster algebra, and thus also cannot triangulate the entire $\pTPTC{k,n}$. An example for such an infinite cluster algebra of rank two, also considered in~\cite{Cordova:2013bza,2016arXiv160300416B,Reading2018b}, is depicted in figure~\ref{fig:infCAEx}.
\begin{figure}[ht]
	\centering
	\begin{tikzpicture}[scale=1.3]
		\node at (0,0) (a) {$a_1$};
		\node at (1,0) (b) {$a_2$};
		
		\draw[->] ([yshift=3pt]a.east) -- ([yshift=3pt]b.west);
		\draw[->] ([yshift=-3pt]a.east) -- ([yshift=-3pt]b.west);
		\draw[->] (a) -- (b);
		
	\end{tikzpicture}
	\caption{Example for an infinite cluster algebra.}
	\label{fig:infCAEx}
\end{figure}

The fan of this cluster algebra is two-dimensional with the rays of the initial clusters being the canonical unit vectors. In it, there are two infinite mutation sequences -- repeatedly mutating at either the sink or the source of the quiver -- which converge to two different rays, as depicted in fig.~\ref{fig:fanGaps}. The two-dimensional gap of the fan is also clearly visible in this figure. If a cluster algebra contains such an algebra as a subalgebra, its fan is expected to also be incomplete.
\begin{figure}[ht]
	\centering
	\includegraphics[width=0.3\textwidth]{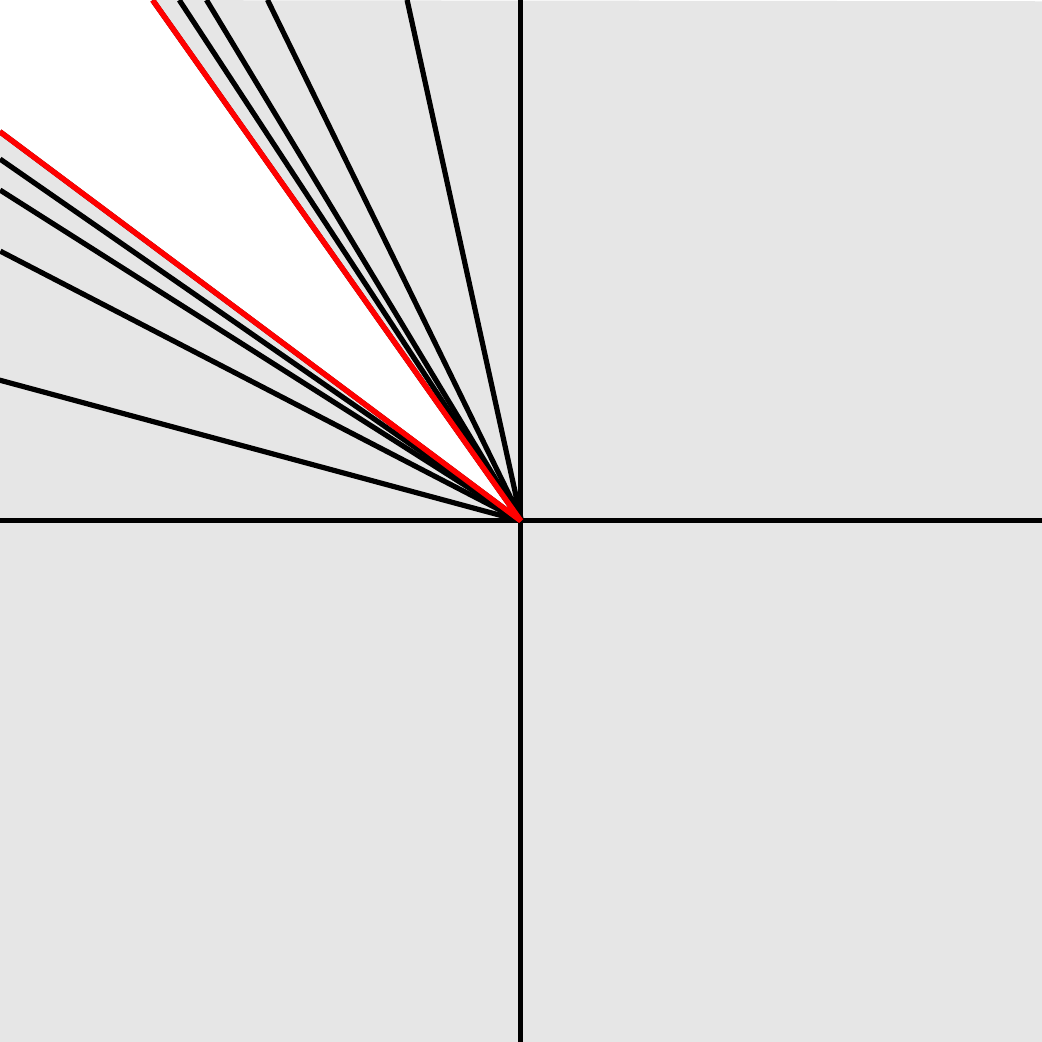}
	\caption{Sketch of the fan of the cluster algebra with two nodes connected by three arrows. The cluster algebra is infinite with two infinite sequences approaching the two rays highlighted in red.}
	\label{fig:fanGaps}
\end{figure}

In the case of eight particles, the truncated cluster algebra only contains clusters with infinite mutation sequences whose fans leave one-dimensional gaps, which could be taken care of by including the limit ray of the sequence. This is no longer the case for nine particles, since the truncated cluster algebra also contains clusters with nodes connected by three arrows, such as the one depicted in fig.~\ref{fig:kroneckerQuiverInNineParticles}.
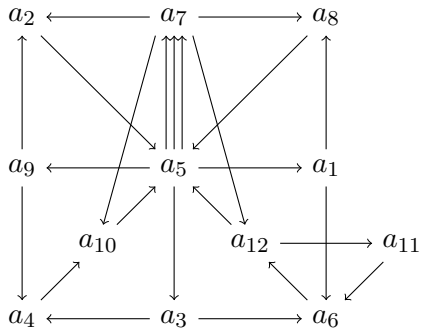
\begin{figure}[ht]
	\centering
	\begin{tikzpicture}[scale=1]
		\node at (0,1) (b) {$a_7$};
		\node at (-2,1) (b1) {$a_2$};
		\node at (2,1) (b2) {$a_8$};

		\node at (0,-1) (a) {$a_5$};
		\node at (-2,-1) (a1) {$a_9$};
		\node at (2,-1) (a2) {$a_1$};
		
		\node at (-1,-2) (c1) {$a_{10}$};
		\node at (1,-2) (c2) {$a_{12}$};
		
		\node at (0,-3) (d) {$a_3$};
		\node at (-2,-3) (d1) {$a_4$};
		\node at (2,-3) (d2) {$a_6$};
		
		\node at (3,-2) (e) {$a_{11}$};
		
		\node at (-3,-2) (f) {};
		
		\draw[->] ([xshift=3pt]a.north) -- ([xshift=3pt]b.south);
		\draw[->] ([xshift=-3pt]a.north) -- ([xshift=-3pt]b.south);
		\draw[->] (a) -- (b);
		
		\draw[->] (b) -- (b1);
		\draw[->] (b) -- (b2);
		\draw[->] ([xshift=-7pt]b.south) -- (c1);
		\draw[->] ([xshift=7pt]b.south) -- (c2);
		
		\draw[->] (a) -- (a1);
		\draw[->] (a) -- (a2);
		\draw[->] (b1) -- (a);
		\draw[->] (b2) -- (a);
		\draw[->] (c1) -- (a);
		\draw[->] (c2) -- (a);
		
		\draw[->] (a1) -- (b1);
		\draw[->] (a2) -- (b2);
		
		\draw[->] (a1) -- (d1);
		\draw[->] (a2) -- (d2);
		\draw[->] (d1) -- (c1);
		\draw[->] (d2) -- (c2);
		
		\draw[->] (a) -- (d);
		\draw[->] (d) -- (d1);
		\draw[->] (d) -- (d2);
		
		\draw[->] (e) -- (d2);
		\draw[->] (c2) -- (e);
	\end{tikzpicture}
	\caption{Example of a cluster in the truncated cluster algebra of $\G{4,9}$ containing nodes connected by three arrows. The $\A$-variables $a_i$ correspond to certain rational nine-particle letters. The frozen variables are omitted in order to avoid clutter.}
	\label{fig:kroneckerQuiverInNineParticles}
\end{figure}

Due to the existence of such clusters, it is expected that the cluster fan for nine particles contains higher-dimensional gaps. This might suggest that (some of) the $27$ missing rays are located in the interior of such gaps, explaining why they could not be reached by any limiting procedure from within the cluster algebra. Note that the truncation of these infinite cluster algebras by the selection rule provided by the partially tropicalised positive configuration space creates further gaps in the cluster fans.

Having motivated an explanation for the inaccessibility of certain tropical rays starting from the $\G{4,n}$ cluster algebra, some the most important open questions that remain include whether there exist alternative ways for obtaining these rays, that also associate some form of generalised cluster variables to them, and whether the latter provide any further information on the singularities of amplitudes. Perhaps the inaccessibility of the missing rays of the cluster algebra is related to the appearance of functions beyond multiple polylogarithms in $\N=4$ pSYM $n$-particle amplitudes: Indeed, while it is known that such functions certainly appear at $n=10$~\cite{CaronHuot:2012ab}, the possibility that these in fact also appear at lower $n$ is currently not excluded. If this turns out to be true, then the appropriate generalisation of cluster algebras may go hand in hand with a corresponding generalisation of the notion of symbol letters along the lines of~\cite{Broedel:2018iwv}. We leave these exciting questions for future work.

\section{Conclusions \& Outlook}
\label{sec:Conclusion}

In this article, we have developed a general procedure for obtaining a finite collection of rational and square-root letters expected to appear in the symbol of $\N=4$ pSYM amplitudes for arbitrary multiplicity $n$, and we have concretely applied it for the first time to the case $n=9$. Our work builds on the earlier observation that the amplitude symbol letters coincide with the variables of the $\G{4,n}$ cluster algebra for $n=6,7$~\cite{Golden:2013xva}, and on the proposal for curing the infinity of the cluster algebra in question for $n\ge 8$ with the help of geometric objects known as (duals of) tropical Grassmannians~\cite{Drummond:2019qjk,Drummond:2019cxm,Arkani-Hamed:2019rds,Henke:2019hve}. In particular, focusing on the then first nontrivial case $n=8$, the latter papers showed that tropical Grassmannians select a finite subset of rational variables of the $\G{4,n}$ cluster algebra, as well as motivate the inclusion of certain generalisations of cluster variables that contain square roots, and are related to infinite mutation sequences of a rank-two ($\operatorname{A}^{(1)}_1$) subalgebra of the cluster algebra. 

The precise form of these generalisations of cluster variables, or equivalently square-root letters, also depends on certain $\G{4,n}$ cluster variables that appear as so-called coefficients of the rank-two subalgebra. Therefore in order for the aforementioned analysis to be applicable to arbitrary multiplicity $n$, it is necessary to work out $\operatorname{A}^{(1)}_1$ sequences with general coefficients. In this work we fill this gap, and in fact we study infinite mutation sequences of larger class of rank-($m+1$) cluster algebras, denoted as $\operatorname{A}^{(1)}_m$ in the affine Dynkin diagram classification, with general coefficients. As a cross-check of our results, after specialising to $m=1$ we first apply them to the known $\G{4,8}$ case, not only finding perfect agreement with the earlier proposal for the symbol alphabet of the eight-particle amplitude, but also comparing them with a more recent, alternative proposal based on the closely related approach of~\cite{Herderschee:2021dez}. Very interestingly, we find that the two approaches have a highly non-obvious, almost complete overlap in their predictions, the only additional letters provided by scattering diagrams being the two square roots associated to the four-mass box, eqs.~\eqref{eq:Delta8} and~\eqref{eq:Delta8_2}.

With the confidence gained by this comparison, we then move on to the main application of our results, the generation of new predictions for the symbol alphabet of the nine-particle amplitude with the help of cluster algebras and tropical geometry. First, our tropical selection rule picks a finite subset of $3,078$ $\G{4,9}$ cluster variables as a candidate for the rational part of the alphabet, arranged in over 24 million clusters. Then, the analysis of infinite rank-two mutation sequences with general coefficients contained in the aforementioned clusters yields another $2,349$ square-root letters expected to appear in the symbol. We have confirmed that our thus obtained collection of nine-particle letters passes all available consistency checks; namely it respects the discrete symmetries of the amplitude, it agrees with requirements on the position of branch points coming from the Landau equations, and it contains all letters found in an explicit 2-loop calculation of the NMHV nine-particle amplitude~\cite{He:2020vob}.

At the same time, our analysis reveals new qualitative features starting at $n\ge 9$, which call for further inquiry. Both the selected  $\G{4,n}$ cluster variables and their square-root generalisations are associated to building blocks of the tropical Grassmannian, known as tropical rays. While all of these could be accessed from the cluster algebra when also including $\operatorname{A}^{(1)}_1$ infinite mutation sequences for $n=8$, this is no longer the case at $n=9$, where 27 out of 3,429 tropical rays are left unaccounted for. In subsection~\ref{sec:triangulation} we presented evidence suggesting the existence of an obstruction independent of the type of infinite mutation sequence chosen, and commented on the potential physical significance of the missing rays for amplitude singularities.

As a complementary direction for addressing some of these open questions, it is interesting to note that all eight-particle square-root letters can be obtained by considering the Schubert problem~\cite{Arkani-Hamed:2010pyv} of the corresponding four-mass box kinematics in momentum twistor space~\cite{NimaPrivate}. Namely, one first considers the four non-intersecting lines formed by the pairs of momentum twistors parameterising these kinematics, and finds the two lines that intersect them. Then, the square-root letters turn out to correspond to cross ratios formed by the four points on any of these six lines. A generalisation of this analysis to nine points could provide yet another means of comparison with our results and provide hints for the (ir)relevance of the missing rays.

Independent of the latter question, our results also raise a practical issue: Predictions for the alphabet of an amplitude have been essential input for actually computing it via the cluster bootstrap programme, which has been very successful at multiplicity six and seven. This is achieved by first constructing the finite-dimensional function spaces expected to contain each loop correction to the amplitude, which arise as solutions to linear systems whose size and sparsity depends on the number of letters and their form as functions of the independent variables parameterising the kinematics, respectively. Based on an alphabet of (at least) 5,427 letters, some of which are for example polynomials with over 50,000 terms in the web-parameterisation,  the nine-particle amplitude bootstrap would pose a serious challenge to current linear algebra technology. Aside from evolutionary progress on the latter, could the size of linear systems be reduced by restrictions on the specific subsets of letters appearing at each slot in the symbol, stemming from adjacency/extended Steinmann relations \cite{Drummond:2017ssj,Caron-Huot:2019bsq} or the $\bar Q$-equation \cite{Caron-Huot:2011dec}? The integration of the latter has been the main source of explicit two-loop amplitude computations at multiplicity $n\ge 8$, is it feasible to push it to higher loops? Alternatively, the Wilson loop OPE predicts amplitudes at any multiplicity as an expansion around the collinear limit \cite{Alday:2010ku,Basso:2013vsa,Basso:2013aha,Papathanasiou:2013uoa,Basso:2014koa,Papathanasiou:2014yva,Basso:2014nra,Belitsky:2014sla,Belitsky:2014lta,Basso:2014hfa,Basso:2015rta,Basso:2015uxa,Belitsky:2016vyq}, and has been successfully evaluated~\cite{Papathanasiou:2013uoa,Papathanasiou:2014yva} and resummed~\cite{Drummond:2015jea,Cordova:2016woh,Lam:2016rel,Bork:2019aud} in the six-particle case. Could we hope for similar progress also for more legs, once the final ingredient of this approach, known as the matrix part, is better understood? It would be very interesting to address these questions in the future.

\section*{Acknowledgements}
The authors would like to thank the organizers of the workshop on \textit{Cluster Algebras and the Geometry of Scattering Amplitudes} for their hospitality as well as its participants for the insightful talks and discussions. We are grateful to Marcus Spradlin and  Jian-Rong Li for helpful correspondence, Ilke Canakci and Aidan Herderschee for valuable discussions, and Nima Arkani-Hamed for comments on the manuscript. We thank Dmitry Chicherin for pointing out the existence of further rational relations in the eight-particle square-root letters to us. NH is grateful for the support of the graduate programme of the Studienstiftung des deutschen Volkes. The authors acknowledge support from the Deutsche Forschungsgemeinschaft (DFG) under Germany’s Excellence Strategy EXC 2121 \emph{Quantum Universe} 390833306.

\appendix

\section{Proofs for mutation sequences of type $\operatorname{A}^{(1)}_m$}
\label{sec:Am1Proofs}
In this appendix we present the calculations and proofs required for the solution to the infinite mutation sequences of type $\operatorname{A}^{(1)}_m$ that have been omitted in the main text. We first discuss the source direction, that is the sequence obtained by repeatedly mutating $a_{1;j}$. To avoid repetition, we often point to the relevant formulas in the main text. Finally, we turn to the sink direction, the repeated mutation of $a_{m+1;j}$.

\subsection{The source direction}
The key observation in the discussion of the infinite mutation sequences of type $\operatorname{A}^{(1)}_m$ is the existence of the invariants, eqs.~\eqref{eq:Am1K1j} and~\eqref{eq:Am1K2j}. Before proving their invariance, we first establish that they and $\gamma_j$ can be written in terms of the quantities of cluster $j$ only, hence also proving eqs.~\eqref{eq:Am1K1} and~\eqref{eq:Am1K2}. For this, it suffices to prove eq.~\eqref{eq:x1Rel}, that is to express $x_{1;j-m}$ in terms of the variables of cluster $j$. To do so, we consider the mutation sequence depicted in figure~\ref{fig:AM1Seq} in reverse. Since mutation is an involution, we can go from cluster $j+1$ to cluster $j$ by mutating the former at node $m+1$. The relevant mutation relations are given by
\begin{gather}
	x_{m+1;j} = x_{m;j+1}\left(1+x_{m+1;j+1}\right)\,,\quad x_{1;j} = \left(x_{m+1;j+1}\right)^{-1}\,,\label{eq:Am1XReverse}\\
	x_{i;j} = x_{i-1;j+1}\quad\text{for }i\notin\left\{1,2,m+1\right\}\,.
\end{gather}
As can be seen from these relations, mutating from the cluster $j$ to the cluster $j-m$ along the sequence automatically results in a parmeterisation of $x_{1;j-m}$ in terms of the variables of cluster $j$. From these, we can immediately conclude that
\begin{equation}
	x_{m;j-m+i} = x_{m-1;j-m+i+1} = \cdots = x_{i;j}\,,
\end{equation}
for $2 \leq i \leq m$. With these relations in place, we can express the $\X$-variable $x_{1;j-m}$ in terms of the variables of cluster $j$ as
\begin{align}
	x_{1;j-m}^{-1} &= x_{m+1;j-m+1} = x_{m;j-m+2}\left(1+x_{m+1;j-m+2}\right) \nonumber\\
	&= x_{2;j}\left(1 + x_{m;j-m+3}\left(1+x_{m+1;j-m+3}\right)\right) = \cdots \label{eq:Am1XShift}\\
	&= x_{2;j}\left(1+x_{3;j}\left(1+\cdots x_{m;j}\left(1+x_{m+1;j}\right)\right)\right)\,,\nonumber
\end{align}
with the equivalent relation for $y_{1;j-m}$ again obtained by replacing the $\X$-variables by coefficients and addition by cluster-tropical addition. Having established this relation, we now turn to the prove of invariance. For this, consider the following lemma.

\begin{lemma}
	The two quantities $K_{1,j}$ and $K_{2,j}$, defined as
	\begin{align}
		K_{1,j} &= \left(\gamma_0\beta_{0}^{-1}\gamma_{j}^{-1}\beta_{j}\right)\left[1+x_{1;j} + x_{1;j}\left(x_{1;j-m}\right)^{-1}\right] \,, \\
		K_{2,j} &= \left(\gamma_0\beta_{0}^{-1}\gamma_{j}^{-1}\beta_{j}\right)\left(\gamma_0\beta_{0}^{-1}\gamma_{j-m+1}^{-1}\beta_{j-m+1}\right)\left[x_{1;j}\left(x_{1;j-m}\right)^{-1}\right]\,,
	\end{align}
	are invariant along the infinite mutation sequence, whose mutation relations are given by eqs.~\eqref{eq:Am1ASequence}--\eqref{eq:Am1XSequence}.
\end{lemma}
\begin{proof}
	First of all, from eqs.~\eqref{eq:Am1ASequence} it follows that the ratio $\beta_j$ changes as follows
	\begin{equation}
		\beta_{j+1} = \frac{a_{m+1;j+1}}{a_{1;j+1}} = \frac{a_{m+1;j}}{a_{1;j}}\frac{1+x_{1;j}}{1\cltrad y_{1;j}} = \frac{1+x_{1;j}}{1\cltrad y_{1;j}} \beta_j\,,
	\end{equation}
	whereas we used that $a_{1;j+1} = a_{2;j}$. On the other hand, it follows from eqs.~\eqref{eq:Am1XSequence} that
	\begin{align}
		x_{1;j+1}\left(x_{1;j-m+1}\right)^{-1} &= \frac{x_{1;j}x_{2;j}}{x_{1;j-m+1}\left(1+x_{1;j}\right)} = \frac{x_{1;j}x_{m;j-m+2}}{x_{1;j-m+1}\left(1+x_{1;j}\right)} \\
		&= \left(1+x_{1;j}\right)^{-1}\left(1+x_{1;j-m+1}\right)^{-1}\left[x_{1;j}\left(x_{1;j-m}\right)^{-1}\right]\,,
	\end{align}
	whereas we have used eqs.~\eqref{eq:Am1XRels} to write
	\begin{equation}
		\label{eq:x2jRel}
		x_{2;j} = x_{3;j-1} = \cdots = x_{m;j-m+2}\,.
	\end{equation}
	This also implies that
	\begin{equation*}
		1+x_{1;j+1} + x_{1;j+1}\left(x_{1;j-m+1}\right)^{-1} = \left(1+x_{1;j}\right)^{-1}\left[1+x_{1;j}+x_{1;j}\left(x_{1;j-m}\right)^{-1} \right]
	\end{equation*}
	As is a general property of cluster algebras, the corresponding relations for the coefficients can be obtained from those of the $\X$-variables by replacing them with the coefficients and addition by cluster-tropical addition. We thus also have that
	\begin{equation}
		\label{eq:Am1GammaMutation}
		\gamma_{j+1} = \left(1\cltrad y_{1;j}\right)^{-1} \gamma_j\,.
	\end{equation}
	This implies that $\gamma_{j+1}^{-1}\beta_{j+1} = (1+x_{1;j})\gamma_j^{-1}\beta_j$ such that the lemma follows. Note that this holds for all $j \in \mathbb{N}$.
\end{proof}

Having established the invariance of $K_1$ and $K_2$, we now turn to the limit of $\gamma_j$ and prove that it converges to $1$. In order to deal with the cluster-tropical addition, let us remind ourselves how the coefficients are related to the frozen variables. We consider the rank-$(m+1)$ cluster algebra of type $\operatorname{A}^{(1)}_m$ with $M$ frozen variables denoted by $z_i$ for $i=1,\dots,M$. In any cluster, the coefficients are given as a monomial in the frozen variables, see eq.~\eqref{eq:coeffs} and the mutation rule eq.~\eqref{eq:YMutationRule}, which demonstrates that this property holds in all clusters. We thus rewrite the coefficients as
\begin{equation}
	y_j = \prod_{i=1}^{M} z_i^{c^{i}_j}\,.
\end{equation}
Using that cluster-tropical addition, eq.~\eqref{eq:cltrad}, is defined on such monomials in the frozen variables, we rewrite the recursion relation of $y_{1;j}$, eqs.~\eqref{eq:Am1YSequence}, in terms of the new sequences $c^i_j$. For this, consider first that
\begin{equation}
	y_{1;j+1} = \frac{y_{1;j}y_{2;j}}{\left(1\cltrad y_{1;j}\right)} = \frac{y_{1;j}y_{1;j-m+1}}{y_{1;j-m}\left(1\cltrad y_{1;j}\right)\left(1\cltrad y_{1;j-m+1}\right)}\,,
\end{equation}
whereas we have used the equivalent of eq.~\eqref{eq:x2jRel} for the coefficients. This implies that the corresponding relation for the $c^i_j$ is given by
\begin{equation}
	c^i_{j+1} = c^i_j + c^i_{j-m+1} - c^i_{j-m} - \min\left(0,c^i_j\right) - \min\left(0,c^i_{j-m+1}\right)
\end{equation}
Using the notation $[x]_+ = \max(0,x) = -\min(0,-x)$ and $x = [x]_+ - [-x]_+$, this results in the recursion relation
\begin{equation}
	\label{eq:cRec}
	c^i_{j+1} + c^i_{j-m} = \left[c^i_j\right] + \left[c^i_{j-m+1}\right]_+\,.
\end{equation}
While the appearance of $[x]_+$ on the right hand side of this recurrence makes solving it analytically complicated, we can prove the following property of this sequence.
\begin{lemma}
	\label{lem:CProps}
	Fix $m\in\mathbb{N}_{>0}$ and consider the sequence $c_n$ for $n\geq 1$ with initial values $c_1,\dots,c_{m+1}\in\Z$ and recurrence relation
	\begin{equation}
		\label{eq:cProofRec}
		c_{n+1} + c_{n-m} = \left[c_n\right] + \left[c_{n-m+1}\right]_+\,.
	\end{equation}
	There exists a $N\in\mathbb{N}$ such that for all $n\geq N$
	\begin{equation}
		c_n \geq c_{n-m} \geq 0\,.
	\end{equation}
\end{lemma}
\begin{proof}
	To prove the lemma, we introduce an auxiliary sequence $\Delta_n$ defined by
	\begin{equation}
		\Delta_n = c_n - c_{n-m}\,.
	\end{equation}
	By continuing the recurrence for $c_n$ to $n\leq0$, we can use this definition for all $n\geq 1$. We now establish some key properties of these sequences.
	
	\textsc{Positivity/Negativity.} Assume there exists a $N\in\mathbb{N}$ such that $\Delta_n \geq 1$ for all $n\geq N$. By construction, we can write $c_n = \Delta_n + c_{n-m}$ for any $n$ and thus get for any $j \geq 0$ and $i=0,\dots,m-1$ that
	\begin{equation}
		c_{N+j\cdot m + i} = \sum_{l=1}^{j} \Delta_{N+l\cdot m+i} + c_{N+i} \geq j + c_{N+i}\,.
	\end{equation}
	Note that we have included the shift by $i$ because $N+j\cdot m + i = n$ for any $n\geq N$ and appropriate choice of $i$ and $j$. Hence, for $j\geq \max(0,-c_{N+i})$ we conclude that $c_{N+j\cdot m+i} \geq 0$. To summarize, this implies that if there is a $N\in\mathbb{N}$ such that $\Delta_n\geq 1$ for all $n\geq N$, then
	\begin{equation}
		c_n\geq 0\,,\quad\forall n\geq \max_{i=0,\dots,m-1}\left\{N+\max\left(0,-c_{N+i}\right)\cdot m + i\right\}\,.
	\end{equation}
	If we instead assume that there exists a $N\in\mathbb{N}$ such that $\Delta_n \leq -1$ for all $n\leq N$, we get by the same reasoning as before that
	\begin{equation}
		c_n < 0\,,\quad\forall n\geq \max_{i=0,\dots,m-1}\left\{N+\max\left(0,1+c_{N+i}\right)\cdot m + i\right\}\,.
	\end{equation}
	
	\textsc{Monotonicity.} From the recursion relation of $c_n$, eq.~\eqref{eq:cProofRec}, we obtain a corresponding relation for the sequence $\Delta_n$, which is given by
	\begin{equation}
		\Delta_{n+1} = \Delta_n + \left[-c_n\right]_+ + \left[-c_{n-m+1}\right]\,,
	\end{equation}
	whereas we used $x = [x]_+ - [-x]_+$ to arrive at this result. Since $[x]_+ \geq 0$, this relation implies that
	\begin{equation}
		\label{eq:cProofMon}
		\Delta_{n+1} \geq \Delta_n\,,
	\end{equation}
	that is, $\Delta_n$ is a monotonically increasing sequence. Further to that, the relation for $\Delta_{n+1}$ also gives us the following \emph{extended monotonicity property}
	\begin{equation}
		\label{eq:cProofDEqu}
		\Delta_{n+1} = \Delta_n \iff c_n \geq 0 \,\wedge\, c_{n-m+1}\geq 0\,,
	\end{equation}
	which follows because $[-x]_+$ is positive and zero if and only if $x$ is positive.
	
	\textsc{Boundedness.} We now prove that $\Delta_n$ is bounded. For this, assume $\Delta_n$ to not be bounded. Since $\Delta_n$ is a monotonically increasing sequence, this implies that there exists some $N\in\mathbb{N}$ such that $\Delta_n\geq 1$ for all $n\geq N$. Hence, due to the positivity property proven above, this implies that there also exists some $N'\in\mathbb{N}$ such that $c_n\geq 0$ for all $n\geq N'$. Thus, by the extended monoticity property, eq.~\eqref{eq:cProofDEqu}, this also implies that $\Delta_{n+1}=\Delta_n$ for all $n\geq N'$, which is a contradiction to the assumption of $\Delta_n$ to not be bounded.
	
	\textsc{Convergence.} Since $\Delta_n$ is monotonically increasing and bounded, it converges to some constant $K$. Because $c_1,\dots,c_{m+1}\in\Z$ we also have $c_n\in\Z$ and thus $\Delta_n\in\Z$ for all $n$, such that $K\in\Z$ and, together with monotonicity, $\Delta_n = \Delta_N \equiv K$ for all $n\geq N$ and some $N\in\mathbb{N}$. This implies that $\Delta_N \geq 0$. To see why, assume the opposite. Since $\Delta_n\in\Z$, this means we assume $\Delta_n\leq-1$. By the negativity property, this would imply that $c_n < 0$ for all $n\geq N'$ and some $N'\in\mathbb{N}$. However, by the extended monotonicity property, we also have $c_n\geq 0$ for all $n\geq N$, which is a contradiction.
	
	\textsc{Summary.} Taking all this together, we see from the convergence property that there exists a $N\in\mathbb{N}$ such that $\Delta_n = \Delta_N \geq 0$ and thus $c_n\geq 0$ for all $n\geq N$ by the extended monotonicity property. Furthermore, since $c_n = \Delta_n + c_{n-m}$, we also see that $c_n \geq c_{n-m} \geq 0$ for all $n\geq N+m$.
\end{proof}

Consider now the consequence of this lemma on the sequence $\gamma_j$. Rewriting it in terms of the $c^i_j$, we get
\begin{equation}
	\gamma_j =  1\cltrad y_{1;j} \cltrad y_{1;j}\left(y_{1;j-m}\right)^{-1} = \prod_{i=1}^{M} z_i^{\min\left(0,c^i_j + \min\left(0,-c^i_{j-m}\right)\right)}\,.
\end{equation}
From the previous lemma, we know that for some $N\in\mathbb{N}$ we have $c_n \geq c_{n-m} \geq 0$ for all $n\geq N$ and thus for $j \geq N$
\begin{equation}
	\min\left(0,c^i_j + \min\left(0,-c^i_{j-m}\right)\right) = \min\left(0,c^i_j-c^i_{j-m}\right) = 0\,,
\end{equation}
proving that $\gamma_j = 1$ for $j\geq N$. 
 
With this property proven, we have established all parts and, together with the arguments in the main text, obtain the most general solution to the recurrence~\eqref{eq:Am1AlphaRec} as
\begin{equation}
	\label{eq:Am1SolutionSource}
	\alpha_j = \left[c_0^+ + c_1^+\eta_m^j + \cdots + c_{m-1}^+\eta_m^{(m-1)j}\right]\left(\beta_+\right)^{j/m} + \left[c_0^- + c_1^-\eta_m^j + \cdots + c_{m-1}^-\eta_m^{(m-1)j}\right]\left(\beta_-\right)^{j/m}\,,
\end{equation}
which we have repeated here for the further discussion of the associated non-rational letters. We denote the overall coefficients of $\beta_\pm$ in this equation as
\begin{equation}
	C_\pm(j) = c_0^\pm + c_1^\pm\eta_m^j + \cdots + c_{m-1}^\pm\eta_m^{(m-1)j}\,.
\end{equation}
As discussed before, due to the peridocity of the m-th root of unity, these satisfy $C_\pm(j+m) = C_\pm(j)$. Since in the limit $j\to\infty$, the term of $\beta_+$ dominates that of $\beta_-$, we associate the $m$ quantities
\begin{equation}
	\label{eq:Am1PhiDef}
	\phi_i = \frac{C_+(i)}{C_-(i)}
\end{equation}
for $0 \leq i \leq m-1$ to this sequence. These coefficients can be obtained in terms of the variables of the initial cluster by the initial conditions of the sequence $\alpha_j$. They are, again for $0 \leq i \leq m-1$, given by
\begin{equation}
	\label{eq:Am1InitialSource}
	\alpha_i = a_{i+1;0}\,,\quad \alpha_{m+i} = a_{i+1;0}\cdot \gamma_0^{-1}\beta_0F_i\,,
\end{equation}
whereas the $F_i$ are rational functions in the initial $\X$-variables. In order to express them in a convenient way, we define analogs for $K_1$ and $K_2$ by
\begin{equation}
	K_i = x_{1;0}\cdots x_{i;0}\left(1+x_{i+1;0}\left(1+\cdots x_{m;0}\left(1+x_{m+1;0}\right)\right)\right)\,.
\end{equation}
In terms of these expressions, the functions $F_i$ are given by
\begin{equation}
	F_0 = 1\,,\quad F_i = K_1 - K_{i+1}\,.
\end{equation}
These initial conditions together with the general solution, eq.~\eqref{eq:Am1SolutionSource}, form a system of two linear equations for the coefficients $C_\pm(i)$ for any $0\leq i\leq m-1$. It is solved by
\begin{equation}
	C_\pm(i) = a_{i+1;0}\left(\beta_\pm\right)^{-i/m}\frac{\pm 2F_i \mp K_1 + \sqrt{K_1^2-4K_2}}{2\sqrt{K_1^2-4K_2}}\,.
\end{equation}
Using the definition~\eqref{eq:Am1PhiDef} this proves eq.~\eqref{eq:Am1Phi} for the non-rational expressions associated to this sequence.

It remains to prove the initial conditions, eqs.~\eqref{eq:Am1InitialSource}. While the first of these can be seen by noting that $\alpha_i = a_{1;i}$ and using eq.~\eqref{eq:Am1ARels}, to prove the second condition -- and to determine the $F_i$ -- we first observe that $\alpha_{m+i} = \gamma_i^{-1}a_{1;m+i}$ and hence again by eq.~\eqref{eq:Am1ARels} that $\alpha_{m+i} = \gamma_i^{-1}a_{m+1;i}$. For $i=0$, this immediately implies that $\alpha_m = \gamma_0^{-1}a_{m+1;0}$ and hence $F_0 = 1$. For $i\geq 1$, we can use the mutation rule~\eqref{eq:Am1ASequence} and eq.~\eqref{eq:Am1GammaMutation} to get
\begin{equation}
	\label{eq:initCondIntermediate}
	\alpha_{m+i} = a_{2;i-1}\gamma_{i-1}^{-1}\beta_{i-1}\left(1+x_{1;i-1}\right) = a_{i+1;0}\gamma_0^{-1}\beta_0 \cdot \prod_{j=1}^{i} \left(1+x_{1;i-j}\right)\,,
\end{equation}
whereas we have repeatedly applied the relation $\gamma_{j+1}^{-1}\beta_{j+1} = (1+x_{1;j})\gamma_j^{-1}\beta_j$ in the last step, proving the second initial condition with $F_i$ being the product over the $\X$-variables. To express $F_i$ in terms of the variables of the initial cluster, we observe that by the mutation rules, eqs.~\eqref{eq:Am1XSequence} and~\eqref{eq:Am1XRels}, we have
\begin{align}
	1+x_{1;i-1} &= \left(1+x_{1;i-2}\right)^{-1}\left(1+x_{1;i-2}\left(1+x_{2;i-2}\right)\right) \\
	&= \left(1+x_{1;i-2}\right)^{-1}\left(1+x_{1;i-2}\left(1+x_{i;0}\right)\right) \\
	&= \left(1+x_{1;i-2}\right)^{-1}\left(1+x_{1;i-3}\right)^{-1}\left(1+x_{1;i-3}\left(1+x_{i-1;0}\left(1+x_{i;0}\right)\right)\right) \\
	&= \dots \\
	&= \prod_{j=2}^{i} \left(1+x_{1;i-j}\right)^{-1} \cdot \left(1 + x_{1;0}\left(1+x_{2;0}\left(1+\cdots x_{i-1;0}\left(1+x_{i;0}\right)\right)\right)\right)\,.
\end{align}
Using this relation in eq.~\eqref{eq:initCondIntermediate} and noting that $\left(1 + x_{1;0}\left(1+\cdots x_{i-1;0}\left(1+x_{i;0}\right)\right)\right) = K_1 - K_{i+1}$ completes the proof.

\subsection{The sink direction}
In the previous section and the main text, we have analysed the infinite mutation sequence obtained by repeatedly mutating $a_{1;j}$ in the $\operatorname{A}^{(1)}_m$ cluster algebra. However, this corresponds to only one of the two possible directions. As discussed before, the mutation of the sink-variable $a_{m+1;j+1}$ is the inverse to the mutation of $a_{1;j}$ and thus takes us from cluster $j+1$ to $j$ along the sequence, ie. the opposite direction. We now discuss its solution.

First of all, note that using the mutation relations~\eqref{eq:Am1ARels} we can rephrase the linearised recursion relation~\eqref{eq:Am1LinearARec} in terms of the sink variable as
\begin{equation}
	\label{eq:Am1BackwardsRel}
	\gamma_{j+m}^{-1}\gamma_{j+2m}^{-1}a_{m+1;j+2m} - \gamma_0^{-1}\beta_0 K_1\cdot \gamma_{j+m}^{-1}a_{m+1;j+m} + \gamma_0^{-2}\beta_0^2K_2\cdot a_{m+1;j} = 0\,.
\end{equation}
Since all relations required to arrive at this equation are valid for all $j\in\mathbb{N}$, so is this recurrence. In theory, we could now go on and apply the same techniques as for the source direction to solve this. However, we now have to consider the limit $j\to-\infty$, since the sink direction takes cluster $j+1$ to $j$. Accordingly, we define the new variable $\tilde{\alpha}_j$ via
\begin{equation}
	\label{eq:Am1AlphaTilde}
	\tilde{\alpha}_j = \gamma_{-(j\,\text{mod}\,m)}\gamma_{-(j\,\text{mod}\,m) - m}\cdots\gamma_{-j+2m}\gamma_{-j+m}\cdot a_{m+1;-j}\,,
\end{equation}
such that for this variable, the limit $j\to \infty$ is the correct one to consider. In terms of this sequence, the recurrence~\eqref{eq:Am1AlphaRec} can be expressed as
\begin{equation}
	\label{eq:Am1AlphaTildeRec}
	\gamma_0^{-2}\beta_0^2K_2\cdot \tilde{\alpha}_{j+2m} - \gamma_0^{-1}\beta_0 K_1\cdot \tilde{\alpha}_{j+m} + \tilde{\alpha}_j = 0\,,
\end{equation}
with the initial values $\tilde{\alpha}_0,\dots,\tilde{\alpha}_{2m-1}$. 

Before we obtain the solution of this recurrence via its characteristic polynomial, let us first discuss $\gamma_{-j}$ and its limit as $j$ goes to infinity. This sequence is again governed by eq.~\eqref{eq:cRec}, and is given in terms of the variable $d_n = c_{-n}$, which describes $\gamma_{-j}$, by
\begin{equation}
	d_{n-1} + d_{n+m} = \left[d_{n}\right]_+ +  \left[d_{n+m-1}\right]_+\,.
\end{equation}
Since this holds for all $n$, we may shift the index by $n' = n + m -1$, which, due to the symmetry, reduces this equation to the original form of eq.~\eqref{eq:cRec}. Hence, we may apply lemma~\ref{lem:CProps} to this case as well, proving that $\gamma_{-j} = 1$ for $j\geq J$ for some $J\in\mathbb{N}$.

The characteristic polynomial of the recurrence~\eqref{eq:Am1AlphaTildeRec} is given by
\begin{equation}
	\tilde{P}_m(t) = t^{2m} - \gamma_0\beta_0^{-1}\frac{K_1}{K_2}\cdot t^m + \gamma_0^{2}\beta_0^{-2}K_2^{-1}\,.
\end{equation}
Its roots are given by $\tilde{\beta}_\pm^{1/m}\eta_m^{i}$ for $i=0,\dots,m-1$, with $\eta_m$ again being the $m$-th root of unity and whereas
\begin{equation}
	\tilde{\beta}_\pm = \gamma_0\frac{K_1\pm\sqrt{K_1^2-4K_2}}{2\beta_0K_2}\,.
\end{equation}
Note that we have $\tilde{\beta}_\pm = \beta_\mp^{-1}$. Similar to before, we may use the roots to write down the most general solution for $\tilde{\alpha}_j$ and thus get
\begin{equation}
	\label{eq:Am1SolutionSink}
	\tilde{\alpha}_j = \left[\tilde{c}_0^+ + \tilde{c}_1^+\eta_m^j + \cdots + \tilde{c}_{m-1}^+\eta_m^{(m-1)j}\right]\left(\beta_-\right)^{-\frac{j}{m}} + \left[\tilde{c}_0^- + \tilde{c}_1^-\eta_m^j + \cdots + \tilde{c}_{m-1}^-\eta_m^{(m-1)j}\right]\left(\beta_+\right)^{-\frac{j}{m}}\,,
\end{equation}
whereas we have expressed this in terms of $\beta_\pm$ but have labelled the coefficients $\tilde{c}^\pm_i$ in terms of $\tilde{\beta}_\pm$. The same analysis as before applies to the overall coefficients $\tilde{C}_\pm(j)$, which again can be obtained from the initial conditions. For this direction they are given by
\begin{equation}
	\label{eq:Am1InitialSink}
	\tilde{\alpha}_i = a_{m+1-i;0}\,,\quad \tilde{\alpha}_{m+i} = a_{m+1-i;0}\cdot\gamma_0\beta_0^{-1}\tilde{F}_i\,,
\end{equation}
whereas the $\tilde{F}_i$ are now given by
\begin{equation}
	\tilde{F}_i = K_{m+1-i}/K_{m+1}\,.
\end{equation}
With these initial conditions, we again obtain a system of two linear equations from the general solution, eq.~\eqref{eq:Am1SolutionSink}, which we can solve in terms of the $\tilde{C}_\pm(i)$, resulting in
\begin{equation}
	\tilde{C}_\pm(i) = a_{m+1-i} \left(\beta_\mp\right)^{i/m} \frac{\pm 2K_2F_i \mp K_1 + \sqrt{K_1^2-4K_2}}{2\sqrt{K_1^2-4K_2}}\,.
\end{equation}
Using the definition $\tilde{\phi} = \tilde{C}_+/\tilde{C}_-$ this proves eq.~\eqref{eq:Am1PhiTilde} for the non-rational expressions associated to this sequence.

Having obtained the general solution for the infinite mutation sequence of type $A_m^{(1)}$ in the sink direction, it remains to prove the initial conditions, eqs.~\eqref{eq:Am1InitialSink}. The first follows directly from $\tilde{\alpha}_i = a_{m+1;-i} = a_{m+1-i;0}$, as can be seen from eq.~\eqref{eq:Am1ARels}. For the other condition, we first observe that $\tilde{\alpha}_{m+i} = \gamma_{-i}a_{m+1;-i-m}$ and hence by eq.~\eqref{eq:Am1ARels} that $\tilde{\alpha}_{m+i} = \gamma_{-i}\beta_{-i}^{-1}\cdot a_{m+1;-i} =\gamma_{-i}\beta_{-i}^{-1}\cdot a_{m+1-i;0}$. For $i=0$ we can immediately conclude that $\tilde{F}_i = 1$. For $i\geq 1$, by again using the relation $\gamma_{j+1}^{-1}\beta_{j+1} = (1+x_{1;j})\gamma_j^{-1}\beta_j$, we arrive at
\begin{align}
	\tilde{\alpha}_{m+i} &= a_{m+1-i;0}\gamma_0\beta_0^{-1} \left(1+x_{1;-i}\right)\cdots \left(1+x_{1;-1}\right) \\
	&= a_{m+1-i;0}\gamma_0\beta_0^{-1} \frac{\left(1+x_{m+1;1-i}\right)\cdots \left(1+x_{m+1;0}\right)}{x_{m+1;1-i}\cdots x_{m+1;0}}\,,
\end{align}
whereas we have used eq.~\eqref{eq:Am1XReverse} for the last step. This already proves eq.~\eqref{eq:Am1InitialSink}, with $\tilde{F}_i$ being the fraction of the $\X$-variables. To obtain an expression in terms of the variables of the initial cluster, we note that $x_{m+1;j-i} = x_{m+j-i;0}\left(1+x_{m+1;j-i+1}\right)$ such that we get
\begin{equation}
	\tilde{F}_i = \prod_{j=0}^{i-1}\left(x_{m+1-j;0}\right)^{-1} \cdot \left(1+x_{m+1;1-i}\right)\,.
\end{equation}
Similar to eq.~\eqref{eq:Am1XShift}, we obtain from the mutation relations, eq.~\eqref{eq:Am1XReverse}, that
\begin{equation}
	1+x_{m+1;1-i} = \prod_{j=i}^{m}\left(x_{m+1-j;0}\right)^{-1} \cdot K_{m+1-i}\,,
\end{equation}
such that we can conclude that
\begin{equation}
	\tilde{F}_i = \frac{K_{m+1-i}}{K_{m+1}}\,,
\end{equation}
completing our analysis of the infinite mutation sequences in cluster algebras of type $\operatorname{A}^{(1)}_m$.

\section{Full non-rational alphabet of eight-particle scattering}
\label{sec:EightParticleAlgebraicAlphabet}
In this section, for completeness, we present the entire non-rational alphabet of eight-particle scattering obtained from the perspective of scattering diagrams, see~\cite{Herderschee:2021dez} or section~\ref{sec:SDAlphabet}. We begin with the initial cluster of the cluster algebra of $\G{4,8}$, which is depicted in figure~\ref{fig:gr48Seed}. In there, we included our convention for the unfrozen variables $a^I_i$ of this cluster.

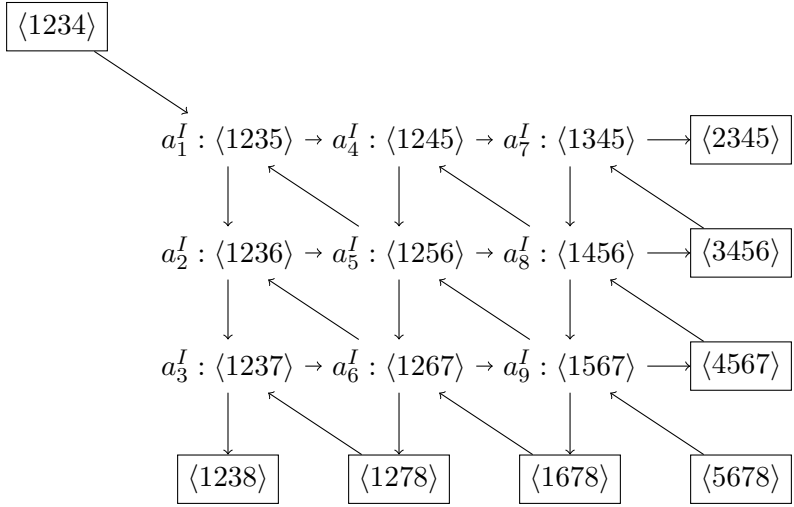
\begin{figure}[ht]
	\centering
	\begin{tikzpicture}[scale=1.5]
		\node at (0,0) (a1) {$a^I_1 : \pl{1235}$};
		\node at (1.5,0) (a2) {$a^I_4 : \pl{1245}$};
		\node at (3,0) (a3) {$a^I_7 : \pl{1345}$};
		\node at (4.5,0) [draw, rectangle] (a4) {$\pl{2345}$};
		
		\node at (0,-1) (b1) {$a^I_2 : \pl{1236}$};
		\node at (1.5,-1) (b2) {$a^I_5 : \pl{1256}$};
		\node at (3,-1) (b3) {$a^I_8 : \pl{1456}$};
		\node at (4.5,-1) [draw, rectangle] (b4) {$\pl{3456}$};
		
		\node at (0,-2) (c1) {$a^I_3 : \pl{1237}$};
		\node at (1.5,-2) (c2) {$a^I_6 : \pl{1267}$};
		\node at (3,-2) (c3) {$a^I_9 : \pl{1567}$};
		\node at (4.5,-2) [draw, rectangle] (c4) {$\pl{4567}$};
		
		\node at (0,-3) [draw, rectangle] (d1) {$\pl{1238}$};
		\node at (1.5,-3) [draw, rectangle] (d2) {$\pl{1278}$};
		\node at (3,-3) [draw, rectangle] (d3) {$\pl{1678}$};
		\node at (4.5,-3) [draw, rectangle] (d4) {$\pl{5678}$};
		
		\node at (-1.5,1) [draw, rectangle] (f) {$\pl{1234}$};
		
		\draw[->] (f) -- (a1);
		\draw[->] (a1) -- (b1);
		\draw[->] (b1) -- (c1);
		\draw[->] (c1) -- (d1);
		
		\draw[->] (a2) -- (b2);
		\draw[->] (b2) -- (c2);
		\draw[->] (c2) -- (d2);
		
		\draw[->] (a3) -- (b3);
		\draw[->] (b3) -- (c3);
		\draw[->] (c3) -- (d3);
		
		\draw[->] (a1) -- (a2);
		\draw[->] (a2) -- (a3);
		\draw[->] (a3) -- (a4);
		
		\draw[->] (b1) -- (b2);
		\draw[->] (b2) -- (b3);
		\draw[->] (b3) -- (b4);
		
		\draw[->] (c1) -- (c2);
		\draw[->] (c2) -- (c3);
		\draw[->] (c3) -- (c4);
		
		\draw[->] (b2) -- (a1);
		\draw[->] (c3) -- (b2);
		\draw[->] (d4) -- (c3);
		\draw[->] (b3) -- (a2);
		\draw[->] (b4) -- (a3);
		\draw[->] (c2) -- (b1);
		\draw[->] (c4) -- (b3);
		\draw[->] (d2) -- (c1);
		\draw[->] (d3) -- (c2);
		
	\end{tikzpicture}
	\caption{Initial seed of the cluster algebra of $\G{4,8}$.}
	\label{fig:gr48Seed}
\end{figure}

\pagebreak

We can use eqs.~\eqref{eq:coeffs} and~\eqref{eq:XVars} to immediately read off the $\X$-variables associated to each $\A$-variable. They are given by
\begin{align}
	x^I_1 &= \frac{\pl{1234}\pl{1256}}{\pl{1245}\pl{1236}}\,, & x^I_4 &= \frac{\pl{1235}\pl{1456}}{\pl{1345}\pl{1256}}\,, & x^I_7 &= \frac{\pl{1245}\pl{3456}}{\pl{2345}\pl{1456}}\,, \\
	x^I_2 &= \frac{\pl{1235}\pl{1267}}{\pl{1256}\pl{1237}}\,, & x^I_5 &= \frac{\pl{1236}\pl{1245}\pl{1567}}{\pl{1235}\pl{1456}\pl{1267}}\,, & x^I_8 &= \frac{\pl{1256}\pl{1345}\pl{4567}}{\pl{1245}\pl{3456}\pl{1567}}\,, \\
	x^I_3 &= \frac{\pl{1236}\pl{1278}}{\pl{1267}\pl{1238}}\,, & x^I_6 &= \frac{\pl{1237}\pl{1256}\pl{1678}}{\pl{1236}\pl{1567}\pl{1278}}\,, & x^I_9 &= \frac{\pl{1267}\pl{1456}\pl{5678}}{\pl{1256}\pl{4567}\pl{1678}}\,.
\end{align}

Mutating along the mutation sequence $\{1,2,4,1,6,8\}$, that is sequentially mutating the nodes with the corresponding index, we arrive at the $\operatorname{A}^{(1)}_1$ origin quiver depicted in figure~\ref{fig:A11OriginQuiver}. Note that since we only require the $\X$-coordinates for the computation of the non-rational alphabet, in the quiver we only show those, labelled by $x_i$ in the origin quiver, and omit the (frozen and unfrozen) $\A$-variables.
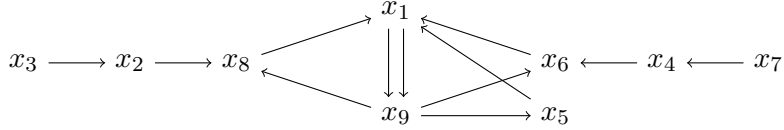
\begin{figure}[h]
	\centering
	\begin{tikzpicture}[scale=1.4]
		\node at (0,0.5) (6) {$x_3$};
		\node at (1,0.5) (9) {$x_2$};
		\node at (2,0.5) (2) {$x_8$};
		\node at (3.5,1) (5) {$x_1$};
		\node at (5,0.5) (1) {$x_6$};
		\node at (6,0.5) (7) {$x_4$};
		\node at (7,0.5) (8) {$x_7$};
		
		\node at (3.5,0) (3) {$x_9$};
		\node at (5,0) (4) {$x_5$};
		
		\draw[->] (6) -- (9);
		\draw[->] (9) -- (2);
		\draw[->] (2) -- (5);
		\draw[->] (1) -- (5);
		\draw[->] (7) -- (1);
		\draw[->] (8) -- (7);
		
		\draw[->] (3) -- (2);
		\draw[->] (3) -- (1);
		\draw[->] (4) -- (5);
		\draw[->] (3) -- (4);
		
		\draw[->] ([xshift=2pt]5.south) -- ([xshift=2pt]3.north);
		\draw[->] ([xshift=-2pt]5.south) -- ([xshift=-2pt]3.north);
	\end{tikzpicture}
	\caption{Principal part of the origin cluster in $\G{4,8}$ utilized to find the square-root letters.}
	\label{fig:A11OriginQuiver}
\end{figure}

Performing the sequence of mutations given above, we land in the origin quiver whose $\X$-variables $x_i$ are rational functions in the original variables $x^I_i$ of the initial cluster. These functions are given by

\begin{align}
	x_1 &= \frac{\left(1+ x^I_6 + x^I_1\left(1+x^I_4\right)\left(1+x^I_6\left(1+x^I_2\right)\right)\right)\left(1+ x^I_8 + x^I_1\left(1+x^I_2\right)\left(1+x^I_8\left(1+x^I_4\right)\right)\right)}{x^I_1x^I_2x^I_4}\,, \label{eq:originXVars1}\\
	x_9 &= \frac{\left(1+x^I_1\left(1+x^I_2\right)\left(1+x^I_4\right)\right)x^I_6x^I_8x^I_9}{\left(1+ x^I_6 + x^I_1\left(1+x^I_4\right)\left(1+x^I_6\left(1+x^I_2\right)\right)\right)\left(1+ x^I_8 + x^I_1\left(1+x^I_2\right)\left(1+x^I_8\left(1+x^I_4\right)\right)\right)}\,, \\
	x_5 &= \frac{x^I_1x^I_2x^I_4x^I_5}{1+x^I_1\left(1+x^I_2\right)\left(1+x^I_4\right)}\,, \\
	x_8 &= \frac{1+x^I_1\left(1+x^I_2\right)}{x^I_8\left(1+x^I_1\left(1+x^I_2\right)\left(1+x^I_4\right)\right)}\,,\\
	x_6 &= \frac{1+x^I_1\left(1+x^I_4\right)}{x^I_6\left(1+x^I_1\left(1+x^I_2\right)\left(1+x^I_4\right)\right)}\,,\\
	x_2 &= \frac{x^I_4x^I_8}{1+x^I_8+x^I_1\left(1+x^I_2\right)\left(1+x^I_8\left(1+x^I_4\right)\right)}\,, \\
	x_4 &= \frac{x^I_2x^I_6}{1+x^I_6+x^I_1\left(1+x^I_4\right)\left(1+x^I_6\left(1+x^I_2\right)\right)}\,, \\
	x_3 &= \frac{x^I_1x^I_2x^I_3}{1+x^I_1\left(1+x^I_2\right)}\,,\\
	x_7 &= \frac{x^I_1x^I_4x^I_7}{1+x^I_1\left(1+x^I_4\right)}\,.\label{eq:originXVars7}
\end{align}

As is discussed in section~\ref{sec:SDAlphabet} and~\cite{Herderschee:2021dez}, from this origin quiver we perform the limit of the infinite $\operatorname{A}^{(1)}_1$ mutation sequence. Working within the framework of scattering diagrams, we first construct the cone variables and take their limit, which is well-defined and finite, so that we land in an asymptotic chamber around the limit ray\footnote{This is one of the two tropical rays of $\pTPTC{4,8}$ that is not contained in the fan of the truncated cluster algebra, the other being $\left(0, 1, 0, 1, 0, -1, 0, -1, 0\right)$. Since the variables obtained from this limit ray can be obtained by a cyclic shift $\pl{ijkl}\to\pl{i+1\,j+1\,k+1\,l+1}$ we limit our analysis to the quantities around the first limit ray only.}
\begin{equation}
	\label{eq:limitRay}
	r_\infty = \left(1, -1, 0, -1, 0, 1, 0, 1, -1\right)\,.
\end{equation}
The limits of the cone variables along the sequence, ie. the cone variables of the asymptotic chamber, can be 
\begin{align}
	x_{\gamma_i}^0 &= x_i\,\quad\text{for}\quad i \in \{2,3,4,7\}\,, \label{eq:obviouslyRational} \\
	x_{\gamma_i}^0 &= \frac{x_i}{2}\left(1+x_1\left(1+x_9\right) + \sqrt{\Delta'}\right)\,\quad\text{for}\quad i \in \{5,6,8\}\,,\\
	x_{\gamma_1}^0 &= \frac{4x_1\Delta'}{\left(1+x_1-x_1x_9+\sqrt{\Delta'}\right)^2}\,,\quad 	x_{\gamma_9}^0 = \frac{x_9}{4}\left(1+\frac{1-x_1(1+x_9)}{\sqrt{\Delta'}}\right)^2\,, \label{eq:xG01and9}\\
	\Delta' &= \left(1+x_1(1+x_9)\right)^2 - 4x_1x_9\,.
\end{align}

The entire non-rational alphabet is obtained from the variables of \emph{all} asymptotic chambers around the limit ray. As is outlined in~\cite{Herderschee:2021dez}, a computer search yields a basis of $36$ multiplicatively independent polynomials. It consists of the $20$ polynomials given by

\begin{align}
	\tilde{f}_1&=x_{\gamma_1}^0\,, & \tilde{f}_2&=x_{\gamma_9}^0\,, & \tilde{f}_3&=1-x_{\gamma_1}^0x_{\gamma_9}^0\,, \label{eq:20LB1}\\
	\tilde{f}_4&=x_{\gamma_5}^0\,, & \tilde{f}_5&=1+x_{\gamma_5}^0\,, & \tilde{f}_6&=1+x_{\gamma_5}^0x_{\gamma_1}^0x_{\gamma_9}^0\,, \\
	\tilde{f}_7&=x_{\gamma_8}^0\,, & & \\
	\tilde{f}_8&=1+x_{\gamma_8}^0\,, & \tilde{f}_9&=1+x_{\gamma_2}^0\tilde{f}_8\,, & \tilde{f}_{10}&=1+x_{\gamma_3}^0\tilde{f}_9\,, \\ 
	\tilde{f}_{11}&=1+x_{\gamma_8}^0x_{\gamma_1}^0x_{\gamma_9}^0\,, & \tilde{f}_{12}&=1+x_{\gamma_2}^0\tilde{f}_{11}\,, & \tilde{f}_{13}&=1+x_{\gamma_3}^0\tilde{f}_{12}\,, \\
	\tilde{f}_{14}&=x_{\gamma_6}^0\,, & & \\
	\tilde{f}_{15}&=1+x_{\gamma_6}^0\,, & \tilde{f}_{16}&=1+x_{\gamma_4}^0\tilde{f}_{15}\,, & \tilde{f}_{17}&=1+x_{\gamma_7}^0\tilde{f}_{16}\,, \\ 
	\tilde{f}_{18}&=1+x_{\gamma_6}^0x_{\gamma_1}^0x_{\gamma_9}^0\,, & \tilde{f}_{19}&=1+x_{\gamma_4}^0\tilde{f}_{18}\,, & \tilde{f}_{20}&=1+x_{\gamma_7}^0\tilde{f}_{19}\,.\label{eq:20LB4} \\
	\intertext{as well as $16$ more polynomials given by}
	\tilde{f}_{21} &= x_{\gamma_2}^0\,, & \tilde{f}_{22} &= x_{\gamma_3}^0\,,	& \\
	\tilde{f}_{23} &= 1 + x_{\gamma_2}^0\,, & \tilde{f}_{24} &= 1+ x_{\gamma_3}^0\,,	& \\
	\tilde{f}_{25} &= 1 + x_{\gamma_3}^0\tilde{f}_{23}\,, & \tilde{f}_{26} &= 1 + x_{\gamma_2}^0\tilde{f}_{8}\tilde{f}_{11}\,, & \tilde{f}_{27} &= 1 + x_{\gamma_3}^0 \tilde{f}_{26}\,, \\
	\tilde{f}_{28} &= \mathrlap{1 + x_{\gamma_2}^0 \tilde{f}_{27} + x_{\gamma_3}^0\left(1 + x_{\gamma_2}^0\left(\tilde{f}_7 + \tilde{f}_{11}\right)\right)\,,} \\
	\tilde{f}_{29} &= x_{\gamma_4}^0\,, & \tilde{f}_{30} &= x_{\gamma_7}^0\,,	& \\
	\tilde{f}_{31} &= 1 + x_{\gamma_4}^0\,, & \tilde{f}_{32} &= 1+ x_{\gamma_7}^0\,,	& \\
	\tilde{f}_{33} &= 1 + x_{\gamma_7}^0\tilde{f}_{31}\,, & \tilde{f}_{34} &= 1 + x_{\gamma_4}^0\tilde{f}_{15}\tilde{f}_{18}\,, & \tilde{f}_{35} &= 1 + x_{\gamma_7}^0 \tilde{f}_{34}\,, \\
	\tilde{f}_{36} &= \mathrlap{1 + x_{\gamma_4}^0 \tilde{f}_{35} + x_{\gamma_7}^0\left(1 + x_{\gamma_4}^0\left(\tilde{f}_{14} + \tilde{f}_{18}\right)\right)\,.}
\end{align}

As can be seen from eq.~\eqref{eq:obviouslyRational} and eqs.~\eqref{eq:originXVars1}--\eqref{eq:originXVars7}, the variables $x_{\gamma_i}^0$ are rational for $i \in \{2,3,4,7\}$ and hence so are $10$ of the polynomials of the above basis. In fact, by parameterising the Plücker variables in terms of the web-parameterisation and evaluating the web-variables with prime values, it is easy to see that the polynomials $\tilde{f}_{21}$ to $\tilde{f}_{36}$ are rational, that is the square-roots cancel. Even more than that, these $16$ polynomials are actually contained in the $272$-letter rational alphabet of~\cite{Henke:2019hve,Drummond:2019cxm,Arkani-Hamed:2019rds} and can be expressed as
\begin{align}
	\tilde{f}_{21} &= \frac{R_{157}}{R_{163}}\,, & \tilde{f}_{22} &= \frac{R_9}{R_2}\,, & \tilde{f}_{23} &= \frac{R_2R_{143}}{R_{163}}\,, & \tilde{f}_{24} &= \frac{R_3}{R_2}\,, \\
	\tilde{f}_{25} &= \frac{R_{164}}{R_{163}}\,, & \tilde{f}_{26} &= \frac{R_2R_{194}}{R_9R_{157}}\,, & \tilde{f}_{27} &= \frac{R_{192}}{R_{157}}\,, & \tilde{f}_{28} &= \frac{R_2R_{146}}{R_{163}}\,, \\
	\tilde{f}_{29} &= \frac{R_{36}}{R_{42}}\,, & \tilde{f}_{30} &= \frac{R_{97}}{R_{10}}\,, & \tilde{f}_{31} &= \frac{R_{10}R_{34}}{R_{42}}\,, & \tilde{f}_{32} &= \frac{R_{43}}{R_{10}}\,, \\
	\tilde{f}_{33} &= \frac{R_{44}}{R_{42}}\,, & \tilde{f}_{34} &= \frac{R_{10}R_{93}}{R_{36}R_{97}}\,, & \tilde{f}_{35} &= \frac{R_{91}}{R_{36}}\,, & \tilde{f}_{36} &= \frac{R_{10}R_{89}}{R_{42}}\,,
\end{align}
whereas $R_i$ refers to the $i$-th rational letter in the alphabet provided in the attached file~\texttt{Gr48Alphabet.m}.

Since $16$ of these polynomials are contained in the $272$-letter rational alphabet, we are left with the $20$ letters given by $\tilde{f}_1$ to $\tilde{f}_{20}$. While these $20$ letters can be numerically checked to actually contain square-roots, we find $10$ multiplicative combinations that are contained in the rational alphabet. They are given by
\begin{align}
	\tilde{f}_1\tilde{f}_2\tilde{f}_4^2 &= \frac{R_{222}}{R_9R_{97}}\,,& \frac{\tilde{f}_7}{\tilde{f}_4} &= \frac{R_2R_{97}}{R_{157}}\,, & \frac{\tilde{f}_{14}}{\tilde{f}_4} &= \frac{R_9R_{10}}{R_{36}}\,,\\
	\tilde{f}_5 \tilde{f}_6 &= \frac{R_{4}R_{131}}{R_{9}R_{97}} \,, & \tilde{f}_8\tilde{f}_{11} &= \frac{R_{163}R_{197}}{R_{9}(R_{157})^2}\,, &	\tilde{f}_{15}\tilde{f}_{18} &= \frac{R_{42}R_{94}}{(R_{36})^2R_{97}}\,,\\
	\tilde{f}_{9}\tilde{f}_{12} &= \frac{(R_{2})^2R_{147}}{R_{9}R_{163}}\,, & \tilde{f}_{16}\tilde{f}_{19} &= \frac{(R_{10})^2R_{90}}{R_{42}R_{97}}\,,\\
	\tilde{f}_{10}\tilde{f}_{13} &= \frac{R_{196}}{R_{163}}\,, & \tilde{f}_{17}\tilde{f}_{20} &= \frac{R_{96}}{R_{42}}\,, 
\end{align}
whereas again $R_i$ corresponds to the $i$-th rational letter in the rational alphabet. Using these $10$ relations, we can further reduce the square-root letters to the basis of $10$ multiplicatively independent letters given by
\begin{align}
	f_1 &= \left(x_{\gamma_1}^0\right)^{-1}\left(1-x_{\gamma_1}^0x_{\gamma_9}^0\right)^2\,, &  f_2 &=x_{\gamma_9}^0\left(1-x_{\gamma_1}^0x_{\gamma_9}^0\right)^2 \,,& f_3 &= \frac{1+x_{\gamma_5}^0x_{\gamma_1}^0x_{\gamma_9}^0}{1+x_{\gamma_5}^0}\,, \nonumber\\
	f_4 &= \frac{1+x_{\gamma_8}^0x_{\gamma_1}^0x_{\gamma_9}^0}{1+x_{\gamma_8}^0}\,, & f_5 &= \frac{1+x_{\gamma_2}^0\left(1+x_{\gamma_8}^0x_{\gamma_1}^0x_{\gamma_9}^0\right)}{1+x_{\gamma_2}^0\left(1+x_{\gamma_8}^0\right)}\,, \nonumber\\
	f_6 &= \mathrlap{\frac{1+x_{\gamma_3}\left(1+x_{\gamma_2}^0\left(1+x_{\gamma_8}^0x_{\gamma_1}^0x_{\gamma_9}^0\right)\right)}{1+x_{\gamma_3}\left(1+x_{\gamma_2}^0\left(1+x_{\gamma_8}^0\right)\right)}\,,} \label{eq:full10SDbasis} \\
 	f_7 &= \frac{1+x_{\gamma_6}^0x_{\gamma_1}^0x_{\gamma_9}^0}{1+x_{\gamma_6}^0}\,, & f_8 &= \frac{1+x_{\gamma_4}^0\left(1+x_{\gamma_6}^0x_{\gamma_1}^0x_{\gamma_9}^0\right)}{1+x_{\gamma_4}^0\left(1+x_{\gamma_6}^0\right)}\,, \nonumber\\
	f_{9} &= \mathrlap{\frac{1+x_{\gamma_7}\left(1+x_{\gamma_4}^0\left(1+x_{\gamma_6}^0x_{\gamma_1}^0x_{\gamma_9}^0\right)\right)}{1+x_{\gamma_7}\left(1+x_{\gamma_4}^0\left(1+x_{\gamma_6}^0\right)\right)}\,,} \nonumber \\
	f_{10} &= x_{\gamma_5}^0\left(1-x_{\gamma_1}^0x_{\gamma_9}^0\right)\,. \nonumber
\end{align}

It can be easily demonstrated that the set of $112$ $\operatorname{A}^{(1)}_1$-letters obtained in section~\ref{sec:A11Application} for the limit ray $r_\infty$, eq.~\eqref{eq:limitRay}, is equivalently described by the basis of the 9 multiplicatively independent square-root letters given by $f_1$ to $f_{9}$ of eqs.~\eqref{eq:full10SDbasis}.

In summary, we see that the non-rational alphabet obtained from the scattering diagram adds one further letter $f_{10} = x_{\gamma_5}^0\left(1-x_{\gamma_1}^0x_{\gamma_9}^0\right)$ per limit ray compared to the previously known 9 letters, see sec.~\ref{sec:A11Application} or \cite{Drummond:2019cxm,Zhang:2019vnm}. Using eqs.~\eqref{eq:obviouslyRational} and~\eqref{eq:xG01and9}, we can simplify this letter and find that
\begin{equation}
	f_{10} = x_5 \sqrt{\Delta'}\,.
\end{equation}
With $x_5$ being one of the $\X$-variables of the origin quiver, this already demonstrates that $f_1$ corresponds to the square-root up to a monomial in the rational alphabet. In fact, this square-root is proportional to the square-root $\Delta_{1,3,5,7}$ of one of the two eight-particle four-mass boxes. In terms of the four-mass box, we can write the additional letter as
\begin{equation}
	f_{10} = \frac{\pl{1256}\pl{3478}}{\pl{1278}\pl{3456}}\sqrt{\Delta_{1,3,5,7}}\,,
\end{equation}
whereas we have
\begin{align}
	\Delta_{1,3,5,7} =& \left(1 - \frac{\pl{1234}\pl{5678}}{\pl{1256}\pl{3478}} - \frac{\pl{1278}\pl{3456}}{\pl{1256}\pl{3478}}\right)^2 - 4 \frac{\pl{1278}\pl{1234}\pl{3456}\pl{5678}}{\left(\pl{1256}\pl{3478}\right)^2}\,.
\end{align}
The square-root $\Delta_{2,4,6,8}$ of the other eight-particle four-mass box appears in a similar way in the non-rational alphabet of the other limit ray, which is obtained by the cyclic shift $i \to i + 1$ on the indices of the Plücker variables.

\section{Web-parameterisation of $\G{4,n}$}
\label{sec:WebParam}
As we have reviewed in section~\ref{sec:BackgroundGrConf}, the totally positive configuration space $\PConf{k,n}$ can be constructed as the space of all Plücker variables, restricted to non-negative values, up to some scalings. The space has dimension $d= (k-1)(n-k-1)$ and can be parameterised in terms of $d$ independent parameters by the web-parameterisation~\cite{Speyer2005}. In this appendix we briefly discuss how this parameterisation is constructed and present it explicitly for $\PConf{4,8}$ and $\PConf{4,9}$.

Instead of parameterising the individual Plücker variables, it is more convenient to instead obtain a parameterisation of the momentum twistors instead, see e.g.~\cite{Drummond:2019cxm}. The Plücker variables can in turn be obtained as the minors of the $k\times n$ matrix $Z$ whose columns are the momentum twistors. When parameterised, this matrix is of the form
\begin{equation}
	Z = \left(\mathbf{1}_k | M\right)\,
\end{equation}
with the $k\times k$ identity matrix $\mathbf{1}_k$. The entries $m_{ij}$ of the $k\times (n-k)$ matrix $M$ are given by
\begin{equation}
	m_{ij} = \left(-1\right)^i\sum_{\underline{\lambda}\in Y_{ij}}\prod_{m=1}^{k-i}\prod_{l=1}^{\lambda_m}x_{ml}\,,
\end{equation}
whereas $Y_{ij}$ denotes the multi-dimensional range $0 \leq \lambda_{k-i} \leq \dots \leq \lambda_1 \leq j-1$ and the $x_{ml}$ are an alternative labeling of the web-parameters.

The explicit web-parameterisation of $\PConf{4,8}$ and $\PConf{4,9}$ can be found in the ancillary file \texttt{WebParameterisation.m} attached to the \texttt{arXiv} submission of this article. For $n=8$, the columns $\underline{m}_1,\underline{m}_2,\underline{m}_3$ of the $4\times 4$ matrix $M$ are given by
\begin{equation}
	\underline{m}_1 = \begin{pmatrix}
		1 \\ -1 \\ 1 \\ -1
	\end{pmatrix}\,,\quad 	
	\underline{m}_2 =  
	\begin{pmatrix}
	- 1 - x_1\left(1+x_2\left(1+x_3\right)\right) \\
	1 + x_1\left(1+x_2\right) \\
	-1 - x_1 \\
	1
	\end{pmatrix}\,,
\end{equation}
\begin{equation}
	\underline{m}_3 =  
	\begin{pmatrix}
		-1-x_1\left(1+x_4+x_2\left(1+x_4\left(1+x_5\right)\right)+x_3\left(1+x_4\left(1+x_5\left(1+x_6\right)\right)\right)\right) \\
		1+x_1\left(1+x_4+x_2\left(1+x_4\left(1+x_5\right)\right)\right)\\
		-1-x_1\left(1+x_4\right) \\
		1
	\end{pmatrix}\,,
\end{equation}
and the components of the fourth column $\underline{m}_4$ are given by
\begin{align}
	m_{41} =& -1 - x_1 - x_1 x_2 - x_1 x_2 x_3 - x_1 x_4 - x_1 x_2 x_4 - x_1 x_2 x_3 x_4 - x_1 x_2 x_4 x_5 - x_1 x_2 x_3 x_4 x_5 \nonumber \\
	&- x_1 x_2 x_3 x_4 x_5 x_6 - x_1 x_4 x_7 - x_1 x_2 x_4 x_7 - x_1 x_2 x_3 x_4 x_7 - x_1 x_2 x_4 x_5 x_7 \nonumber \\
	&- x_1 x_2 x_3 x_4 x_5 x_7 - x_1 x_2 x_3 x_4 x_5 x_6 x_7 - x_1 x_2 x_4 x_5 x_7 x_8 - x_1 x_2 x_3 x_4 x_5 x_7 x_8 \nonumber \\ 
	&- x_1 x_2 x_3 x_4 x_5 x_6 x_7 x_8 - x_1 x_2 x_3 x_4 x_5 x_6 x_7 x_8 x_9\,,
\\	
	m_{42} =&\, 1 + x_1 + x_1 x_2 + x_1 x_4 + x_1 x_2 x_4 + x_1 x_2 x_4 x_5 + x_1 x_4 x_7 + x_1 x_2 x_4 x_7 + x_1 x_2 x_4 x_5 x_7 \nonumber \\
	&+ x_1 x_2 x_4 x_5 x_7 x_8\,,
\\	
	m_{43} =& -1 - x_1 - x_1 x_4 - x_1 x_4 x_7\,,
\\	
	m_{44} =&\, 1\,.
\end{align}
The first four columns of the parameterised $4\times 5$ matrix $M$ for $n=9$ are the same as that for $n=8$. The components of the last column $\underline{m}_5$ are given by
\begin{align}
	m_{51} =&\, m_{41} - x_1 x_4 x_7 x_{10} - x_1 x_2 x_4 x_7 x_{10} - x_1 x_2 x_3 x_4 x_7 x_{10} - x_1 x_2 x_4 x_5 x_7 x_{10} \nonumber \\
	&- x_1 x_2 x_3 x_4 x_5 x_7 x_{10} - x_1 x_2 x_3 x_4 x_5 x_6 x_7 x_{10} - x_1 x_2 x_4 x_5 x_7 x_8 x_{10} \nonumber \\
	&- x_1 x_2 x_3 x_4 x_5 x_7 x_8 x_{10} - x_1 x_2 x_3 x_4 x_5 x_6 x_7 x_8 x_{10} - x_1 x_2 x_3 x_4 x_5 x_6 x_7 x_8 x_9 x_{10} \nonumber \\
	&- x_1 x_2 x_4 x_5 x_7 x_8 x_{10} x_{11} - x_1 x_2 x_3 x_4 x_5 x_7 x_8 x_{10} x_{11} - x_1 x_2 x_3 x_4 x_5 x_6 x_7 x_8 x_{10} x_{11} \nonumber \\
	&- x_1 x_2 x_3 x_4 x_5 x_6 x_7 x_8 x_9 x_{10} x_{11} - x_1 x_2 x_3 x_4 x_5 x_6 x_7 x_8 x_9 x_{10} x_{11} x_{12}\,,
\\
	m_{52} =&\, m_{42} + x_1 x_4 x_7 x_{10} + x_1 x_2 x_4 x_7 x_{10} + x_1 x_2 x_4 x_5 x_7 x_{10} + x_1 x_2 x_4 x_5 x_7 x_8 x_{10} \nonumber \\
	&+ x_1 x_2 x_4 x_5 x_7 x_8 x_{10} x_{11}\,,
\\
	m_{53} =&\, m_{43} - x_1 x_4 x_7 x_{10}\,,
\\
	m_{54} =&\, 1\,.
\end{align}

\bibliographystyle{SupportingFiles/JHEP}
\bibliography{SupportingFiles/references}

\end{document}